\documentclass[11pt]{article}
\usepackage{tocloft}
\usepackage{etoolbox}
\makeatletter
\patchcmd{\l@section}{\hfil}{\hfil\parfillskip=0pt \relax \raggedright}{}{}
\makeatother

 \usepackage{basket_Fourier}
\usepackage{float}      
\usepackage[section]{placeins} 
\usepackage{afterpage} 
\usepackage{comment}
\usepackage{algpseudocode}

\makeatletter
\def\BState{\State\hskip-\ALG@thistlm}
\makeatother

\usepackage[thinlines]{easytable}
\renewcommand{\arraystretch}{2.5}
\usepackage{array}
\newcolumntype{P}[1]{>{\raggedright\arraybackslash}p{#1}}
\newcolumntype{Y}{>{\raggedright\arraybackslash}X}
\usepackage{diagbox} 
\usepackage{multirow} 
\usepackage{calc} 
\newcommand{\diagonalentry}[2]{%
	\settowidth{\dimen0}{$#1$}%
	\settowidth{\dimen2}{$#2$}%
	\ifdim\dimen0>\dimen2
	\raisebox{0.5\dimexpr\dimen0-\dimen2}{#1}%
	\else
	\raisebox{-0.5\dimexpr\dimen2-\dimen0}{#2}%
	\fi
}

\title{Quasi-Monte Carlo with Domain Transformation for Efficient Fourier Pricing of Multi-Asset Options}

\author[1]{Christian Bayer}
\author[2]{Chiheb Ben Hammouda}
\author[3]{Antonis Papapantoleon}
\author[4]{Michael Samet\thanks{samet@uq.rwth-aachen.de}}
\author[4,5]{Ra\'ul Tempone}
\affil[1]{Weierstrass Institute for Applied Analysis and Stochastics (WIAS), Berlin, Germany.}
\affil[2]{Mathematical Institute, Utrecht University, Utrecht, the Netherlands}
\affil[3]{Delft Institute of Applied Mathematics, EEMCS, TU Delft, 2628 Delft, The Netherlands, and Department of Mathematics, SAMPS, NTUA, 15780 Athens, Greece and Institute of Applied and Computational Mathematics, FORTH, 70013 Heraklion, Greece.}
\affil[4]{Chair of Mathematics for Uncertainty Quantification, RWTH Aachen University, Aachen, Germany.}
\affil[5]{King Abdullah University of Science and Technology (KAUST), Computer, Electrical and Mathematical Sciences \& Engineering Division (CEMSE), Thuwal, Saudi Arabia.}

\begin{document}
	\date{}
\maketitle

\begin{abstract}
Efficiently pricing multi-asset options poses a significant challenge in quantitative finance. Fourier methods leverage the regularity properties of the integrand in the Fourier domain to accurately and rapidly value options that typically lack regularity in the physical domain.  However, most of the existing Fourier approaches face hurdles in high-dimensional settings due to the tensor product (TP) structure of the commonly employed numerical quadrature techniques. To overcome this difficulty,  this work advocates using the randomized quasi-Monte Carlo (RQMC) quadrature to improve the scalability of Fourier methods in high dimensions. The RQMC technique benefits from the smoothness of the integrand and alleviates the curse of dimensionality while providing practical error estimates. Nonetheless, the applicability of RQMC on the unbounded domain, $\mathbb{R}^d$, requires a domain transformation to $[0,1]^d$, which may result in singularities of the transformed integrand at the corners of the hypercube, and hence deteriorate the performance of RQMC. To circumvent this difficulty, we design an efficient domain transformation procedure based on boundary growth conditions on the transformed integrand. The proposed transformation preserves sufficient regularity of the original integrand for fast convergence of the RQMC method. To validate our analysis, we demonstrate the efficiency of employing RQMC with an appropriate transformation to evaluate options in the Fourier space for various pricing models, payoffs, and dimensions. Finally, we highlight the computational advantage of applying RQMC over TP quadrature in the Fourier domain, and over MC in the physical domain for options with up to 15 assets. \\
 

\end{abstract}

\tableofcontents

\section{Introduction}

Computing the price of options depending on multiple assets presents a long-standing challenge in quantitative finance. Option pricing approaches can be broadly divided into two classes. The first class consists of partial differential equation  (PDE) approaches, which are typically solved using finite difference, finite volume or finite element methods \cite{duffy2013finite,wang2004novel,hilber2013computational}. Despite the wide literature on the PDE approaches, they tend to suffer from the curse of dimensionality \cite{reisinger2007efficient} and, hence, are not suitable for multi-asset option pricing. Recent studies have employed deep learning techniques to address the curse of dimensionality \cite{han2018solving,glau2022deep,georgoulis2026deep,gnoatto2022deep}; although the developed methods are efficient, they often rely on computationally intensive offline training procedure.  The second class of methods is based on the integral representation of option prices, which are given as expected values  \cite{longstaff1995option}. For multi-asset derivatives,  an advantage of computing these expectations using the Monte Carlo (MC) method  is that  it yields a dimension-independent convergence rate, $\mathcal{O}(N^{-\frac{1}{2}})$ \cite{glasserman2004monte}, with $N$ being the number of simulated paths. However, this rate of convergence is considered to be rather slow \cite{glasserman2004monte}, and a faster convergence rate can be achieved by applying efficient numerical integration techniques, such as the quasi-Monte Carlo (QMC) \cite{dick2013high}  or the (adaptive) sparse grid quadrature (SGQ)  \cite{holtz2010sparse} methods. The shortcoming of these quadrature methods is that their convergence rates are sensitive to the dimensionality and the regularity of the integrand. Moreover, since most of the option payoff functions are discontinuous or have discontinuous derivatives, off-the-shelf application of QMC or SGQ to such non-regular pricing problems leads to poor performance \cite{bayer2018smoothing,bayer2023numerical}.  Consequently, considerable research has been devoted to developing analytical and numerical smoothing techniques. For instance, analytical smoothing can be performed by applying a conditional expectation with respect to (w.r.t) selected integration variables, as in \cite{xiao2018conditional,bayer2018smoothing,bayer2020hierarchical}.  Moreover, when analytical smoothing is inapplicable, numerical smoothing can be carried out by combining root-finding methods with preintegration   \cite{bayer2023numerical}. A more implicit smoothing technique is to map the integration problem from the direct space to an image space via an integral transform (e.g., Fourier, Laplace, Mellin, or Hilbert) \cite{lewis2001simple,raible2000levy,feng2008pricing,panini2004option}, where the integrand is more regular, or possibly analytic in some region of the complex plane.  The integration problem can then be solved more efficiently in the image space if appropriate numerical methods are employed, and  a suitable contour of integration is chosen (for more details, see \cite{bayer2023optimal,lord2007optimal}). Several works \cite{ruijter2012two,kirkby2016frame,colldeforns2017two,hurd2010fourier,von2015benchop} demonstrated that when the characteristic function of the log-price is computable, Fourier methods may present a solid alternative to the MC method and PDE approaches for options with up to two underlying assets. However,  the literature on Fourier pricing in high dimensions is scarce. 

There are three widely used Fourier pricing approaches, which we briefly describe, but for a more elaborate discussion, refer to the introduction in \cite{bayer2023optimal}. The first approach \cite{carr1999option} takes the Fourier transform of the exponentially dampened option price w.r.t the log-strike variable and applies the fast Fourier transform (FFT) algorithm to evaluate options for multiple strikes simultaneously.  The second approach  is known as the COS method \cite{fang2008novel} and is based on expanding the density function of the log-price in a Fourier cosine series, and expressing the Fourier cosine coefficients in terms of the characteristic function. The third approach\footnote{This approach was concurrently interpreted in terms of the double Laplace transform by Raible \cite{raible2000levy} and the extended Fourier transform by Lewis \cite{lewis2001simple}. } \cite{lewis2001simple,raible2000levy} takes the inverse extended Fourier transform\footnote{The extended Fourier transform is sometimes referred to as the Fourier-Laplace transform \cite{duistermaat1996fourier} or the generalized Fourier transform \cite{titchmarsh1948introduction}.}  \cite{duistermaat1996fourier} of the payoff function and the density function separately in the log-price variable. 
The main advantage of the latter approach is that it can be straightforwardly extended to multiple dimensions  \cite{eberlein2010analysis}, given explicitly as a multivariate contour integral in the complex plane. 
 For instance,  \cite{hurd2010fourier} adopted this approach and proposed a numerical scheme based on the FFT for the pricing of two-dimensional (2D) spread options. While their method demonstrated efficiency in handling 2D options, due to its tensor product (TP) nature, the curse of dimensionality occurs for high-dimensional cases. The results in  \cite{bayer2023optimal} revealed that the curse of dimensionality can be alleviated by employing a dimension-adaptive SGQ method and by parametrically smoothing the integrand via an optimized choice of the contour of integration. Their numerical experiments show that basket and rainbow options with up to six underlying assets following Lévy dynamics can be priced very efficiently. 
 Other related work \cite{kastoryano2022highly}  uses the tensor-train cross algorithm to exploit the possibly low rank structure of the integrand in the Fourier domain. 
  Nevertheless, this Tensor-Fourier method was reported to become numerically unstable in high dimensions, and the numerical experiments were restricted to the geometric Brownian motion (GBM) model.  
   Thus, only a limited number of studies have offered efficient numerical methodologies for the pricing of multi-asset options in a Fourier representation.

   Despite the wide applicability of the QMC method in option pricing in the physical space (we refer to \cite{l2009quasi} for an overview), to the best of our knowledge, there have been no attempts to apply the QMC method in Fourier pricing. The critical aspect of applying QMC in the Fourier domain is that the target integrand is supported on $\mathbb{R}^d$, whereas QMC low discrepancy (LD) points are designed on $[0,1]^d$ \cite{dick2013high}. The standard technique to transform the integration problem from $\mathbb{R}^d$ to $[0,1]^d$ is by composing the original integrand with an inverse cumulative distribution function (ICDF) \cite{hartinger2004quasi, owen2006halton, nichols2014fast}. This transformation poses two main challenges. The first challenge is that if the ICDF and its parameters are not chosen carefully, it can lead to integrands which are unbounded at the boundary of $[0,1]^d$, and hence deteriorate the performance of  the QMC method \cite{owen2006halton}. The second challenge is  that the CDF of a multivariate distribution with dependent components is generally not invertible. In their study, \cite{kuo2006randomly} analyzed the optimal rate of convergence of the QMC method applied on $\mathbb{R}^d$ when computing expected values of a random variable (RV) for various combinations of parameters of a weighted function space and distributions of the RV, highlighting the need for appropriate choice of the weight space.  Moreover, it was proved in \cite{owen2006halton} that the convergence rate of the QMC method can be severely impacted by a singularity with polynomial growth at some corner of the hypercube, depending on the rate of growth conditions on the  integrand, and corner-avoidance properties of the LD sequences. More recently,  \cite{ouyang2024achieving} extended the work of  \cite{owen2006halton} and studied the impact of the boundary singularities on the rate of convergence of the RQMC method for unbounded integrands with exponential rate of growth.  They show that  RQMC combined with importance sampling (IS) can achieve an asymptotic convergence rate $\mathcal{O}(N^{- \frac{3}{2} + \epsilon_r})$, $\epsilon_r > 0$. A related work \cite{liu2023nonasymptotic} studies the nonasymptotic convergence rate of the QMC method aided with IS.  In contrast to the mentioned research, the integration problem in the Fourier space is deterministic. Consequently, IS is enrooted in the proposed approach, and the choice of the proposal IS density may have adverse effects on the convergence of the RQMC method. The contributions of this paper are as follows:

 \begin{itemize}
 	\item To the best of our knowledge, we are the first to propose the use of RQMC in the Fourier space for option pricing in high dimensions. We provide a practical model-specific domain transformation strategy that avoids introducing the singularity of the integrand near the boundaries of $[0,1]^d$. The key idea is to preserve the original features of the integrand by choosing a proposal domain transformation distribution that shares the same functional form as the asymptotically dominant part of the integrand, in particular, the extended characteristic function of the log-price.  Then, we tune the parameters of the proposal density to satisfy boundary growth conditions that ensure fast convergence of the RQMC method.
 	\item Compared to other related works \cite{nichols2014fast, ouyang2024achieving,liu2023nonasymptotic}, we do not treat weighted integration problems w.r.t the Gaussian density. We consider more challenging integrands that may have a slower decay, e.g., root-exponential decay in the case of the generalized hyperbolic (GH) model and power-law decay in the case of the variance gamma (VG) model. In addition, most research works \cite{nichols2014fast, ouyang2024achieving, liu2023nonasymptotic} assume that the components of the density they integrate against are independent. Consequently, they perform the domain transformation to the hypercube independently for each dimension via marginal ICDFs. In contrast, this work proposes a two-step domain transformation strategy that accommodates multivariate transformation distributions with dependent RVs. In the first step, we express the transformation density in terms of a normal variance-mean mixture form and eliminate dependencies between the components of the multivariate normal distribution using the Cholesky factorization. In the second step, we apply the ICDF mapping to the hypercube separately for the mixing distribution and the multivariate standard normal distribution.
 	\item We demonstrate the computational advantage of employing RQMC in the Fourier space compared to the TP-Laguerre quadrature or MC method in the Fourier space and compared to the MC method in the direct space. We provide several numerical experiments for various pricing models and options with up to 15 underlying assets.
 	
 	\item The methodology of mapping the problem to the Fourier space and applying RQMC is extendable to high-dimensional expectation problems, beyond the specific context of option pricing. Furthermore, the domain transformation procedure we propose represents a general approach to addressing deterministic integrals on $\mathbb{R}^d$ using RQMC methods.
 \end{itemize} 
 
The outline of this paper is as follows. Section \ref{sec:Problem Setting and Pricing Framework} introduces the problem setting of multivariate Fourier pricing and provides the necessary background on the RQMC method.  In Section \ref{sec:RQMC_fourier_pricing}, we explain our methodology. In Section \ref{sec:general_formulation}, we motivate the importance of appropriately handling the domain transformation to obtain nearly optimal convergence rates of the RQMC method. Then, in Section \ref{sec:model_spec_dom_transf}, we present practical domain transformation strategies for the GBM, the VG and the GH models, based on boundary growth conditions on the transformed integrand, summarized in Tables \ref{tab:univariate_dom_tranf} and \ref{tab:multivariate_dom_tranf}. Finally, in Section \ref{sec: num_exp_results}, we report and analyze the obtained numerical results. We illustrate the advantages of the proposed domain transformation on the rate of convergence of the RQMC method. Furthermore, in Section \ref{sec:qmc_vs_quad_run}, we highlight the considerable computational gains achieved compared to the TP-Laguerre quadrature and MC method in the Fourier space and compared to the MC method in the physical space, for options with up to 15 assets.

\section{Problem Setting and  Background}\label{sec:Problem Setting and Pricing Framework}
Section~\ref{sec:fourier_val_formula} briefly revisits the general Fourier valuation framework for multi-asset options considered in this work (explained in more detail in a previous study \cite{bayer2023optimal}). Then, Section~\ref{sec:qmc_intro} describes the RQMC method.


\subsection{Fourier Valuation Formula}
\label{sec:fourier_val_formula}
Prior to presenting the valuation formula in Proposition~\ref{prop:Multivariate Fourier pricing valuation formula}, we introduce the necessary notation, definitions, and assumptions.
\begin{notation}[Notations and Definitions]\label{notation}\
	
	\begin{itemize}
		\item   $\boldsymbol{X}_t:=\left(X_t^1,\dots, X_t^d\right)$ is a $d$-dimensional ($d \in \mathbb{N}$) vector of log-asset prices\footnote{$X^i_t:= \log(S^i_t), i = 1, \ldots, d$, $\{S^i_t \}_{i=1}^d$ are asset prices at time $t > 0$.} whose dynamics follow a multivariate stochastic model with the market parameters denoted by the vector $\boldsymbol{\Theta}_X$. 
		\item $\Phi_{\boldsymbol{X}_T}(\boldsymbol{z}):= \mathbb{E}[e^{\mathrm{i} \boldsymbol{z}^{\top} \boldsymbol{X}_T }]$, for $\boldsymbol{z} \in \mathbb{C}^d$, denotes the extended characteristic function, where $\boldsymbol{z}^{\top}$ represents the transpose of $\boldsymbol{z}$. We define the bilinear form $\boldsymbol{x}^{\top} \boldsymbol{y}= \sum_{j = 1}^d x_j y_j$ for $\boldsymbol{x}, \boldsymbol{y} \in \mathbb{C}^d$. 
		\item  $P: \mathbb{R}^d \mapsto \mathbb{R}_+$ denotes the payoff function, and $\hat{P}(\boldsymbol{z}):= \int_{\mathbb{R}^d} e^{ - \mathrm{i} \boldsymbol{z}^{\top} \boldsymbol{x}} P(\boldsymbol{x}) d\boldsymbol{x} $, for $\boldsymbol{z} \in \mathbb{C}^d$ represents its extended Fourier transform. 
		\item  $\boldsymbol{\Theta}_P = (K, T, r)$ represents the vector of the payoff parameters, where $K$ denotes the strike price, $T$ is the maturity time, and $r$ is the risk-free interest rate.
		\item $\mathrm{i}$ denotes the unit imaginary number, and $\Re[\cdot]$ and $\Im[\cdot]$ represent the real and imaginary parts of a complex number, respectively.
		\item $L^1(\mathbb{R}^d)$ denotes the space of integrable functions on $\mathbb{R}^d$.
		\item $\boldsymbol{A} \succ 0$ (respectively $\boldsymbol{A} \succeq 0$) denotes positive (semi-)definiteness, and $\boldsymbol{A} \prec 0$ (respectively $\boldsymbol{A} \preceq 0$) denotes the negative (semi-)definiteness of the matrix $\boldsymbol{A} \in \mathbb{R}^{d \times d}$. 
		
		\item Let $\mathbb{I}_d := \{1,\ldots,d\}$.
		
		\item $\Gamma(z) = \int_{0}^{+ \infty} e^{-t} t^{z-1} dt$ is the complex Gamma function defined for $\Re{[z]} > 0 $.
		
		\item $K_{v}(y)$ is the modified Bessel function of the second kind with $v = \frac{2 - d}{2}$, defined for $y >0$, see \cite{kotz2001laplace}.
		
		\item Let $\boldsymbol{u} = (u_1, u_2, \ldots, u_m) \in [0,1]^m$ with $m,d \in \mathbb{N}, m \geq d$. We define the subvector $\boldsymbol{u}_{1:d}$ by $\boldsymbol{u}_{1:d} := (u_1, u_2, \ldots, u_d)$.
	\end{itemize}
\end{notation}

\begin{assumption}[Assumptions on the payoff]\label{ass:Assumptions on  the payoff}\
	\begin{itemize} 
	
		\item $ \delta_P :=  \{ \boldsymbol{R} \in \mathbb{R}^d \; |  \; \boldsymbol{x} \mapsto e^{ \boldsymbol{R}^{\top} \boldsymbol{x} }P(\boldsymbol{x}) \in L^1(\mathbb{R}^d)  \} \neq \emptyset$.
	\end{itemize}
\end{assumption} 

\begin{assumption}[Assumptions on the model]\label{ass:Assumptions on  the distribution}\
	\begin{itemize}	
			\item $ \delta_X : = \{ \boldsymbol{R} \in \mathbb{R}^d \; | \; \boldsymbol{y} \mapsto \mid \Phi_{\boldsymbol{X}_T}(\boldsymbol{y} + \mathrm{i} \boldsymbol{R}) \mid < \infty, \text{and} \; \boldsymbol{y}   \mapsto  \Phi_{\boldsymbol{X}_T}(\boldsymbol{y} + \mathrm{i} \boldsymbol{R}) \in L^1(\mathbb{R}^d) \} \neq \emptyset $.

	\end{itemize} 	
\end{assumption}

\begin{proposition}[Multivariate Fourier Pricing Valuation Formula]\label{prop:Multivariate Fourier pricing valuation formula} We employ the notation in \ref{notation} and suppose that Assumptions~\ref{ass:Assumptions on the payoff} and \ref{ass:Assumptions on the distribution} hold and that $\delta_V = \delta_X \cap \delta_P  \neq \emptyset$. Then, for $\boldsymbol{R} \in \delta_V$, the option value is given by the following:
	\begin{equation}
	 V\left(\boldsymbol{\Theta}_X, \boldsymbol{\Theta}_P\right)   = (2 \pi)^{-d} e^{-r T} \int_{\mathbb{R}^d} \Re\left[  \Phi_{\boldsymbol{X}_{T}}(\boldsymbol{y}+\mathrm{i} \boldsymbol{R}) \widehat{P}(\boldsymbol{y}+\mathrm{i} \boldsymbol{R}) \right] \mathrm{d} \boldsymbol{y}.
		\label{QOI}
	\end{equation}
\end{proposition}
\begin{proof} We refer the reader to Appendix \ref{sec:proof_of_fourier_valuation_formula} for the detailed proof. For continuous payoffs, the valuation formula \eqref{QOI} can be derived under less restrictive assumptions on the model; see \cite{bayer2023optimal}.
\end{proof}

From \eqref{QOI}, we define the integrand of interest as follows:
\begin{equation}
	g\left(\boldsymbol{y} ; \boldsymbol{R}, \boldsymbol{\Theta}_X, \boldsymbol{\Theta}_P\right) := (2 \pi)^{-d} e^{-r T} \Re [  \Phi_{\boldsymbol{X}_{T}}(\boldsymbol{y}+\mathrm{i} \boldsymbol{R}) \widehat{P}(\boldsymbol{y}+\mathrm{i} \boldsymbol{R}) ], \boldsymbol{y} \in \mathbb{R}^d, \boldsymbol{R} \in \delta_V.
	\label{g_integrand}
\end{equation}

Bayer et al. \cite{bayer2023optimal} proposed a rule for the choice of the damping parameters, $\boldsymbol{R}$, that leads to a regular integrand, and numerical evidence shows that their rule can accelerate the convergence of numerical quadrature methods significantly. Their proposed rule  is given by:
\begin{equation}
	\label{eq:optimal_damping_rule}
	\mathbf{R}^* =\underset{\mathbf{R} \in \delta_V}{\arg \min } \; g\left(\mathbf{0}; \mathbf{R}, \boldsymbol{\Theta}_X, \boldsymbol{\Theta}_P\right) .
\end{equation}
The advantage of this rule is that the numerical computation of the optimal damping parameters is very fast (in the order of milliseconds), and it works for a wide range of payoff functions and asset-price dynamics; we refer to \cite{bayer2023optimal} for more details. In the remainder of the paper, we use the values of the damping parameters calculated according to rule \eqref{eq:optimal_damping_rule}.

The integrand in \eqref{g_integrand} is analytic along the strip $\boldsymbol{z} = \boldsymbol{y} + \mathrm{i} \boldsymbol{R} \in \mathbb{R}^d + \mathrm{i} \delta_V \subseteq \mathbb{C}^d$. For examples of strips of analyticity we refer to Tables \ref{table:strip_table} and \ref{table:payoff_strip_table}. This observation motivates the use of quadrature methods that leverage the analyticity to enhance the convergence rates compared to traditional approaches, such as the MC method \cite{bayer2023optimal}. In addition, we suggest the use of the RQMC quadrature in the Fourier space to address the integration problem efficiently in high dimensions. Section~\ref{sec:qmc_intro} introduces the RQMC method, and Section~\ref{sec:RQMC_fourier_pricing} explains the necessary transformations to implement the RQMC method in the Fourier setting. Next, Section~\ref{sec: num_exp_results} presents the benefits of adopting this approach through concrete examples involving basket and rainbow options under various pricing models.

\subsection{RQMC Method}
\label{sec:qmc_intro}
This section introduces the RQMC method.  The QMC estimator of an integral of a function, $f:[0,1]^d \mapsto \mathbb{R}$, is an $N$-point equal-weighted quadrature rule  denoted by $Q^{QMC}_{N, d}[\cdot]$, and reads as:
\begin{equation}
	\label{qmc_estimator}
	I[f] := \int_{[0,1]^d} f(\boldsymbol{u}) \mathrm{d}\boldsymbol{u} \approx Q^{QMC}_{N,d}[f]:= \frac{1}{N}\sum_{n = 1}^N  f(\boldsymbol{u}_n),
	\end{equation}
	where $\boldsymbol{u}_1,\ldots,\boldsymbol{u}_N$ is a set of deterministic LD sequences (e.g., Sobol, Niederreiter, Halton, Hammersley, and Faure, see \cite{dick2013high}), $\boldsymbol{u}_n \in [0,1]^d, n \in \{ 1,\ldots, N\}$. The advantage of \eqref{qmc_estimator} compared to the MC estimator is that the points $\{\boldsymbol{u}_n  \}_{n = 1}^N$ are generated to ensure the more uniform coverage of $[0,1]^d$. Consequently, the estimator \eqref{qmc_estimator} may achieve a convergence rate of order $\mathcal{O}(N^{-1 + \epsilon_r})$ \cite{owen2006halton}, with $\epsilon_r> 0$, depending on the regularity of $f(\cdot)$ and the dimension, $d$, of the domain.  In contrast, MC points are sampled randomly and independently and may cluster and miss important regions of the integrand, unless importance sampling techniques are employed \cite{glasserman2004monte}. Nonetheless, the shortcoming of the estimator in \eqref{qmc_estimator} is that the central limit theorem cannot be directly applied to obtain probabilistic error estimates as in the case of the MC method. The points $\{\boldsymbol{u}_n \}_{n = 1}^N$ are not sampled independently. In \cite{hlawka1961funktionen}, a deterministic error bound for the estimator in \eqref{qmc_estimator} was derived, known as the Koksma--Hlawka inequality. 
	This error bound is usually impractical because its computation involves the integration of the mixed first partial derivatives of the integrand, which can be more difficult than solving the original problem. As a remedy, a randomized variant of the QMC estimator \eqref{qmc_estimator} was introduced (see \cite{l2002recent}), called the RQMC estimator, and it is given by:
	\begin{equation}
		\label{eq:RQMC_estimator}
		Q^{RQMC}_{N,S,d}[f] := \frac{1}{S} \sum_{s = 1}^S  \frac{1}{N}\sum_{n = 1}^N f( \boldsymbol{u}_n^{(s)} ),
	\end{equation}
where $\{ \boldsymbol{u}_n\}_{n = 1}^N$ is the sequence of deterministic QMC points, and for $n =1,\ldots N$, $\{ \boldsymbol{u}_n^{(s)} \}_{s =1}^S$ is obtained by the appropriate randomization of $\{ \boldsymbol{u}_n\}_{n = 1}^N$, such that $ \boldsymbol{u}^{(s)}_n \sim \mathcal{U}([0,1]^d) $. For fixed $n=1,\ldots, N$, $ \boldsymbol{u}_n^{(s)}$ are independent for any $s = 1, \ldots, S$. An additional rationale to apply the RQMC estimator is that the set of points $\{ \boldsymbol{u}_n^{(s)} \}_{s =1}^S$ yields an unbiased estimator \eqref{eq:RQMC_estimator} (i.e., $\mathbb{E}[Q^{RQMC}_{N, S, d}[f]] = I[f]$). Several randomization methods exist with different theoretical guarantees (for an overview of the most popular methods, see \cite{l2002recent}). This work adopts Sobol sequences \cite{sobol2011construction} with digital shifting for the randomization \cite{cranley1976randomization}. 
Finally, the randomization of the LD points enables the derivation of the root mean squared error of the estimator, given by \cite{dick2013high}
	\begin{equation}
		\label{rqmc_error}
		\mathcal{E}^{RQMC}_{N,S,d}[f]: = \frac{C_{\alpha}}{\sqrt{S}} \sqrt{    \frac{1}{S-1} \sum_{s= 1}^S \left(   \frac{1}{N} \sum_{n = 1}^N f( \boldsymbol{u}^{(s)}_n )  - Q^{RQMC}_{N,S,d}[f]   \right)^2    },
	\end{equation} 
	where $C_{\alpha}$ denotes the $(1-\frac{\alpha}{2})$-quantile of the standard normal distribution for a confidence level $0< \alpha \ll1$. In this paper, we work with $C_{\alpha} = 1.96$, corresponding to a $95 \%$ confidence interval. Moreover, \eqref{rqmc_error} reveals that the statistical error can be controlled by the number of digital shifts, $S$, which we apply in the order of $S = 30$ to compute the error estimate. 
	
  There is growing literature on the application of the QMC method to singular/unbounded integrands which arise, for instance, when dealing with integrands on unbounded domains \cite{owen2006halton,kuo2006randomly,ouyang2024achieving,liu2023nonasymptotic}. Owen \cite{owen2006halton} analysed the asymptotic rate of convergence of the RQMC method for integrands that are singular at the boundary of $[0,1]^d$, satisfying the following boundary growth condition:
\begin{equation}		
	\label{eq:owen_bgc}
		\left|\partial^{\boldsymbol{k}} f(\boldsymbol{u})\right| \leq B \prod_{j=1}^d \min \left(u_j, 1-u_j\right)^{-A_j-1_{j \in \boldsymbol{k}}},	
\end{equation}
where $\partial^{\boldsymbol{k}} f(\boldsymbol{u}) = \prod_{j \in \boldsymbol{k}} \left(\frac{ \partial f(\boldsymbol{u})}{\partial u_j} \right)$, for some $A_j > 0$,  some $B < \infty$, and for all $\boldsymbol{k} \subset \mathbb{I}_d$.  He proved that if $\boldsymbol{u}_1, \ldots \boldsymbol{u}_N$  are randomly sampled LD points on  $[0,1]^d$, with $\boldsymbol{u}_n \sim \mathcal{U}([0,1]^d)$ for $n \in \{1,\ldots,N\}$,  and  star discrepancy\footnote{We refer to \cite{dick2013high} for the definition of the star discrepancy.} satisfying $\mathbb{E}\left(D_n^*\left(x_1, \ldots, x_n\right)\right)=\mathcal{O}\left(N^{-1+\epsilon_r}\right)$ for all $\epsilon_r>0$, then the asymptotic rate of convergence of the RQMC method is given by:
\begin{equation}
	\label{eq:owen_rqmc_rate}
	\mathbb{E}\left[ \mid  I[f] - 	Q^{RQMC}_{N,S,d}[f] \mid \right] = \mathcal{O}\left(N^{-1 + \max_j A_j + \epsilon_r}\right).
\end{equation}
Equation \eqref{eq:owen_rqmc_rate} shows that the RQMC method may converge even for singular integrands and have a better rate than the MC method if $\max_j A_j < \frac{1}{2}$. However, the rate of convergence of RQMC is significantly impacted by the rate of growth of the integrand at the boundary, described by the condition on the rate of growth of the mixed first partial derivatives of the integrand in \eqref{eq:owen_bgc}. Moreover, Owen pointed out that the implied constant implicit in the $\mathcal{O}(\cdot)$ factor in \eqref{eq:owen_rqmc_rate} can blow up for unbounded integrands, deteriorating the performance of RQMC.  

All the aforementioned works lead to the same conclusion, that the stronger the singularity at the boundary, the worse the error rates of RQMC. However, the application of appropriate importance sampling \cite{ouyang2024achieving,liu2023nonasymptotic}, which in our setting relates to an appropriate domain transformation (see Section \ref{sec:RQMC_fourier_pricing}), can lead to improved  convergence.
	


Finally, despite the advantages of using RQMC to obtain  computable error estimates, this method suffers from a significant drawback, which may impede its application, especially for high-dimensional problems. Specifically, the construction of QMC points is constrained to simple geometries, such as the hypercube, $[0,1]^d$; hence, transformations of the original domains to the hypercube are necessary to address unbounded integrals in the form of \eqref{QOI}, as we discuss in Section~\ref{sec:RQMC_fourier_pricing}.
\section{Efficient Domain Transformation for RQMC in  Fourier Pricing}
\label{sec:RQMC_fourier_pricing}
\label{sec:RQMC_fourier_pricing}
{
	This section explains how the RQMC estimator, defined in \eqref{eq:RQMC_estimator}, can be applied to approximate the Fourier integral defined on the unbounded domain $\mathbb{R}^d$. Section~\ref{sec:general_formulation} details the general approach to transforming the pricing problem defined in Proposition~\ref{prop:Multivariate Fourier pricing valuation formula} from $\mathbb{R}^d$ to $[0,1]^d$. Section~\ref{sec:model_spec_dom_transf} develops a domain transformation strategy for the three decay classes defined in Table~\ref{tab:identification}, namely light-tailed, semi-heavy-tailed, and heavy-tailed characteristic functions. These classes are illustrated by the GBM, GH/NIG, and VG models, respectively. We note that the same logic of the domain transformation applies for the general computation of deterministic integrals on $\mathbb{R}^d$ using QMC, and is not limited to option pricing applications.

\subsection{General Formulation}
\label{sec:general_formulation}

The problem we address is the computation of a deterministic integral over $\mathbb{R}^d$ in the form of:

\begin{equation}
	\label{eq:deterministic_integral}
	\int_{\mathbb{R}^d} g(\boldsymbol{y}) \mathrm{d} \boldsymbol{y}
\end{equation}

where $g(\cdot)$ is given in \eqref{g_integrand}. Motivated by the smoothness of the integrand, we aim to apply the RQMC method to evaluate the integral in \eqref{eq:deterministic_integral}. However, the main QMC constructions are restricted to the generation of LD point sets on $[0,1]^d$ (for a comprehensive survey, see [41]). Consequently, an appropriate integral transformation is necessary to apply RQMC to the unbounded domain, $\mathbb{R}^d$. In the first step, we express the integral as an expectation with respect to a probability density function  (PDF) $\psi_{\boldsymbol{Y}}(\cdot)$ as follows
\begin{equation}
	\int_{\mathbb{R}^d} g(\boldsymbol{y}) \mathrm{d} \boldsymbol{y} = 	\int_{\mathbb{R}^d} \frac{g(\boldsymbol{y})}{\psi_{\boldsymbol{Y}}(\boldsymbol{y})} \psi_{\boldsymbol{Y}}(\boldsymbol{y}) \mathrm{d} \boldsymbol{y},
\end{equation}
where $\boldsymbol{Y} $ is an $\mathbb{R}^d$-valued random vector. Then, we perform a domain transformation to map the integral from $\mathbb{R}^d$ to the hypercube, which we describe in the next paragraph.

In this work, we distinguish three cases for the domain transformation i) $\boldsymbol{Y}$ can be expressed in the normal variance-mean mixture form ii) $\boldsymbol{Y}$ is normally distributed, iii) $\boldsymbol{Y}$ has independent components. To illustrate the main idea, we introduce the different transformations and refer to Section 3.2 for a more detailed explanation. In the first case i), we have that $\boldsymbol{Y} \stackrel{d}{=} \tilde{\boldsymbol{\mu}}+W^{\frac{\gamma}{2}} \boldsymbol{N},{ }^5$ with $\gamma \in\{-1,1\}, W>0$ a scalar mixing RV, $\boldsymbol{N} \stackrel{d}{=} \tilde{\boldsymbol{\mu}}+W^{\frac{\gamma}{2}} \tilde{\boldsymbol{L}} \boldsymbol{Z}$ follows a $d$-variate normal distribution with mean $\tilde{\boldsymbol{\mu}}$ and covariance matrix $\tilde{\boldsymbol{\Sigma}}, \boldsymbol{Z}$ follows the $d$-variate standard normal distribution with $\tilde{\boldsymbol{L}}$ being the square root matrix of $\tilde{\boldsymbol{\Sigma}}$. In what follows, we set $\tilde{\boldsymbol{\mu}}=\mathbf{0}$, and we denote by $\Psi_{\boldsymbol{Y}}^{-1}(\boldsymbol{u})=\left(\Psi_{Y_1}^{-1}\left(u_1\right), \ldots, \Psi_{Y_d}^{-1}\left(u_d\right)\right), \boldsymbol{u} \in[0,1]^d$, the componentwise application of the inverse cumulative distribution function (ICDF) of each component of $\boldsymbol{Y}$. The domain trasformation can then be achieved by the following transformation function
\begin{equation}
	\label{eq:general_transformation_function}
	\begin{aligned}
		\mathcal{T}_{\text{mix}} \colon [0,1]^{d+1} &\mapsto \mathbb{R}^d \times \mathbb{R}^+ \\
		\boldsymbol{u} &\mapsto	\mathcal{T}_{\text{mix}}(\boldsymbol{u}) := \left(
		(\Psi_W^{-1}(u_{d+1}))^{\frac{\gamma}{2}} \tilde{\boldsymbol{L}} \Psi_{\boldsymbol{Z}}^{-1}(\boldsymbol{u}_{1: d}), \,
		\Psi_W^{-1}(u_{d+1})
		\right),
	\end{aligned}
\end{equation}
The mapping \eqref{eq:general_transformation_function} leads to the following $(d+1)$-dimensional integration problem:
\begin{equation}
	\begin{aligned}
		\label{eq:dom_transf_general}
		\int_{\mathbb{R}^d} 	
		\frac{g(\boldsymbol{y})}{\psi_{\boldsymbol{Y}}(\boldsymbol{y})} \psi_{\boldsymbol{Y}}(\boldsymbol{y}) \mathrm{d} \boldsymbol{y} &= \int_{\mathbb{R}^d}  \frac{g(\boldsymbol{y})}{\psi_{\boldsymbol{Y}}(\boldsymbol{y})} \left(   \int_{\mathbb{R}^+} w^{\frac{\gamma d}{2}} \psi_W(w) \psi_{\boldsymbol{N}}\left(w^{\frac{\gamma}{2}} \boldsymbol{y}  \right) \mathrm{d} w \right) \mathrm{d}\boldsymbol{y} \\
		& = \int_{\mathbb{R}^d} 	\int_{\mathbb{R}^+} \underbrace{\frac{g\left(w^{-\frac{\gamma}{2}} \boldsymbol{y}\right)}{\psi_{\boldsymbol{Y}}\left(w^{-\frac{\gamma}{2}} \boldsymbol{y}\right)} \psi_W(w) \psi_{\boldsymbol{N}}(\boldsymbol{y})}_{g_{\psi}(\boldsymbol{y}, w)}  \mathrm{d}\boldsymbol{y}  \mathrm{d}w\\
		& = \int_{[0,1]^{d+1}}\underbrace{ g_{\psi} \circ \mathcal{T}_{\text{mix}}(\boldsymbol{u}) \left|\operatorname{det} (J_{\mathcal{T}_{\text{mix}}}(\boldsymbol{u})) \right| }_{:= \tilde{g}_{\mathcal{T}_{\text{mix}}}(\boldsymbol{u})} \mathrm{d} \boldsymbol{u}. \\
	\end{aligned}
\end{equation}

	In contrast, in the remaining cases, the domain transformation results in the following $d$-dimensional integration problem:

\begin{equation}
	\label{eq:transformed_fourier_pricing_integral}
	\begin{aligned}
		\int_{\mathbb{R}^d}  
		\frac{g(\boldsymbol{y})}{\psi_{\boldsymbol{Y}}(\boldsymbol{y})} \psi_{\boldsymbol{Y}}(\boldsymbol{y}) \, \mathrm{d} \boldsymbol{y} 
		& = \int_{[0,1]^{d}} \underbrace{g \circ \mathcal{T}(\boldsymbol{u}) \left|\det \left(J_{\mathcal{T}}(\boldsymbol{u})\right) \right| }_{\tilde{g}_{\mathcal{T}}(\boldsymbol{u})} \, \mathrm{d} \boldsymbol{u}, \quad \mathcal{T}: [0,1]^{d} \mapsto \mathbb{R}^d. \\
	\end{aligned}
\end{equation}
{
	The transformation maps used in cases ii) and iii) are defined by

	\begin{equation}
		\label{eq:T_nor_T_ind}
		\mathcal{T}_{\mathrm{nor}}
		:=
		\tilde{\boldsymbol{L}}\Psi_{\boldsymbol{Z}}^{-1},
		\qquad
		\mathcal{T}_{\mathrm{ind}}(\boldsymbol{u})
		:=
		\Psi^{-1}_{\boldsymbol{Y}}(\boldsymbol{u}),
		\qquad
		\boldsymbol{u}\in[0,1]^d .
	\end{equation}

	In \eqref{eq:transformed_fourier_pricing_integral}, $\mathcal{T}=\mathcal{T}_{\mathrm{nor}}$ in case ii), while $\mathcal{T}=\mathcal{T}_{\mathrm{ind}}$ in case iii). The latter case is the one that is typically treated in the literature.

The primary challenge of the domain transformation in \eqref{eq:transformed_fourier_pricing_integral}  is that it often results in integrands that are unbounded near the boundary \cite{owen2006halton}. In fact, since $\lim_{u_j \to \{0,1\}} \mathcal{T}(\boldsymbol{u}) = \pm \infty$, then we have that $\lim_{u_j \to \{0,1\}}  g \circ \mathcal{T}(\boldsymbol{u}) = 0$ for any $j \in \mathbb{I}_d$. Moreover, in the considered cases we have that $\lim_{u_j \to \{0,1\}} \left|\det \left(J_{\mathcal{T}}(\boldsymbol{u})\right) \right| \to \infty$.\footnote{Through symmetry; the same argument applies when $u_j \to 0$ or $u_j \to 1$, hence, in the remainder of the paper, we choose to study the limiting behavior when $u_j \to 0$.} For instance, $\lim_{u_j \to \{0,1\}} \left|\det \left(J_{\mathcal{T_{\text{ind}}}}(\boldsymbol{u})\right) \right|  = \lim_{u_j \to \{0,1\}} \frac{1}{\psi_{\boldsymbol{Y}} \circ \Psi_{\boldsymbol{Y}} ^{-1}(\boldsymbol{u})} = +\infty$.  Consequently, in all cases i), ii), and iii), depending on the choice of $\psi_{\boldsymbol{Y}}(\cdot)$ and its parameters, the resulting transformed integrand, $\tilde{g}(\cdot)$,\footnote{In the remainder of the paper we refer to $\tilde{g}_{\mathcal{T}}(\cdot)$ by $\tilde{g}(\cdot)$ to simplify the notation.} can be singular at the boundaries of $[0,1]^d$, which would deteriorate the rate of convergence of RQMC, as explained in Section \ref{sec:qmc_intro}.    
Therefore, it is critical to find an appropriate choice of $\psi_{\boldsymbol{Y}}(\cdot)$ that ensures $\tilde{g}(\cdot)$ has no singularities.

Compared to our setting, most of the literature \cite{kuo2006randomly,kuo2011quasi, nichols2014fast, kuo2010randomly,ouyang2024achieving} addresses the weighted integration problem \eqref{eq:literature_integrand}, which takes a different form than the problem \eqref{eq:transformed_fourier_pricing_integral}, and is given by:
\begin{equation}
	\label{eq:literature_integrand}
	\int_{\mathbb{R}^d} g( \boldsymbol{y}) \left( \prod_{j = 1}^d \rho(y_j) \right)  \mathrm{d}y_1 \ldots \mathrm{d} y_d,
\end{equation}
where $\rho(\boldsymbol{y}) = \prod_{j = 1}^d \rho(y_j)$ is the joint PDF of independent RVs (case (iii)). Applying the change of variable $\boldsymbol{y} = \mathcal{T}_{\text{ind}}(\boldsymbol{u})$ to \eqref{eq:literature_integrand} results in the following
\begin{equation}
	\label{eq:literature_transformed_integrand}
	\int_{\mathbb{R}^d} g( \boldsymbol{y}) \left( \prod_{j = 1}^d \rho(y_j) \right)  \mathrm{d}y_1 \ldots \mathrm{d} y_d = \int_{[0,1]^d} g \circ \Psi_{\boldsymbol{Y}}^{-1}(\boldsymbol{u}) \left( \prod_{j = 1}^d \frac{\rho \circ \Psi_{\boldsymbol{Y}}^{-1}(u_j)}{ \psi_{\boldsymbol{Y}}\circ \Psi_{\boldsymbol{Y}}^{-1}(u_j)} \right) \mathrm{d} u_1\ldots \mathrm{d}u_d.
\end{equation}
The computation of the $\rho$-weighted integral in \eqref{eq:literature_transformed_integrand} is simpler than the integration in \eqref{eq:transformed_fourier_pricing_integral} due to more constrained choices of $\psi_{\boldsymbol{Y}}(\cdot)$ in the latter setting. Indeed, the set of density functions $\psi_{\boldsymbol{Y}}(\cdot)$ that ensure $\frac{g \circ \Psi_{\boldsymbol{Y}}^{-1}}{\psi_{\boldsymbol{Y}} \circ \Psi_{\boldsymbol{Y}}^{-1}} \in L^1(\mathbb{R}^d)$ is a subset of those guaranteeing $\frac{(g \rho) \circ \Psi_{\boldsymbol{Y}}^{-1}}{\psi_{\boldsymbol{Y}} \circ \Psi_{\boldsymbol{Y}}^{-1}} \in L^1(\mathbb{R}^d)$, particularly when $\rho(\cdot)$ decays rapidly (e.g., a Gaussian density).

	Our objective is to select a transformation density $\psi_{\boldsymbol{Y}}(\cdot)$ that ensures the transformed integrand $\tilde{g}(\cdot)$ remains well-behaved. Specifically, we aim to achieve $\lim_{u_j \to \{0,1\}} \tilde{g}(\boldsymbol{u}) = 0$ for any $j \in \mathbb{I}_d$, which requires comparing the asymptotic decay of $g(\cdot)$ and $\psi_{\boldsymbol{Y}}(\cdot)$.
	
	For $\boldsymbol{R}\in\delta_V\subseteq\delta_P$, Assumption~\ref{ass:Assumptions on  the payoff} implies

	\begin{equation}
		\label{eq:payoff_transform_bound}
		\left|
		\widehat{P}(\boldsymbol{y}+\mathrm{i}\boldsymbol{R})
		\right|
		\leq
		\int_{\mathbb{R}^d}
		e^{\boldsymbol{R}^{\top}\boldsymbol{x}}
		|P(\boldsymbol{x})|
		\,\mathrm{d}\boldsymbol{x}
		=:C_P(\boldsymbol{R})<\infty,
		\qquad
		\boldsymbol{y}\in\mathbb{R}^d .
	\end{equation}

	Consequently, by the definition of $g$ in \eqref{g_integrand}, we obtain

	\begin{equation}
		\label{eq:ratio_controls_full_integrand}
		\left|
		\frac{
			g(\boldsymbol{y};\boldsymbol{R},\boldsymbol{\Theta}_X,\boldsymbol{\Theta}_P)
		}{
			\psi_{\boldsymbol{Y}}(\boldsymbol{y})
		}
		\right|
		\leq
		(2\pi)^{-d}e^{-rT}
		C_P(\boldsymbol{R})
		\left|
		\frac{
			\Phi_{\boldsymbol{X}_T}(\boldsymbol{y}+\mathrm{i}\boldsymbol{R})
		}{
			\psi_{\boldsymbol{Y}}(\boldsymbol{y})
		}
		\right|.
	\end{equation}
}
{
	Thus, \eqref{eq:ratio_controls_full_integrand} shows that controlling the ratio between the extended characteristic function and the transformation density is a sufficient condition for controlling the boundary growth of the full transformed Fourier integrand. This motivates the general ratio
}
{
	\begin{equation}
		\label{eq:general_boundary_ratio}
		r(\boldsymbol{y};\boldsymbol{R})
		:=
		\frac{
			\Phi_{\boldsymbol{X}_T}(\boldsymbol{y}+\mathrm{i}\boldsymbol{R})
		}{
			\psi_{\boldsymbol{Y}}(\boldsymbol{y})
		},
		\qquad
		\boldsymbol{y}\in\mathbb{R}^d,
		\quad
		\boldsymbol{R}\in\delta_V .
	\end{equation}
}
{
	In Sections~\ref{sec:gbm_dom_transf}, \ref{sec:nig_dom_transf}, and~\ref{sec:vg_dom_transf}, the ratio in \eqref{eq:general_boundary_ratio} is analyzed for representative pricing models and transformation densities. The transformation density $\psi_{\boldsymbol{Y}}(\cdot)$ is chosen so that its tail is not lighter than the leading decay of the extended characteristic function. For dependent asset prices, $\psi_{\boldsymbol{Y}}(\cdot)$ is chosen to reflect the dependence structure in $\Phi_{\boldsymbol{X}_T}(\cdot)$. Since the standard component-wise mapping $\mathcal{T}_{\mathrm{ind}}(\cdot)$ defined in \eqref{eq:T_nor_T_ind} cannot be naturally extended to dependent $\boldsymbol{Y}$ due to the non-invertibility of the associated ICDF, we use the transformations $\mathcal{T}_{\mathrm{mix}}(\cdot)$ and $\mathcal{T}_{\mathrm{nor}}$ defined in \eqref{eq:general_transformation_function} and \eqref{eq:T_nor_T_ind}, respectively. Table~\ref{table:chf_table} provides examples of characteristic functions treated in this paper.
}

\FloatBarrier
\begin{table}[h]
	\centering
	\begin{tabular}{| p{1.3cm} | p{9.5cm} | }
		\hline   \textbf{Model} & $\phi_{\boldsymbol{X}_T}(\boldsymbol{z}), \boldsymbol{z} \in \mathbb{C}^d, \: \Im[\boldsymbol{z}] \in \delta_X$ \\
		\hline
		\textbf{GBM} &   \small $  \exp \left( -\frac{T}{2}  \boldsymbol{z}^{\top}\boldsymbol{\Sigma} \boldsymbol{z}\right)$  \\ 
		\hline
		\textbf{NIG}  & \small $ \exp \left(\delta T\left(\sqrt{\alpha^{2}-\boldsymbol{\beta}^{\top} \boldsymbol{\Delta} \boldsymbol{\beta}}-\sqrt{\alpha^{2}-(\boldsymbol{\beta}+ \mathrm{i} \boldsymbol{z})^{\top}\boldsymbol{\Delta}(\boldsymbol{\beta}+ \mathrm{i} \boldsymbol{z})}\right)\right)$    \\
		\hline
		
		\textbf{GH} & \small $\left(\frac{\alpha^2-\boldsymbol{\beta}^{\top} \boldsymbol{\Delta} \boldsymbol{\beta}}{\alpha^2  -(\boldsymbol{\beta}+ \mathrm{i} \boldsymbol{z})^{\top}\boldsymbol{\Delta}(\boldsymbol{\beta}+ \mathrm{i} \boldsymbol{z})}\right)^{\lambda / 2} \frac{\mathrm{~K}_\lambda\left(\delta T \sqrt{\alpha^2-(\boldsymbol{\beta}+ \mathrm{i} \boldsymbol{z})^{\top}\boldsymbol{\Delta}(\boldsymbol{\beta}+ \mathrm{i} \boldsymbol{z})}\right)}{\mathrm{K}_\lambda\left(\delta T \sqrt{\alpha^2-\boldsymbol{\beta}^{\top} \boldsymbol{\Delta} \boldsymbol{\beta}}\right)}$  \\
		\hline
		\textbf{VG} & \small $ \left(1-\mathrm{i} \nu  \boldsymbol{z}^{\top}  \boldsymbol{\theta}+\frac{1}{2} \nu \boldsymbol{z}^{\top} \boldsymbol{\Sigma} \boldsymbol{z}\right)^{-T / \nu} $\\
		\hline
	\end{tabular}
	\caption{The extended characteristic function of the pricing models is given by $\Phi_{\boldsymbol{X}_T}(\boldsymbol{z}) = \exp \left(\mathrm{i} \boldsymbol{z}^{\top} (\boldsymbol{X}_{0} + (r + \boldsymbol{\mu}) T) \right) \phi_{\boldsymbol{X}_T}(\boldsymbol{z})$. $\boldsymbol{\mu} $ is the martingale correction term, defined for each model in the Appendix \ref{sec:pricing_models}, $r$ is the risk-free interest rate. More details on the parameters of each model are provided in the Appendix~\ref{sec:pricing_models}. Table~\ref{table:strip_table} provides the strip of analyticity, $\delta_X$, for each of the characteristic functions.   }	
	\label{table:chf_table}
\end{table}
\FloatBarrier

To summarize, in this work, we base the choice of $\psi_{\boldsymbol{Y}}(\cdot)$ on the following properties:   
\begin{itemize}
	\item We consider transformation densities, $\psi_{\boldsymbol{Y}}(\cdot)$, supported on $\mathbb{R}^d$, that are smooth and symmetric around the origin to match the corresponding features of the original integrand \eqref{g_integrand} (a consequence of Fourier transform properties \cite{lukacs1970characteristic}). 
	\item We consider $\psi_{\boldsymbol{Y}}(\cdot)$, the PDF of a normal variance-mean mixture distribution. This choice provides a flexible framework to handle the dependence structures in our integration problem using $\boldsymbol{Y}$ with dependent components, see \eqref{eq:dom_transf_general}.
	{
		\item We select $\psi_{\boldsymbol{Y}}(\cdot)$ to asymptotically follow the same functional form as the extended characteristic function. Specifically, for fixed $\boldsymbol{R}\in\delta_V$ and $\boldsymbol{y}\in\mathbb{R}^d$, we use the three decay classes defined in Table~\ref{tab:identification}:
		\begin{itemize}
			\item Light-tailed:
			$
			|\Phi_{\boldsymbol{X}_T}(\boldsymbol{y}+\mathrm{i}\boldsymbol{R})|
			\leq
			C\exp(-\boldsymbol{y}^{\top}\boldsymbol{A}\boldsymbol{y}),
			$
			where $C>0$ and $\boldsymbol{A}\succ0$.
			
			\item Semi-heavy-tailed:
			$
			|\Phi_{\boldsymbol{X}_T}(\boldsymbol{y}+\mathrm{i}\boldsymbol{R})|
			\leq
			C\exp(-\gamma\sqrt{\boldsymbol{y}^{\top}\boldsymbol{A}\boldsymbol{y}}),
			$
			where $C>0$, $\boldsymbol{A}\succ0$, and $\gamma>0$.
			
			\item Heavy-tailed:
			$
			|\Phi_{\boldsymbol{X}_T}(\boldsymbol{y}+\mathrm{i}\boldsymbol{R})|
			\leq
			C(\boldsymbol{y}^{\top}\boldsymbol{A}\boldsymbol{y})^{-\gamma},
			$
			where $C>0$, $\boldsymbol{A}\succ0$, and $\gamma>d/2$.
		\end{itemize}
	}
	Table~\ref{table:chf_table} presents examples of the characteristic functions considered in this work.
	\item We select the parameters of $\psi(\cdot)$ to control the boundary growth of the transformed integrand in \eqref{eq:transformed_fourier_pricing_integral}. We separately derive the boundary growth conditions for three examples of pricing models, which are summarized in Tables \ref{tab:univariate_dom_tranf} and \ref{tab:multivariate_dom_tranf} (for more details, see Sections~\ref{sec:gbm_dom_transf}, \ref{sec:vg_dom_transf}, and \ref{sec:nig_dom_transf}).
\end{itemize}
Once the choice of $\psi(\cdot)$ (respectively $\Psi^{-1}(\cdot)$) is determined, the RQMC estimator of \eqref{eq:transformed_fourier_pricing_integral} can be expressed as follows:
$$
Q^{RQMC}_{N,S,d}[\tilde{g}] := \frac{1}{S} \sum_{s= 1}^S  \frac{1}{N}\sum_{n = 1}^N \tilde{g}( \boldsymbol{u}^{(s)}_n).
$$

To illustrate the importance of efficiently designing the domain transformation procedure, we show in Figure~\ref{fig:domain_transformation_2D_illustration} an example of a 1D put option (i.e., $P(X_T) = \max( K - e^{X_T}, 0)$) under the GBM model, and  compare the  original integrand on $\mathbb{R}$ to the transformed integrand on $[0,1]$ for different values of the scale parameter, $\tilde{\sigma}$, of a Gaussian proposal density, i.e., $\psi(y) =   \frac{1}{ \sqrt {2 \tilde{\sigma}^2 }} \exp(- \frac{y^2}{2\tilde{\sigma}^2})$. Figure~\ref{singular_integrand} indicates that, for $\tilde{\sigma} = 1$, the integrand is singular near the boundary, whereas the original integrand in Figure~\ref{original_integrand} is bounded and decays to zero at infinity. Moreover, for values of $\tilde{\sigma}  \geq 5$, the integrand decays to zero, with the rate of decay being faster the larger the value of $\tilde{\sigma}$. Section~\ref{sec:gbm_dom_transf} clarifies the reason for this decay and provides a strategy to control it.


\FloatBarrier
\begin{figure}[h!]
	\centering	
	\begin{subfigure}{0.4\textwidth}
		\includegraphics[width=\linewidth]{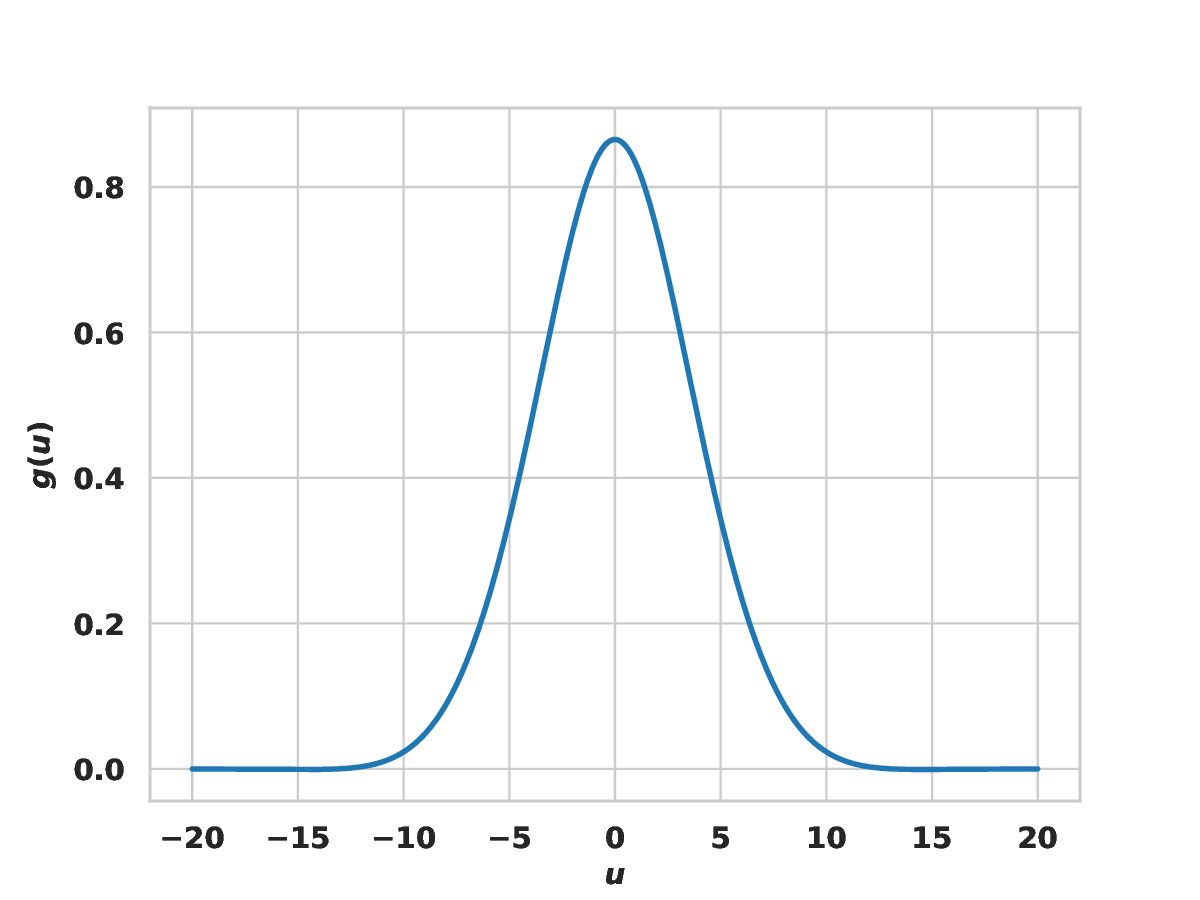}
		\caption{Original integrand in \eqref{g_integrand}}
		\label{original_integrand}
	\end{subfigure}
	\begin{subfigure}{0.4\textwidth}
		\includegraphics[width=\linewidth]{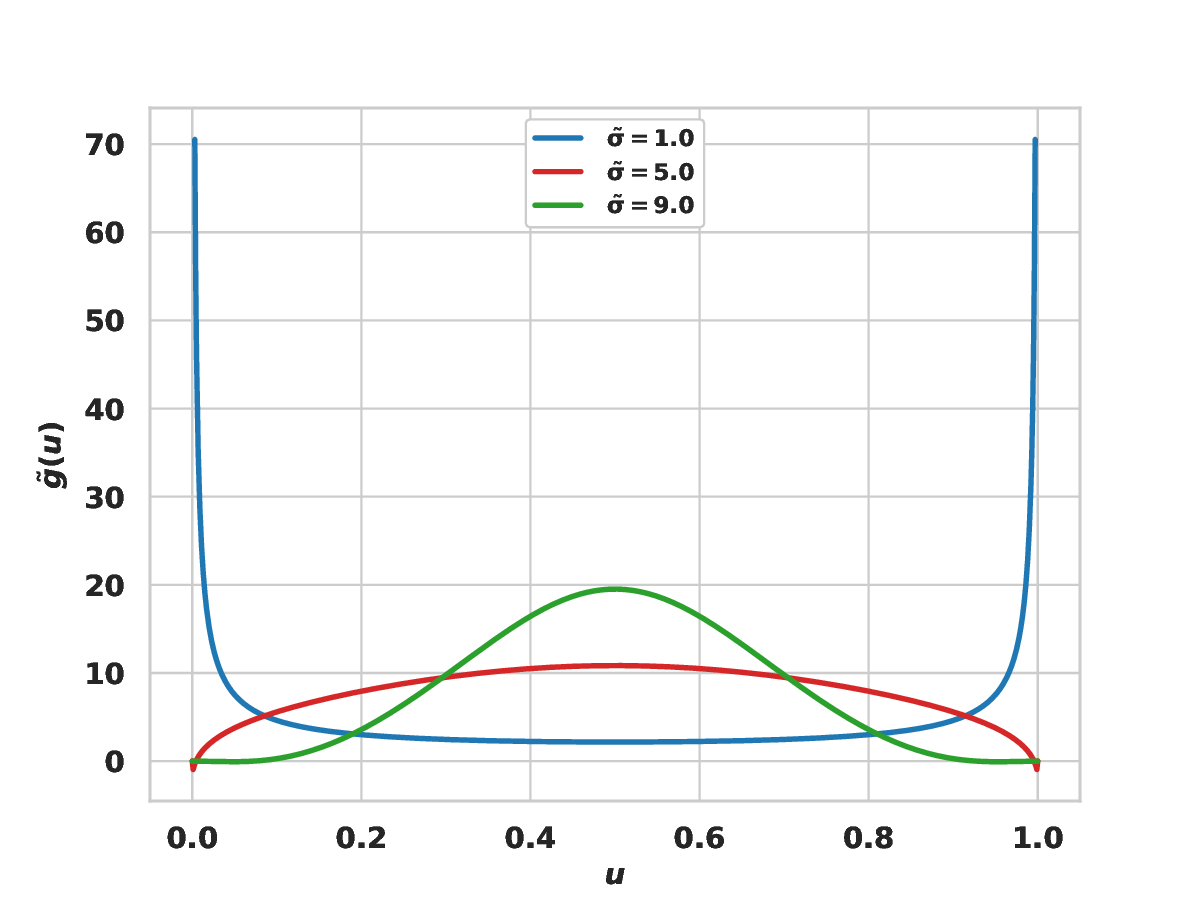}
		\caption{Transformed integrand in \eqref{eq:transformed_fourier_pricing_integral}}
		\label{singular_integrand}
	\end{subfigure}
	\caption[]{Effect of the domain transformation on the smoothness of the transformed integrand \eqref{eq:transformed_fourier_pricing_integral} for a 1D put option under GBM with volatility $\sigma = 0.2$. Gaussian density: $\psi(\cdot;\tilde{\mu}, \tilde{\sigma})$, mean: $\tilde{\mu} = 0$, scale: $\tilde{\sigma}$. The used parameters are: $K = S_0 = 100, r = 0, T = 1, R = 6.58$, }  
	\label{fig:domain_transformation_2D_illustration}
\end{figure}
\FloatBarrier

\subsection{Model-dependent Domain Transformation}
\label{sec:model_spec_dom_transf}
In this section, we develop a general domain transformation strategy to apply the RQMC method in the Fourier domain to a wide range of pricing models. We derive the boundary growth conditions on the transformed integrand for models with different classes of decay of the characteristic functions, namely, light-tailed, semi-heavy-tailed, and heavy-tailed, for the case of independent and dependent RVs, respectively.

\FloatBarrier
\subsubsection[Light-tailed characteristic functions: illustration for the GBM model]{Domain transformation for light-tailed characteristic functions: illustration for the GBM model} 
\label{sec:gbm_dom_transf}

\paragraph{Product-form domain transformation}
To simplify the first analysis, we consider the case in which the asset price processes are independent. Consequently, the characteristic function of the pricing model can be written as the product of univariate characteristic functions; hence, 
$$
	\phi^{GBM}_{\boldsymbol{X}_T}(\boldsymbol{z}) = \prod_{j = 1}^d \phi^{GBM}_{X^j_T}(z_j), \boldsymbol{z}\in \mathbb{C}^d, \Im[\boldsymbol{z}] \in \delta^{GBM}_X,
$$
where
$$
	\phi^{GBM}_{X^j_T}(z_j) :=  \operatorname{exp}\left(- \frac{\sigma_j^2T}{2}z_j^2\right), z_j \in \mathbb{C}, \Im[z_j] \in \delta^{GBM}_X.
$$
Because $\phi^{GBM}_{X^j_T}(z_j)$ is a Gaussian function, it is natural to consider a Gaussian domain transformation density in the form of $\psi^{nor}(\boldsymbol{y}) = \prod_{j = 1}^d \psi^{nor}_j(y_j)$, where
$$
\psi^{nor}_j(y_j) := \frac{ \exp(- \frac{y_j^2}{2\tilde{\sigma}^2_j})}{ \sqrt {2\pi \tilde{\sigma}_j^2 }}, \quad  y_j \in \mathbb{R}, \; \tilde{\sigma}_j > 0.
$$
After specifying the functional form of $\psi^{nor}_j(\cdot)$, the aim is to determine an appropriate choice of the parameters $\{\tilde{\sigma}_j\}_{j = 1}^d$. The function $r^{GBM}_{nor,j}(\cdot)$ is defined as the ratio of the characteristic function of the variable $X^j_T$ and the proposed density $\psi^{nor}_j(\cdot)$:
$$
	r_{nor,j}^{GBM}(\Psi_{nor}^{-1}(u_j)) :=  \frac{ \phi^{GBM}_{X^j_T}\left(\Psi_{nor}^{-1}(u_j)+  \mathrm{i} R_j\right)}{	\psi^{nor}_j\left(\Psi_{nor}^{-1}(u_j)\right)}, \quad u_j \in [0,1], \boldsymbol{R} \in \delta^{GBM}_V.
$$
The parameters $\{\tilde{\sigma}_j\}_{j = 1}^d$ should be set to control the growth of the function $	r_{nor,j}^{GBM}(\cdot)$ near the boundary of $[0,1]$ as follows:
$$
	\lim_{u_j \to \{0,1\} } \mid	r_{nor,j}^{GBM}(\Psi_{nor}^{-1}(u_j)) \mid < \infty , \; \forall \; j \in 	\mathbb{I}_d. 
$$

We recall that by symmetry of the integrand $g(\cdot)$ and the proposal density $\psi^{nor}(\cdot)$ around the origin, it is sufficient to study the behavior of the transformed integrand as $u_j \to 0$ to ensure that it is also controlled for   $u_j \to 1$ for all $j \in \mathbb{I}_d$, and hence it is well-behaved on all the $2d$ faces of $[0,1]^d$.  To determine the suitable range of parameters $\{\tilde{\sigma}_j\}_{j = 1}^d$, we replace the characteristic function and the proposed density with their explicit expressions. Thus, $	r_{nor,j}^{GBM}(\cdot)$ can be written as follows:
$$
		\small
	\hspace{-1.5cm}
	\begin{aligned}
		\label{univariate_gbm:_ratio}
	r_{nor,j}^{GBM}(\Psi_{nor}^{-1}(u_j)) &=  \operatorname{exp}\left( - \frac{T\sigma_j^2}{2}\left(\Psi_{nor}^{-1}(u_j) +  \mathrm{i}R_j  \right)^2    \right)  \times  \sqrt{2\pi \tilde{\sigma}_j^2 }\operatorname{exp}\left(\frac{(\Psi_{nor}^{-1}(u_j))^2}{2 \tilde{\sigma}_j^2}  \right)  \\
		& =  \newline \underbrace{ \sqrt{2\pi } \operatorname{exp}\left(- \mathrm{i} \Psi_{nor}^{-1}(u_j) T \sigma_j^2 R_j + \frac{T  \sigma_j^2 R_j^2}{2}   \right) }_{:= h^{GBM}_{nor,1}(\Psi_{nor}^{-1}(u_j))}  \\ &
		 \times  	 \underbrace{\tilde{\sigma}_j\operatorname{exp}\left( -(\Psi_{nor}^{-1}(u_j))^2 \left(\frac{T\sigma_j^2}{2} - \frac{1}{2 \tilde{\sigma}_j^2}\right) \right)}_{:= h^{GBM}_{nor,2}(\Psi_{nor}^{-1}(u_j))},
	\end{aligned}
$$
where the function $h^{GBM}_{nor,1}(\Psi_{nor}^{-1}(u_j))$ is bounded  for all $ u_j \in [0,1]$, and the function $h^{GBM}_{nor,2}(\Psi_{nor}^{-1}(u_j))$ determines the growth of the integrand at the boundary of $[0,1]$. Depending on the values of $\{\tilde{\sigma}_j\}_{j = 1}^d$, we enumerate three possible cases 
\begin{equation}
	\label{eq:gbm_conditions}
	\lim_{u_j \to 0} h^{GBM}_{nor,2}(\Psi_{nor}^{-1}(u_j)) =
	\left\{
	\begin{array}{lll}
		+\infty & \text{if } \tilde{\sigma}_j < \frac{1}{\sqrt{T} \sigma_j} & (i), \\[0.5pt]
		\tilde{\sigma}_j & \text{if } \tilde{\sigma}_j = \frac{1}{\sqrt{T} \sigma_j} & (ii), \\[0.5pt]
		0 & \text{if } \tilde{\sigma}_j > \frac{1}{\sqrt{T} \sigma_j} & (iii).
	\end{array}
	\right.
\end{equation}

To summarize, a suitable choice for $\tilde{\sigma}_j$ is $\tilde{\sigma}_j = \overline{\sigma}_j + \epsilon_j$, where $\overline{\sigma}_j = \frac{1}{ \sqrt{T} \sigma_j}$ defines the critical value, and $\epsilon_j \geq 0$. However, different values of $\tilde{\sigma}_j$ satisfying (ii) and (iii) may lead to differing error rates of the RQMC method, as demonstrated in Section~\ref{sec: num_exp_results}. Although a higher value for $\epsilon_j$ accelerates the integrand decay, selecting an arbitrarily large $\epsilon_j$ is not advisable because it amplifies the integrand peak around the origin and hence augments the magnitude of the mixed first partial derivatives of the integrand. These factors substantially influence the performance of RQMC, as illustrated in Section~\ref{sec: num_exp_results}. Moreover, in Case~(ii) in \eqref{eq:gbm_conditions}, the dominant term in the characteristic function vanishes, and the integrand decays at the rate of the payoff transform.

The previous derivation of the rule for the domain transformation relies on the assumption of the independence of assets. In this simplified framework, the performance of RQMC is classically studied; however, in practical applications, the variables may be correlated. Figure~\ref{fig:gbm_correlation_effect} demonstrates that, when the assets are positively correlated, the proposed transformation  must be generalized to account for the correlation parameters; otherwise, the boundary growth conditions are violated, and the performance of RQMC significantly deteriorates. In the uncorrelated case, a rate of convergence of $\mathcal{O}(N^{-1.48})$ was numerically estimated. In contrast, in the positively correlated case ($\rho = 0.7$), the rate of RQMC is substantially worse, $\mathcal{O}(N^{-0.69})$, and the size of the implied constant in the error estimate is also significantly larger than in the uncorrelated case.
In the second part of this section, we explain the shortcoming of the proposed rule for the domain transformation when the assets are positively correlated. Moreover, we propose a more general  domain transformation rule which accounts for the dependence between the underlying assets using nonlinear matrix inequalities.  Figure~\ref{fig:gbm_multivariate_vs_univariate_rule} illustrates that by generalizing the domain transformation rule to the case of dependent assets, the convergence rate of RQMC is significantly improved in the correlated setting.
\vspace{-0.2cm}
\FloatBarrier
\begin{figure}[h!]
	\centering	
	\begin{subfigure}{0.5\textwidth}
		\includegraphics[width=\linewidth]{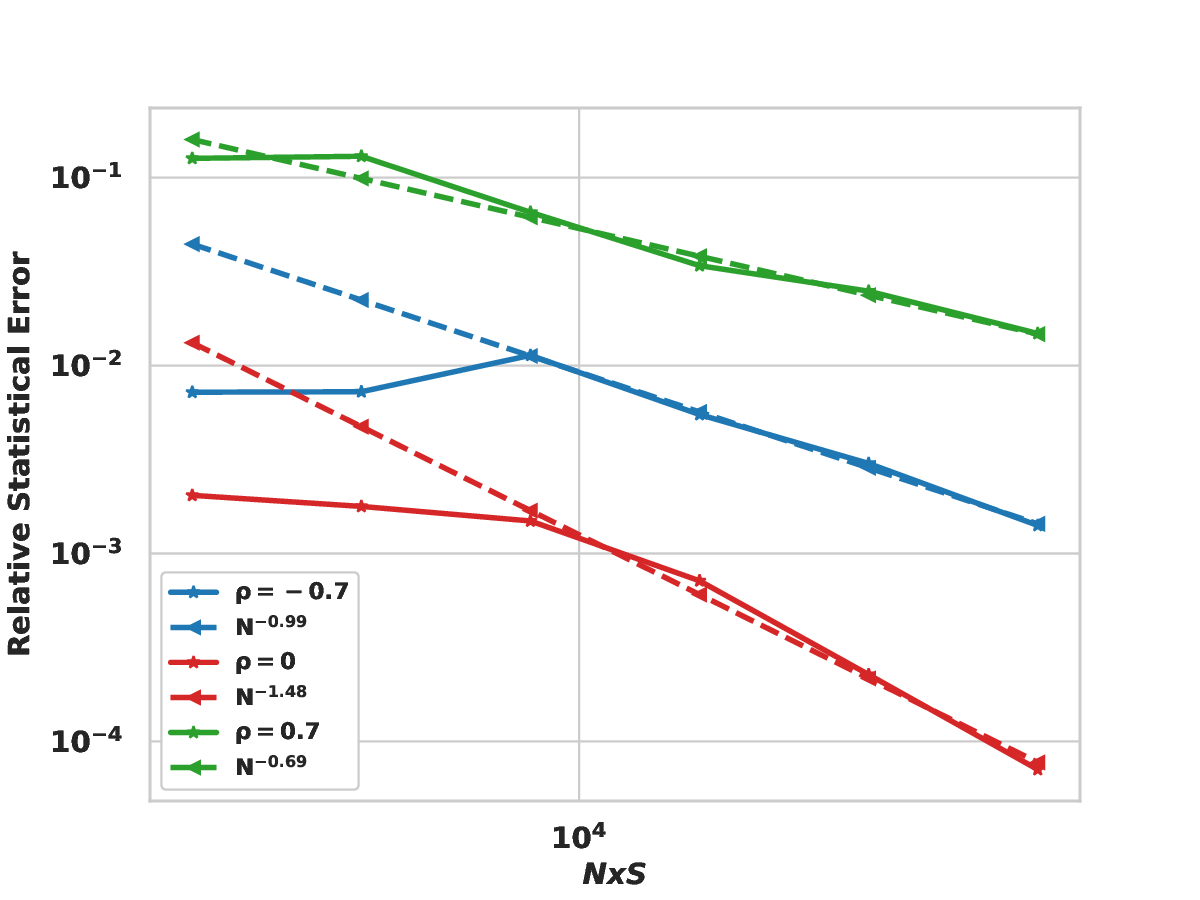}
	\end{subfigure}
	\caption[]{Effect of the correlation parameter, $\rho$, on the convergence of RQMC for a two-dimensional call on the minimum option under the GBM model with $S_0^j = 100$, $K = 100$, $r = 0$, $T= 1$, $\sigma_j = 0.2$, with $\boldsymbol{\Sigma}_{ij} =   \rho\sigma_i \sigma_j $ for $i,j= 1,2$, $i \neq j$,  $\boldsymbol{\Sigma}_{ii}  = \sigma_i^2$. For the domain transformation, $\tilde{\sigma}_j = \frac{1}{\sqrt{T} \sigma_j} = 5$, where $j =1,2$. $N$: number of QMC points; $S = 30$: number of digital shifts. 
	} 
	\label{fig:gbm_correlation_effect}
\end{figure}
\FloatBarrier
\begin{remark}
	The presented convergence rates of the RQMC method are obtained numerically and may not correspond to the theoretically expected rates which are asymptotic \cite{liu2023nonasymptotic}. In general, the number of QMC points needed to achieve the asymptotic regime is problem-dependent. In our framework, the domain transformation  can have a significant impact on the regularity of the transformed integrand, and hence the number of  points to reach the asymptotic regime depends on the used transformation and its parameters. 
\end{remark}

\paragraph{Non-factorizable domain transformation}
For dependent assets, the joint characteristic function cannot be factorized into the product of univariate characteristic functions and  we have that
\begin{equation}
\phi^{GBM}_{\boldsymbol{X}_T} (\boldsymbol{z})	:=  \operatorname{exp}\left( - \frac{T}{2} \boldsymbol{z}^{\top} \boldsymbol{\Sigma} \boldsymbol{z}   \right), \boldsymbol{z} \in \mathbb{C}^d, \Im[\boldsymbol{z}] \in \delta^{GBM}_X.
\end{equation}
Consequently, we select the proposal density $\psi^{nor}(\cdot)$ corresponding to the multivariate normal PDF, given by
$$
	\label{eq:gbm_multivariate_density}
	\psi^{nor}(\boldsymbol{y}) = (2\pi)^{-\frac{d}{2}}(\operatorname{det}(\boldsymbol{ \tilde{\Sigma}}))^{-\frac{1}{2}} \exp\left(- \frac{1}{2}(\mathbf{y}^{\top} \boldsymbol{ \tilde{\Sigma}}^{-1} \mathbf{y} ) \right), \;\boldsymbol{y} \in \mathbb{R}^d, \: \boldsymbol{ \tilde{\Sigma}} \succeq 0
$$
 In this case, the singularity of the integrand is controlled by the function, $r_{nor}^{GBM}(\boldsymbol{y})$, defined as
\begin{equation}
	\label{eq:r_gbm_multivar}
	r_{nor}^{GBM}(\boldsymbol{y}) := \frac{ \phi_{\boldsymbol{X}_T} (\boldsymbol{y}+ \mathrm{i} \boldsymbol{R}) }{\psi^{nor}( \boldsymbol{y}) }, \quad  \boldsymbol{y}\in \mathbb{R}^d,  \boldsymbol{R} \in \delta^{GBM}_V.
\end{equation}
\begin{remark}
		Analyzing the boundary growth of the transformed integrand as $u_j \to 0$ or $u_j \to 1$ is equivalent to analyzing \eqref{eq:r_gbm_multivar} as $|y_j| \to \infty$, due to the positive definiteness of $\boldsymbol{\Sigma}$ and $\boldsymbol{\tilde{\Sigma}}^{-1}$.
	\end{remark}

By substituting in the explicit expressions of the characteristic function and the proposal density, we obtain the following expression for $r_{nor}^{GBM}(\boldsymbol{y})$:

\begin{equation}
	\label{eq:ratio_mgbm}
	\small
	\begin{aligned}
	r_{nor}^{GBM}(\boldsymbol{y}) &= \underbrace{ (2\pi)^{\frac{d}{2}}  \operatorname{exp}\left( - \mathrm{i} T \boldsymbol{R}^{\top} \boldsymbol{\Sigma}\boldsymbol{y}+ \frac{T}{2} \boldsymbol{R}^{\top} \boldsymbol{\Sigma} \boldsymbol{R} \right)}_{:=h^{GBM}_{nor,1}(\boldsymbol{y})} \\
		&\times 	 \underbrace{(\operatorname{det}(\boldsymbol{\tilde{\Sigma}}))^{\frac{1}{2}}  \exp\left(-\frac{1}{2}(\boldsymbol{y}^{\top} \left( T \boldsymbol{\Sigma} - \boldsymbol{\tilde{\Sigma}}^{-1}  \right) \boldsymbol{y}  \right )}_{:=h_{nor,2}^{GBM}(\boldsymbol{y})} .
	\end{aligned}
\end{equation}
From \eqref{eq:ratio_mgbm}, the function $h^{GBM}_{nor,1}(\boldsymbol{y})$ is bounded for all $\boldsymbol{y} \in \mathbb{R}^d$; hence, the part controlling the boundary growth of the integrand is given by $h_{nor,2}^{GBM}(\boldsymbol{y})$. Similarly to \eqref{eq:gbm_conditions}, we enumerate three possible limits, depending on the choice of $\boldsymbol{\tilde{\Sigma}}$:
\begin{equation}
	\label{eq:mgbm_conditions}
	\lim_{|y_j| \to \infty} h_{nor,2}^{GBM}(\boldsymbol{y}) =
	\left\{
	\begin{array}{lll}
		+\infty & \text{if } \boldsymbol{\Sigma} - \frac{1}{T} \boldsymbol{\tilde{\Sigma}}^{-1} \prec 0 & (i), \\[0.5pt]
		(\operatorname{det}(\boldsymbol{\tilde{\Sigma}}))^{\frac{1}{2}} 
		& \text{if } \boldsymbol{\tilde{\Sigma}} = \frac{1}{T} \boldsymbol{\Sigma}^{-1} & (ii), \\[0.5pt]
		0 & \text{if } \boldsymbol{\Sigma} - \frac{1}{T} \boldsymbol{\tilde{\Sigma}}^{-1} \succ 0 & (iii).
	\end{array}
	\right.
\end{equation}

From \eqref{eq:mgbm_conditions}, a sufficient choice of the matrix $\boldsymbol{\tilde{\Sigma}} $ satisfies either Condition (ii) or (iii).  Furthermore, the transformed integrand is multiplied by the factor $(\operatorname{det}(\boldsymbol{\tilde{\Sigma}}))^{\frac{1}{2}}$; thus, the aim is to select a matrix $\boldsymbol{\tilde{\Sigma}}$ to satisfy Condition (iii), with the minimum possible determinant to avoid high peaks of the integrand around the origin, which was motivated by previous findings (see \cite{bayer2023optimal}). Optimally, the choice of $\boldsymbol{\tilde{\Sigma}}$ is given by the following constrained optimization problem:
\begin{equation}
	\label{eq:param_opt_prob}
	\begin{aligned}
		&\boldsymbol{\tilde{\Sigma}}^{*} =\argmin_{	\boldsymbol{\tilde{\Sigma}}} \det( \mathbf{\tilde{\Sigma}}) \\
		&	\text{s.t} \;  \boldsymbol{\Sigma}-\frac{1}{T} \boldsymbol{\tilde{\Sigma}}^{-1} \succeq 0
	\end{aligned}
\end{equation}
Instead, we propose a simpler construction of the matrix $\boldsymbol{\tilde{\Sigma}}$. We have that the matrix $\boldsymbol{\Sigma}$ is real symmetric; thus, by the spectral theorem, it has an eigenvalue decomposition (EVD) (i.e., $\boldsymbol{\Sigma} = \boldsymbol{P} \boldsymbol{D} \boldsymbol{P}^{-1} $, with $\boldsymbol{D} = \diag(\lambda_1,\ldots, \lambda_d)$ and $\lambda_j >0$ for all $ j \in \mathbb{I}_d$ because $\boldsymbol{\tilde{\Sigma}} \succ 0$. To simplify the problem in \eqref{eq:param_opt_prob}, we choose the matrix $\boldsymbol{\tilde{\Sigma}}$ to be in the form of $\boldsymbol{\tilde{\Sigma}} : = \boldsymbol{P} \boldsymbol{\tilde{D}} \boldsymbol{P}^{-1} $, where $\boldsymbol{\tilde{D}} = \diag(\tilde{\lambda}_1,\ldots,\tilde{\lambda}_d)$. In this case, we can express the constraint in \eqref{eq:param_opt_prob} as follows:
$$
\boldsymbol{P} (\boldsymbol{D} -  \frac{1}{T} \boldsymbol{\tilde{D}}^{-1} )  \boldsymbol{P}^{-1} \succeq 0 \iff  \tilde{\lambda}_j  \geq \frac{1}{ \lambda_j T}, \; j \in \mathbb{I}_d
$$
Hence, a suboptimal choice for $\boldsymbol{\tilde{\Sigma}} $ is $\boldsymbol{\tilde{\Sigma}} = \boldsymbol{P} \diag(\frac{1}{\lambda_1 T}, \ldots, \frac{1}{\lambda_d T}) \boldsymbol{P}^{-1} = \frac{1}{T} \boldsymbol{\Sigma}^{-1}$, which is the choice adopted in this work.
\begin{remark}
	The conditions we provided represent  sufficient boundary conditions on $\boldsymbol{\tilde{\Sigma}}$ to ensure the boundedness of the transformed integrand. In other words, enforcing the semi-positive definiteness in the matrix inequality is sufficient, but not necessary. In fact, the boundedness can be achieved  even if the difference between the matrices is indefinite as long as they satisfy the following
	\begin{equation}
		\label{eq:necessary_condition}
	\boldsymbol{y}^{\top}\left( \boldsymbol{\Sigma}-\frac{1}{T} \boldsymbol{\tilde{\Sigma}}^{-1}\right)  \boldsymbol{y} = \sum_{i=1}^d \sum_{j=1}^d y_i y_j\left(\boldsymbol{\Sigma}_{i j}-\frac{1}{T} \boldsymbol{\tilde{\Sigma}}_{i j}^{-1}\right) \geq  0,  \quad |y_j| \to \infty.
	\end{equation}
			\eqref{eq:necessary_condition} can be rewritten in a matrix form in terms of the EVD of $\boldsymbol{\tilde{\Sigma}}$, $\boldsymbol{\tilde{\Sigma}} = \boldsymbol{\tilde{P}} \boldsymbol{\tilde{D}} \boldsymbol{\tilde{P}}^{-1}$, by
			$$\sum_{j=1}^d \lambda_j\left(\boldsymbol{P}^{\top} \boldsymbol{y}\right)_j^2-\frac{1}{T} \sum_{j=1}^d \frac{1}{\tilde{\lambda}_j}\left(\boldsymbol{\tilde{P}}^{\top} \boldsymbol{y}\right)_j^2	\geq 0, \quad  \quad |y_j| \to \infty	$$		
\end{remark}

In order to apply the ICDF mapping to $[0,1]^d$, we first need to decouple the dependencies between the components of the multivariate normal distribution $\boldsymbol{Y} \sim \mathcal{N}_d(\mathbf{0}, \boldsymbol{\tilde{\Sigma}})$ with covariance matrix $\boldsymbol{\tilde{\Sigma}}$. In fact,  $\boldsymbol{Y}$ can be represented as $\boldsymbol{Y} = \boldsymbol{\tilde{L}} \boldsymbol{Z}$, where $\boldsymbol{\tilde{L}} \boldsymbol{\tilde{L}}^{\top} = \boldsymbol{\tilde{\Sigma}}$ is the Cholesky decomposition of $\boldsymbol{\tilde{\Sigma}}$, and $\boldsymbol{Z} \sim \mathcal{N}_d(\mathbf{0}, \boldsymbol{I_d})$ follows a multivariate standard  normal distribution. Alternatively, one may employ the eigenvalue decomposition (EVD) of the covariance matrix, expressed as $\boldsymbol{\tilde{\Sigma}} = \boldsymbol{P} \boldsymbol{D} \boldsymbol{P}^{-1}$, and define $\boldsymbol{\tilde{L}} = \boldsymbol{P} \boldsymbol{D}^{1/2}$. The procedure consists of two steps. First, a change of variables is applied to transform the original random vector as $\boldsymbol{y}' =  \boldsymbol{\tilde{L}}^{-1} \boldsymbol{y}$, thereby eliminating the dependence between the RVs. Second, a domain transformation is performed using the ICDF of the standard normal distribution: $\boldsymbol{y}' = \Psi_{nor}^{-1}(\boldsymbol{u};\boldsymbol{I_d})$. This results in the following expression:
	\begin{equation}
		\int_{\mathbb{R}^d} g(\boldsymbol{y}) \, \mathrm{d}\boldsymbol{y} = \int_{[0,1]^d} \frac{g\left( \boldsymbol{\tilde{L}}  \Psi_{nor}^{-1}(\boldsymbol{u}; \boldsymbol{I_d})  \right)} {\psi^{nor}\left( \boldsymbol{\tilde{L}}  \Psi_{nor}^{-1}(\boldsymbol{u}; \boldsymbol{I_d})  \right)}\, \mathrm{d}\boldsymbol{u},
	\end{equation}
	where $g(\cdot)$ is defined in Equation  \eqref{g_integrand}.

Figure~\ref{fig:gbm_multivariate_vs_univariate_rule} reveals the importance of using the multivariate rule for the domain transformation (i.e., $\boldsymbol{\tilde{\Sigma}} = \frac{1}{T} \boldsymbol{\Sigma}^{-1}$) compared to the univariate transformation rule (i.e, $\tilde{\sigma}_j = \frac{1}{\sqrt{T} \sigma_j}, j = 1,\ldots d$), when assets are correlated. The multivariate rule results in significantly superior convergence  behavior of RQMC.
\FloatBarrier
\begin{figure}[h!]
	\centering	
	\begin{subfigure}{0.5\textwidth}		\includegraphics[width=\linewidth]{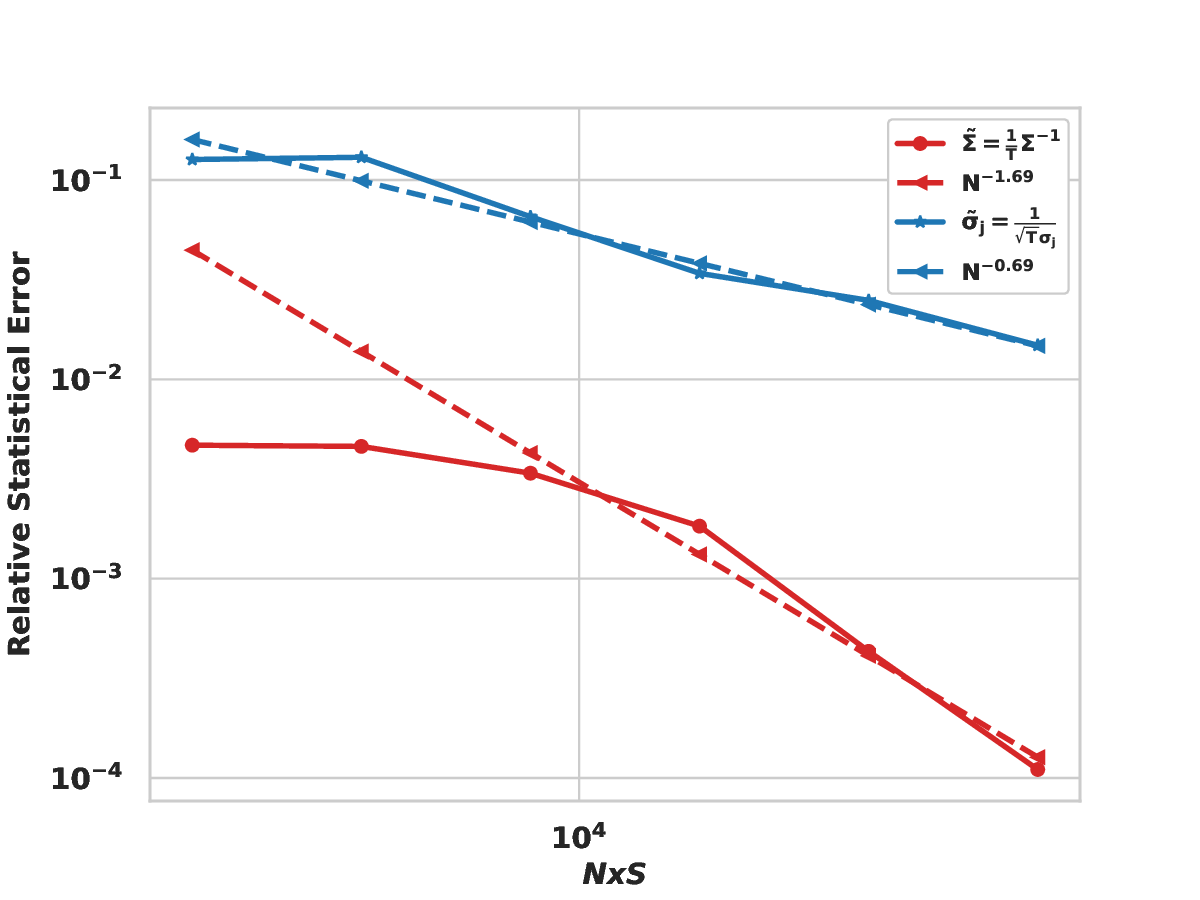}		\caption{}	\end{subfigure}
	\caption[NIG: Effect of the correlation $\rho$ $\tilde{\sigma}$ on (a) the shape of the transformed integrand, $\tilde{g}(u)$,  and (b) convergence of the relative statistical error of QMC]{Convergence of the RQMC error for multivariate and univariate domain transormations in the case of a 2D call on the minimum option under the GBM model with $S_0^j = 100$, $K = 100$, $r = 0$, $T= 1$, $\sigma_j = 0.2, \rho = 0.7$, with $\boldsymbol{\Sigma}_{ij} =   \rho\sigma_i \sigma_j $ for $i,j= 1,2$, $i \neq j$,  $\boldsymbol{\Sigma}_{ii}  = \sigma_i^2$. $N$: number of QMC points; $S = 30$: number of digital shifts. 
	 } 
	\label{fig:gbm_multivariate_vs_univariate_rule}
\end{figure}
\FloatBarrier
\begin{remark}
	Alternative approaches to deal with the multivariate ICDF can rely on the Rosenblatt transformation, copula theory or normalizing flow mappings \cite{liu2024transport}. Investigating the efficiency of these alternatives is left for future work. 
\end{remark}

\subsubsection[Semi-heavy-tailed characteristic functions: illustration for the GH model]{Domain transformation for  semi-heavy-tailed characteristic functions: illustration for the GH model}
\label{sec:nig_dom_transf}
For the product-form domain transformation (independent assets), we refer  to Appendix \ref{sec:domain_transf_gh_independent}.
\paragraph{Non-factorizable domain transformation} 
The multivariate characteristic function of the GH model is defined for $\boldsymbol{z} \in \mathbb{C}^d, \Im[\boldsymbol{z}] \in \delta_X^{GH}$, by \cite{prause1999generalized}
\begin{equation}
	\label{eq:gh_cf}
	\phi^{GH}_{ \boldsymbol{X}_T}(\boldsymbol{z}) = \left(\frac{\alpha^2-\boldsymbol{\beta}^{\top} \boldsymbol{\Delta} \boldsymbol{\beta}}{\alpha^2  -(\boldsymbol{\beta}+ \mathrm{i} \boldsymbol{z})^{\top}\boldsymbol{\Delta}(\boldsymbol{\beta}+ \mathrm{i} \boldsymbol{z})}\right)^{\lambda / 2} \frac{\mathrm{~K}_\lambda\left(\delta T \sqrt{\alpha^2-(\boldsymbol{\beta}+ \mathrm{i} \boldsymbol{z})^{\top}\boldsymbol{\Delta}(\boldsymbol{\beta}+ \mathrm{i} \boldsymbol{z})}\right)}{\mathrm{K}_\lambda\left(\delta T \sqrt{\alpha^2-\boldsymbol{\beta}^{\top} \boldsymbol{\Delta} \boldsymbol{\beta}}\right)}
\end{equation}
where  $\mathrm{K}_\lambda(\cdot)$ is a modified Bessel function of the second kind and $\lambda \in \mathbb{R}$. We note that $\mathrm{K}_{- 1/2}(x) = \sqrt{\frac{\pi}{2x}}e^{-x}$, hence for $\lambda = - \frac{1}{2}$, we recover the characteristic function of the NIG model given in Table \ref{table:chf_table}, and the analysis provided in this section remains valid for any $\lambda \in \mathbb{R}$.

We choose a density that matches the functional form of the characteristic function \eqref{eq:gh_cf}, corresponding to the multivariate Laplace distribution, which  is given by
\begin{equation}
	\label{eq:mlap_pdf}
	\psi_{lap}(\boldsymbol{y})=\frac{2}{(2 \pi)^{\frac{d}{2}}  \sqrt{\det( \mathbf{\tilde{\Sigma}}) }}\left(\frac{\boldsymbol{y}^{\top} \boldsymbol{\tilde{\Sigma}}^{-1} \boldsymbol{y}}{2}\right)^{ \frac{v}{2}} K_{v}\left(\sqrt{2 \boldsymbol{y}^{\top} \boldsymbol{\tilde{\Sigma}}^{-1} \boldsymbol{y}}\right), \quad \boldsymbol{y} \in \mathbb{R}^d, \boldsymbol{\tilde{\Sigma}} \succeq 0, 
\end{equation}
with $v = \frac{2  -  d}{2}$.
First, we note that the Bessel function satisfies the following property \cite{prause1999generalized}:
\begin{equation}
	\label{eq:bessel_asymtotic}
	K_{\lambda}(x) \stackrel{x\to +\infty}{\sim} \sqrt{\frac{\pi}{2x}} e^{-x}.
\end{equation}
Also, we have that for $\boldsymbol{z} = \boldsymbol{y} + \mathrm{i } \boldsymbol{R}$, with $\boldsymbol{y}, \boldsymbol{R} \in \mathbb{R}^d$,
\begin{equation}
	\label{eq:gh_polynomial_asymptotic}
	\alpha^2 -  (\boldsymbol{\beta}  + \mathrm{i}\boldsymbol{y}  - \boldsymbol{R})^{\top}
	\boldsymbol{\Delta}  (\boldsymbol{\beta}  + \mathrm{i}\boldsymbol{y}  - \boldsymbol{R}) \stackrel{|y_j| \to \infty}{\sim}  \boldsymbol{y}^{\top}  	\boldsymbol{\Delta}   \boldsymbol{y}.
\end{equation}
Using both approximations in \eqref{eq:bessel_asymtotic} and \eqref{eq:gh_polynomial_asymptotic}, we can approximate the characteristic function of the GH model as follows:
\begin{equation}
	\label{eq:gh_asymp_approx}
	\begin{aligned}
		\phi^{GH}_{ \boldsymbol{X}_T}(\boldsymbol{y} + \mathrm{i}  \boldsymbol{R}) \stackrel{|y_j| \to \infty}{\sim}  \left(\frac{\alpha^2 - \boldsymbol{\beta}^{\top} \boldsymbol{\Delta } \boldsymbol{\beta}}{\boldsymbol{y}^{\top}  	\boldsymbol{\Delta}   \boldsymbol{y}} \right)^{\lambda /2 } \sqrt{\frac{\pi}{2 \delta T \sqrt{\boldsymbol{y}^{\top}  	\boldsymbol{\Delta}   \boldsymbol{y}}}} \frac{\exp\left(-  \delta T \sqrt{\boldsymbol{y}^{\top}  	\boldsymbol{\Delta}   \boldsymbol{y}}\right)}{K_{\lambda}\left(\delta T \sqrt{\alpha^2 - \boldsymbol{\beta}^{\top} \boldsymbol{\Delta } \boldsymbol{\beta}}\right) }.
	\end{aligned}
\end{equation}
Similarly, the PDF of the multivariate Laplace distribution can be approximated  by
\begin{equation}
	\label{eq:mlap_asymp_approx}
	\psi_{lap}(\boldsymbol{y}) \stackrel{|y_j| \to \infty}{\sim} \frac{2 \left(\boldsymbol{y}^{\top} \boldsymbol{\tilde{\Sigma}}^{-1} \boldsymbol{y} \right)^{v/2} }{(2\pi)^{d/2} \sqrt{\det(\boldsymbol{\tilde{\Sigma}})}}   \sqrt{\frac{\pi}{2 \sqrt{2 \boldsymbol{y}^{\top}  	\boldsymbol{\tilde{\Sigma}}^{-1}  \boldsymbol{y}}}} \exp\left(- \sqrt{2\boldsymbol{y}^{\top}  	\boldsymbol{\tilde{\Sigma}}^{-1} \boldsymbol{y}} \right).
\end{equation}
Focusing on the leading asymptotic terms, we encapsulate the polynomial prefactors in the following notation
\begin{equation}
	Q_{GH}(\boldsymbol{y}) :=	 \left(\frac{\alpha^2 - \boldsymbol{\beta}^{\top} \boldsymbol{\Delta } \boldsymbol{\beta}}{\boldsymbol{y}^{\top}  	\boldsymbol{\Delta}   \boldsymbol{y}} \right)^{\lambda /2 } \sqrt{\frac{\pi}{2 \delta T \sqrt{\boldsymbol{y}^{\top}  	\boldsymbol{\Delta}   \boldsymbol{y}}}} \frac{1}{K_{\lambda}\left(\delta T \sqrt{\alpha^2 - \boldsymbol{\beta}^{\top} \boldsymbol{\Delta } \boldsymbol{\beta}}\right) }.
\end{equation}
\begin{equation}
	Q_{lap}(\boldsymbol{y}) := \frac{2 \left(\boldsymbol{y}^{\top} \boldsymbol{\tilde{\Sigma}}^{-1} \boldsymbol{y} \right)^{v/2} }{(2\pi)^{d/2}}   \sqrt{\frac{\pi}{2 \sqrt{2 \boldsymbol{y}^{\top}  	\boldsymbol{\tilde{\Sigma}}^{-1}  \boldsymbol{y}}}}.
\end{equation}
To determine a rule for the choice of $\boldsymbol{\tilde{\Sigma}}$, we concentrate on the tail behavior of the function $r_{lap}^{GH}(\boldsymbol{y})$:
\begin{equation}
	\label{eq:ratio_nigdep}
	r_{lap}^{GH}(\boldsymbol{y}) := \frac{	\phi^{GH}_{ \boldsymbol{X}_T}( \boldsymbol{y} + \mathrm{i} \boldsymbol{R})}{	\psi_{lap}( \boldsymbol{y} )}, \quad \boldsymbol{y} \in \mathbb{R}^d, \boldsymbol{R}  \in \delta_V^{GH}.
\end{equation}
With both the asymptotic approximations presented in \eqref{eq:gh_asymp_approx} and \eqref{eq:mlap_asymp_approx}, the boundary growth of $r_{lap}^{GH}(\boldsymbol{y})$ is controlled by the following term 
$$
\small
r^{GH}_{lap}(\boldsymbol{y}) :=  \frac{Q_{GH}( \boldsymbol{y} )}{Q_{lap}(\boldsymbol{y})} 
\underbrace{\sqrt{\operatorname{det}(\boldsymbol{\tilde{\Sigma}})} \exp\left(   - \left( \delta T \sqrt{   \boldsymbol{y}^{\top} \boldsymbol{\Delta}   \boldsymbol{y}} -  \sqrt{2\boldsymbol{y}^{\top} \boldsymbol{\tilde{\Sigma}}^{-1}\boldsymbol{y}  } \right)  \right)}_{h_{lap}^{GH}(\boldsymbol{y})}.
$$
The function $x \mapsto \sqrt{x}$ is monotonic; thus, 
$$
\sqrt{\delta^2 T^2 \boldsymbol{y}^{\top} \boldsymbol{\Delta} \boldsymbol{y}} - \sqrt{2 \boldsymbol{y}^{\top} \tilde{\boldsymbol{\Sigma}}^{-1} \boldsymbol{y}} \geq 0 
\Longleftrightarrow 
\boldsymbol{y}^{\top} \left( \delta^2 T^2 \boldsymbol{\Delta} - 2 \tilde{\boldsymbol{\Sigma}}^{-1} \right) \boldsymbol{y} \geq 0, \quad \boldsymbol{y} \in \mathbb{R}^d.
$$
Therefore, we enumerate three possible limits:
\begin{equation}
	\label{eq:GH_conditions}
	\lim_{|y_j| \to +\infty} h_{lap}^{GH}(\boldsymbol{y}) =
	\left\{
	\begin{array}{lll}
		+\infty & \text{if } \delta^2 T^2 \boldsymbol{\Delta} - 2 \boldsymbol{\tilde{\Sigma}}^{-1} \prec 0 & (i), \\[0.5pt]
		\sqrt{\operatorname{det}(\boldsymbol{\tilde{\Sigma}})} & \text{if } \boldsymbol{\tilde{\Sigma}} = \frac{2}{\delta^2 T^2} \boldsymbol{\Delta}^{-1} & (ii), \\[0.5pt]
		0 & \text{if } \delta^2 T^2 \boldsymbol{\Delta} - 2 \boldsymbol{\tilde{\Sigma}}^{-1} \succ 0 & (iii).
	\end{array}
	\right.
\end{equation}

From \eqref{eq:GH_conditions}, a sufficient condition to eliminate the singularity at the boundary is to set $\boldsymbol{\tilde{\Sigma}}$ such that the matrix $\delta^2 T^2 \boldsymbol{\Delta} - 2 \boldsymbol{\tilde{\Sigma}}^{-1} \succeq 0$. The problem of finding such a matrix does not have a unique solution; hence, we propose a candidate construction. The matrix $\boldsymbol{\Delta}$ is a real symmetric matrix; thus, via the spectral theorem, it has a principal value factorization (i.e., $\boldsymbol{\Delta} = \boldsymbol{P} \boldsymbol{D} \boldsymbol{P}^{-1} $), where $\boldsymbol{D} = \diag\left(\lambda_1,\ldots, \lambda_d \right)$ with $\lambda_j > 0, j =1,\ldots d$, because $\boldsymbol{\Delta}$ is positive-definite. Therefore, we propose the construction $\boldsymbol{\tilde{\Sigma}} := \boldsymbol{P} \boldsymbol{\tilde{D}} \boldsymbol{P}^{-1} $, where $\boldsymbol{\tilde{D}} = \diag(\tilde{\lambda}_1, \ldots ,\tilde{\lambda}_d)$ such that $\delta^2 T^2 \boldsymbol{\Delta} - 2 \boldsymbol{\tilde{\Sigma}}^{-1} \succeq 0$. This condition can be rewritten as follows:
\begin{equation}
	\label{eq:eigval_ineq_nig}
	\boldsymbol{P} ( \delta^2 T^2\boldsymbol{D}  - 2 \boldsymbol{\tilde{D}}^{-1})\boldsymbol{P}^{-1} \succeq 0 \iff \tilde{\lambda}_j  \geq \frac{2}{ \lambda_j \delta^2 T^2}, \quad j  \in \mathbb{I}_d.
\end{equation}
Furthermore, the integrand is proportional to $\det(\boldsymbol{\tilde{\Sigma}} ) = \prod_{j = 1}^d \tilde{\lambda}_j$; therefore, the aim is to select the matrix with a minimal determinant that satisfies the inequality in \eqref{eq:eigval_ineq_nig}. Consequently, we propose the matrix $\boldsymbol{\tilde{\Sigma}}$ by setting $\boldsymbol{\tilde{\Sigma}} = \boldsymbol{P} \diag(\frac{2}{\lambda_1 \delta^2 T^2}, \ldots, \frac{2}{\lambda_d \delta^2 T^2} ) \boldsymbol{P}^{-1} = \frac{2}{ \delta^2 T^2 } \boldsymbol{\Delta}^{-1}$.

In order to perform the domain transformation via the ICDF mapping, we represent the multivariate Laplace distribution in the variance-mean mixture form as in Theorem~6.3.1 in \cite{kotz2001laplace}, and derive an alternative integral representation of \eqref{QOI}.

\begin{proposition}
	\label{prop:nig_proposition}
	We let $g(\cdot)$ denote the integrand defined in \eqref{g_integrand}. Then its integral over $\mathbb{R}^d$ can be expressed using Fubini's theorem, as follows:
	\begin{equation}
		\label{eq:nig_transformed_integral_rep}
		\int_{\mathbb{R}^d} g(\boldsymbol{y}) d\boldsymbol{y} = \int_{[0,1]^{d+1}} 
		\frac{
			g\left( 
			\sqrt{\Psi_W^{-1}(u_{d+1})}
			\boldsymbol{\tilde{L}}	\Psi_{\boldsymbol{Z}}^{-1}
			(\boldsymbol{u}_{1:d})
			\right)
		}{
			\psi_{\boldsymbol{Y}}\left( 
			\sqrt{\Psi_W^{-1}(u_{d+1})}
			\boldsymbol{\tilde{L}}		\Psi_{\boldsymbol{Z}}^{-1}
			(\boldsymbol{u}_{1:d})
			\right)
		} 
		\mathrm{d}\boldsymbol{u}
	\end{equation}
	where $\psi_{\boldsymbol{Y}}(\cdot)$ is the PDF of the multivariate Laplace distribution with zero mean and covariance matrix $\boldsymbol{\tilde{\Sigma}}$, as given in \eqref{eq:mlap_pdf}. $\Psi^{-1}_W(\cdot)$ is the ICDF of the exponential distribution with rate equal to 1. $\Psi^{-1}_{\boldsymbol{Z}}(\cdot)$ is the ICDF of the multivariate standard normal distribution, and $\boldsymbol{\tilde{L}}$ corresponds to the square root of the matrix $\boldsymbol{\tilde{\Sigma}}$.
\end{proposition}

\begin{proof}
	Appendix~\ref{sec:nig_dom_transf_proof} presents the proof.
\end{proof}

\subsubsection[Heavy-tailed characteristic functions: illustration for the VG model]{Domain transformation for heavy-tailed characteristic functions: illustration for the VG model}

\label{sec:vg_dom_transf}

In this section we follow the same steps as in Sections \ref{sec:gbm_dom_transf} and \ref{sec:nig_dom_transf}    to obtain an appropriate domain transformation for models with heavy-tailed characteristic functions, illustrating with an example of the VG model. For the treatment of the product-form domain transformation we refer the reader to Appendix \ref{sec:domain_transf_vg_independent}.

\paragraph{Non-factorizable domain transformation}
In general, we cannot factor the joint characteristic function into the product of the marginal characteristic functions and we have that
\begin{equation}
	\phi^{VG}_{\boldsymbol{X}_T} :=   \left( 1 - \mathrm{i} \nu \boldsymbol{z}^{\top} \boldsymbol{\theta} + \frac{\nu}{2} \boldsymbol{z}^{\top} \boldsymbol{\Sigma} \boldsymbol{z} \right)^{- \frac{T}{\nu}},  \quad \boldsymbol{z} \in \mathbb{C}^d, \Im[ \boldsymbol{z} ]\in \delta^{VG}_X.
\end{equation}	
Hence, we select the proposal density $\psi^{stu}(\cdot)$ corresponding to the PDF of the multivariate generalized Student's $t$-distribution, given by
\begin{equation}
	\label{eq:general_form_vg_density}
	\psi^{stu}(\boldsymbol{y})= \frac{C_{\tilde{\nu}} }{ \sqrt{\operatorname{det}(\boldsymbol{\tilde{\Sigma}})}}\left(1+\frac{1}{\tilde{\nu}}\left(\boldsymbol{y}^{\top} \boldsymbol{\tilde{\Sigma}}^{-1} \boldsymbol{y}\right)\right)^{-\frac{\tilde{\nu}+d}{2}}, \quad \boldsymbol{y} \in \mathbb{R}^d, \tilde{\nu} > 0, \boldsymbol{\tilde{\Sigma}}\succ  0,
\end{equation}
where
\begin{equation}
	C_{\tilde{\nu}} := \frac{\Gamma\left(\frac{\tilde{\nu}+d}{2}\right)}{\Gamma\left(\frac{\tilde{\nu}}{2}\right) \tilde{\nu}^{\frac{d}{2}} \pi^{\frac{d}{2}}} > 0.
\end{equation}
Then, the function $r_{stu}^{VG}(\boldsymbol{y})$ controls the growth of the integrand near the boundary:
\begin{equation}	
	r_{stu}^{VG}( \boldsymbol{y}) := \frac{ \phi^{VG}_{\boldsymbol{X}_T} (\boldsymbol{y} \\
		+ \mathrm{i} \boldsymbol{R}) }{\psi^{stu}(\boldsymbol{y}) }, \quad \boldsymbol{y} \in \mathbb{R}^d,  \boldsymbol{R} \in \delta^{VG}_V.
\end{equation}
The characteristic function can be approximated near $|y_j|\to \infty$ as follows:
\begin{equation}
	\label{eq:chf_approx}
	\small
	\begin{aligned}
		\phi_{\boldsymbol{X}_T}^{VG}(\boldsymbol{y} + \mathrm{i} \boldsymbol{R}) &=  \left(  \frac{\nu}{2} \boldsymbol{y}^\top \boldsymbol{\Sigma} \boldsymbol{y}  \left(1 + \mathrm{i}  \frac{2}{\nu}  \frac{\boldsymbol{R}^\top\boldsymbol{\Sigma} \boldsymbol{y} - \boldsymbol{\theta}^\top\boldsymbol{y}}{\boldsymbol{y}^\top \boldsymbol{\Sigma} \boldsymbol{y}}  +  \frac{2}{\nu} \frac{1 - \nu \boldsymbol{\theta}^\top \boldsymbol{R} - \frac{\nu}{2}  \boldsymbol{R}^\top \boldsymbol{\Sigma} \boldsymbol{R}}{\boldsymbol{y}^\top \boldsymbol{\Sigma} \boldsymbol{y}} 
		\right) 
		\right)^{-T / \nu}\\
		&\stackrel{|y_j|\to \infty}{\sim}    \left( \frac{\nu}{2}\, \boldsymbol{y}^\top \boldsymbol{\Sigma} \boldsymbol{y} 
		\right)^{-T / \nu}.
	\end{aligned}
\end{equation}

Furthermore, the multivariate generalized Student’s $t$ density can be approximated asymptotically as $|y_j| \to \infty$ by 
\begin{equation}
	\label{eq:density_approx}
	\psi^{stu}(\boldsymbol{y}) \stackrel{|y_j|\to \infty}{\sim}  
	\frac{C_{\tilde{\nu}} }{ \sqrt{\operatorname{det}(\boldsymbol{\tilde{\Sigma}})}} 
	\left( 
	\frac{1}{\tilde{\nu}} \boldsymbol{y}^\top \boldsymbol{\tilde{\Sigma}}^{-1} \boldsymbol{y}
	\right)^{-\frac{\tilde{\nu}+d}{2}}.
\end{equation}
By applying the asymptotic relations in \eqref{eq:chf_approx} and \eqref{eq:density_approx}, we approximate the function $r_{stu}^{VG}(\boldsymbol{y})$ as follows:
\begin{equation}
	\label{eq:r_VG_multivariate}
	\small
	\begin{aligned}
		r_{stu}^{VG}(\boldsymbol{y}) 
		&\stackrel{|y_j| \to \infty}{\sim}  
		\underbrace{
			\frac{\sqrt{\operatorname{det}(\boldsymbol{\tilde{\Sigma}})}}{C^{\prime}_{\tilde{\nu}}} 
			\exp\left( - \left( 
			\frac{T}{\nu} 
			\log\left(\boldsymbol{y}^{\top} \boldsymbol{\Sigma} \boldsymbol{y}\right) 
			-
			\frac{\tilde{\nu}+d}{2} 
			\log\left(
			\boldsymbol{y}^{\top}
			\boldsymbol{\tilde{\Sigma}}^{-1}
			\boldsymbol{y}
			\right)
			\right) \right)
		}_{h_{stu}^{VG}(\boldsymbol{y})}.
	\end{aligned}
\end{equation}
where $$
C_{\tilde{\nu}}^{\prime} := C_{\tilde{\nu}} \exp\left( \frac{T}{\nu} \log\left(\frac{\nu}{2}\right) + \frac{\tilde{\nu} + d}{2} \log\left(\tilde{\nu}\right) \right) .
$$
The function $h_{stu}^{VG}(\boldsymbol{y})$ controls the boundary growth of the integrand. In contrast to the case of the GBM and the GH models, the removal of the singularity at the boundary depends on the interplay of two parameters: $\tilde{\nu}$ and the covariance matrix $\boldsymbol{\tilde{\Sigma}}$. By the monotonicity of $x \mapsto \log(x)$, we have
\begin{equation}
	\small
	\label{eq:mvg_dom_trans_cond}
	\log\left( \boldsymbol{y}^\top \boldsymbol{\Sigma} \boldsymbol{y} \right) 
	- \log\left( \boldsymbol{y}^\top \boldsymbol{\tilde{\Sigma}}^{-1} \boldsymbol{y} \right) 
	\geq 0 
	\iff  
	\boldsymbol{y}^\top (\boldsymbol{\Sigma} - \boldsymbol{\tilde{\Sigma}}^{-1}) \boldsymbol{y} 
	\geq 0.
\end{equation}
Equation~\eqref{eq:mvg_dom_trans_cond} indicates that if $\tilde{\nu} = \frac{2T}{\nu} - d$, then we focus on  the  three following cases
\begin{equation}
	\label{eq:mvg_conditions_1}
	\lim_{|y_j| \to \infty} h_{stu}^{VG}(\boldsymbol{y}) =
	\left\{
	\begin{array}{lll}
		+\infty & \text{if } \boldsymbol{\Sigma} - \boldsymbol{\tilde{\Sigma}}^{-1} \prec 0 & (i), \\[0.5pt]
		\frac{\sqrt{\operatorname{det}(\boldsymbol{\tilde{\Sigma}})}}{C^{\prime}_{\tilde{\nu}}} 
		& \text{if } \boldsymbol{\tilde{\Sigma}} = \boldsymbol{\Sigma}^{-1} & (ii), \\[0.5pt]
		0 & \text{if } \boldsymbol{\Sigma} - \boldsymbol{\tilde{\Sigma}}^{-1} \succ 0 & (iii).
	\end{array}
	\right.
\end{equation}
From \eqref{eq:mvg_conditions_1}, if $\tilde{\nu} = \frac{2T}{\nu} - d$, an appropriate choice of $\boldsymbol{\tilde{\Sigma}}$ is one such that $ \boldsymbol{\Sigma} - \boldsymbol{\tilde{\Sigma}}^{-1} \succeq 0$, but with smallest possible eigenvalues to minimize the term $\sqrt{\operatorname{det}(\boldsymbol{\tilde{\Sigma}})}$.  We remark that the case when $\boldsymbol{\Sigma} - \boldsymbol{\tilde{\Sigma}}^{-1} $ is an indefinite matrix is inconclusive.
In contrast, for \eqref{eq:r_VG_multivariate}, setting $\boldsymbol{\tilde{\Sigma}} = \boldsymbol{\Sigma}^{-1}$, we focus on the following three possible conditions:

\begin{equation}
	\label{eq:mvg_conditions_2}
	\lim_{|y_j| \to \infty}h_{stu}^{VG}(\boldsymbol{y}) =
	\left\{\begin{array}{l}
		+ \infty \; \text{if} \; \tilde{\nu} > \frac{2T}{\nu}-d \; (i), \\
		\frac{ \sqrt{\operatorname{det}(\boldsymbol{\tilde{\Sigma}})}}{C_{\tilde{\nu}} }   \;   \text{if} \;\tilde{\nu} = \frac{2T}{\nu}-d \; (ii),
		\\
		0  \;   \text{if} \;\tilde{\nu} < \frac{2T}{\nu}-d  \; (iii).  \\
	\end{array} \right.
\end{equation}
From \eqref{eq:mvg_conditions_2}, if $\boldsymbol{\tilde{\Sigma}} = \boldsymbol{\Sigma}^{-1}$, an appropriate choice of $\tilde{\nu}$ is given by $\tilde{\nu} = \bar{\nu} - \epsilon$, where $\bar{\nu} = \frac{2T}{\nu} - d , \epsilon \geq 0$.
In this case, increasing the value of $\epsilon$ decreases the value of $\tilde{\nu}$, which increases the value of the constant factor $C^{\prime}_{\tilde{\nu} } \propto (\tilde{\nu})^{- \frac{d}{2} +1} \log{(\tilde{\nu})} $ and hence reduces the constant factor multiplying the integrand. This result indicates that, for a fixed $\boldsymbol{\tilde{\Sigma}} = \boldsymbol{\Sigma}^{-1}$, reducing the value of $\tilde{\nu}$, which makes the tails of $\psi_{stu}(\cdot)$ heavier, may improve the performance of RQMC. 



In order to perform the domain transformation via the ICDF mapping, we  adopt a similar approach to that employed for semi-heavy-tailed models in Section \ref{sec:nig_dom_transf}. We represent the generalized Student’s $t$-distribution in the normal variance-mean mixture form and use the Cholesky or principal component factorization to eliminate the dependence structure.  The resulting representation (see Equation \eqref{eq:vg_nested_integ}) is a critical tool that avoids the need for evaluation of the ICDF, as presented in Proposition~\ref{prop:vg_proposition}.

\begin{proposition}
	\label{prop:vg_proposition}
We let $g(\cdot)$ denote the integrand defined in \eqref{g_integrand}. Then, its integral over $\mathbb{R}^d$ can be expressed using  Fubini's theorem, as follows:
	\begin{equation}
		\label{eq:vg_nested_integ}
			\int_{\mathbb{R}^d} g(\boldsymbol{y}) d\boldsymbol{y} =  	\int_{[0,1]^{d+1}} 
			\frac{
				g\left( 
				\frac{\boldsymbol{\tilde{L}}
					\Psi_{\boldsymbol{Z}}^{-1}
					(\boldsymbol{u}_{1:d})}{\sqrt{\Psi_W^{-1}(u_{d+1})}}
				\right)
			}{
				\psi_{\boldsymbol{Y}}\left( 
				\frac{\boldsymbol{\tilde{L}}
					\Psi_{\boldsymbol{Z}}^{-1}
					(\boldsymbol{u}_{1:d})}{\sqrt{\Psi_W^{-1}(u_{d+1})}}
				\right)
			} 
			\mathrm{d}\boldsymbol{u}
	\end{equation}
where $\psi_{\boldsymbol{Y}}(\cdot)$ is the PDF of the multivariate Student-$t$ distribution with zero mean, covariance matrix $\boldsymbol{\tilde{\Sigma}}$ and degrees of freedom $\tilde{\nu}$, as given in \eqref{eq:general_form_vg_density}. $\Psi^{-1}_W(\cdot)$ is the ICDF of the chi-squared distribution with degrees of freedom $\tilde{\nu}$. $\Psi^{-1}_{\boldsymbol{Z}}(\cdot)$ is the ICDF of the multivariate standard normal distribution, and $\boldsymbol{\tilde{L}}$ corresponds to the square root of the matrix $\tilde{\nu} \boldsymbol{\tilde{\Sigma}}$.

\end{proposition}
\begin{proof}
Appendix~\ref{sec:vg_dom_transf_proof} presents the proof.
\end{proof}


\subsubsection{Boundary growth conditions}

To summarize, we provide in Tables~\ref{tab:univariate_dom_tranf} and~\ref{tab:multivariate_dom_tranf} the boundary growth conditions for the three pricing models, namely the GBM, VG, and GH models. {The derivations in Sections~\ref{sec:gbm_dom_transf}, \ref{sec:vg_dom_transf}, and~\ref{sec:nig_dom_transf} are carried out for these concrete representative models. However, the methodology is not restricted to these examples. It applies to models whose extended characteristic functions have the same asymptotic decay class. Table~\ref{tab:identification} makes this connection explicit by identifying the decay class, the general form, the associated transformation density, and the corresponding boundary growth condition. Thus, GBM, GH/NIG, and VG are used as representative examples of the Gaussian, root-exponential, and polynomial decay classes, respectively.}

\begin{table}[h]
	
	\setlength{\arrayrulewidth}{0.4pt}
	\renewcommand{\arraystretch}{2}
	\caption{Choice of $\psi_j(\cdot)$ and boundary growth conditions for (i) GBM, (ii) VG, and (iii) GH in the case of independent assets. Here $y_j\in\mathbb{R}$. The densities $\psi^{nor}_j$, $\psi^{stu}_j$, and $\psi^{lap}_j$ denote the normal, Student-$t$, and Laplace transformation densities, respectively. $C_{\tilde{\nu}} = \frac{  \sqrt{\tilde{\nu}} \sqrt{ \pi} \Gamma\left(\frac{\tilde{\nu}}{2}\right)}{\Gamma\left(\frac{\tilde{\nu}+1}{2}\right)}$.}
	\centering
	\begin{tabular}{|p{1.5cm}|p{7cm}|p{4.8cm}|}
		\hline
		\textbf{Model} & $\psi_{j}(y_j)$ & Boundary growth condition \\
		\hline
		\textbf{GBM} 
		& 
		$
		\psi^{nor}_j(y_j) := 
		\frac{\exp\left(- \frac{y_j^2}{2\tilde{\sigma}_j^2}\right)}
		{\sqrt{2\pi \tilde{\sigma}_j^2}}
		$
		& 
		$\tilde{\sigma}_j \geq \frac{1}{\sqrt{T} \sigma_j}$ \\
		\hline
		\textbf{VG} 
		& 
		$\psi^{stu}_j(y_j)
		=
		\frac{\Gamma\left(\frac{\tilde{\nu}+1}{2}\right)}
		{\sqrt{\tilde{\nu}\pi} \tilde{\sigma}_j \Gamma\left(\frac{\tilde{\nu}}{2}\right)}
		\left(1+\frac{y_j^2}{\tilde{\nu} \tilde{\sigma}_j^2}\right)^{-(\tilde{\nu}+1)/2}$ 
		& 
		$\tilde{\nu} \leq \frac{2T}{\nu}-1$, \newline 
		$\tilde{\sigma}_j = \left[ \frac{\nu \sigma_j^2 \tilde{\nu}}{2} \right]^{\frac{T}{\nu - 2T}} (C_{\tilde{\nu}})^{-\frac{\nu}{\nu - 2T}}$ \\
		\hline
		\textbf{GH} 
		& 
		$\psi^{lap}_j(y_j)
		=
		\frac{\exp\left(- \frac{|y_j|}{\tilde{\sigma}_j}\right)}{2 \tilde{\sigma}_j}$ 
		& 
		$\tilde{\sigma}_j \geq \frac{1}{\delta T}$ \\
		\hline
	\end{tabular}
	\label{tab:univariate_dom_tranf}
\end{table}
\FloatBarrier

\begin{table}[h]
	
	\setlength{\arrayrulewidth}{0.4pt}
	\setlength{\tabcolsep}{1pt}
	\renewcommand{\arraystretch}{2}
	\caption{Choice of $\psi(\cdot)$ and boundary growth conditions for (i) GBM, (ii) VG, and (iii) GH in the case of dependent assets. Here $\boldsymbol{y}\in\mathbb{R}^d$. The densities $\psi^{nor}$, $\psi^{stu}$, and $\psi^{lap}$ denote the multivariate normal, generalized Student-$t$, and multivariate Laplace transformation densities, respectively.}
	\centering
	\begin{tabular}{|P{1.25cm}|P{10.6cm}|P{\dimexpr\linewidth-1.25cm-10.6cm-6\tabcolsep-4\arrayrulewidth\relax}|}
		\hline
		\textbf{Model} & \centering Density & Boundary growth condition \\
		\hline
		\textbf{GBM} 
		& 
		$\psi^{nor}(\boldsymbol{y})=(2\pi)^{-d/2}(\det(\boldsymbol{\tilde{\Sigma}}))^{-1/2}\exp\left(-\frac{1}{2}\boldsymbol{y}^\top\boldsymbol{\tilde{\Sigma}}^{-1}\boldsymbol{y}\right)$
		& 
		$T\boldsymbol{\Sigma}-\boldsymbol{\tilde{\Sigma}}^{-1}\succeq0$ \\
		\hline
		\textbf{VG} 
		& 
		$\psi^{stu}(\boldsymbol{y})=\frac{\Gamma\left(\frac{\tilde{\nu}+d}{2}\right)(\det(\boldsymbol{\tilde{\Sigma}}))^{-1/2}}{\Gamma\left(\frac{\tilde{\nu}}{2}\right)\tilde{\nu}^{d/2}\pi^{d/2}}\left(1+\frac{1}{\tilde{\nu}}\boldsymbol{y}^\top\boldsymbol{\tilde{\Sigma}}^{-1}\boldsymbol{y}\right)^{-(\tilde{\nu}+d)/2}$
		& 
		$\tilde{\nu}=\frac{2T}{\nu}-d$, 
		$\boldsymbol{\Sigma}-\boldsymbol{\tilde{\Sigma}}^{-1}\succeq0$;
		\newline
		or $\tilde{\nu}\leq\frac{2T}{\nu}-d$, 
		$\boldsymbol{\tilde{\Sigma}}=\boldsymbol{\Sigma}^{-1}$ \\
		\hline
		\textbf{GH} 
		& 
		$\psi^{lap}(\boldsymbol{y})=\frac{2}{(2\pi)^{d/2}}(\det(\boldsymbol{\tilde{\Sigma}}))^{-1/2}\left(\frac{\boldsymbol{y}^\top\boldsymbol{\tilde{\Sigma}}^{-1}\boldsymbol{y}}{2}\right)^{v/2}K_v\left(\sqrt{2\boldsymbol{y}^\top\boldsymbol{\tilde{\Sigma}}^{-1}\boldsymbol{y}}\right)$
		& 
		$\delta^2T^2\boldsymbol{\Delta}-2\boldsymbol{\tilde{\Sigma}}^{-1}\succeq0$ \\
		\hline
	\end{tabular}
	\label{tab:multivariate_dom_tranf}
\end{table}

\FloatBarrier
\begin{table}[h!]
	
	\setlength{\arrayrulewidth}{0.4pt}
	\setlength{\tabcolsep}{3pt}
	\renewcommand{\arraystretch}{1.6}
	\caption{General forms, associated transformation densities, and boundary growth conditions. Here $\boldsymbol{y}\in\mathbb{R}^d$, $C>0$, $\boldsymbol{A}\succ0$, and $\gamma>0$. The densities $\psi^{nor}$, $\psi^{lap}$, and $\psi^{stu}$ are the transformation densities summarized in Tables~\ref{tab:univariate_dom_tranf} and~\ref{tab:multivariate_dom_tranf}. In the heavy-tailed case, $\gamma>d/2$ is imposed.}
	\centering
	\begin{tabular}{|P{2.6cm}|P{4.7cm}|P{2.4cm}|P{\dimexpr\linewidth-2.6cm-4.7cm-2.4cm-8\tabcolsep-5\arrayrulewidth\relax}|}
		\hline
		\textbf{Decay class} & \centering  General form & \centering  Density & Boundary growth condition \\
		\hline
		\textbf{Light-tailed}
		& 
		$C\exp(-\boldsymbol{y}^{\top}\boldsymbol{A}\boldsymbol{y})$ 
		& 
		$\psi^{nor}(\boldsymbol{y})$ 
		& 
		$2\boldsymbol{A}-\boldsymbol{\tilde{\Sigma}}^{-1}\succeq0$ \\
		\hline
		\textbf{Semi-heavy-tailed}
		& 
		$C\exp(-\gamma\sqrt{\boldsymbol{y}^{\top}\boldsymbol{A}\boldsymbol{y}})$ 
		& 
		$\psi^{lap}(\boldsymbol{y})$ 
		& 
		$\gamma^2\boldsymbol{A}-2\boldsymbol{\tilde{\Sigma}}^{-1}\succeq0$ \\
		\hline
		\textbf{Heavy-tailed}
		& 
		$C(\boldsymbol{y}^{\top}\boldsymbol{A}\boldsymbol{y})^{-\gamma}$ 
		& 
		$\psi^{stu}(\boldsymbol{y})$ 
		& 
		$\tilde{\nu}=2\gamma-d,\ \boldsymbol{A}-\boldsymbol{\tilde{\Sigma}}^{-1}\succeq0$;
		or
		$\tilde{\nu}\leq2\gamma-d,\ \boldsymbol{\tilde{\Sigma}}=\boldsymbol{A}^{-1}$ \\
		\hline
	\end{tabular}
	\label{tab:identification}
\end{table}
\FloatBarrier

{
	Table~\ref{tab:identification} identifies the general rule underlying the model-specific transformation choices. The GBM model belongs to the light-tailed class with $\boldsymbol{A}=\frac{T}{2}\boldsymbol{\Sigma}$. The GH and NIG models belong to the semi-heavy-tailed class with $\boldsymbol{A}=\boldsymbol{\Delta}$ and $\gamma=\delta T$. The VG model belongs to the heavy-tailed class with $\boldsymbol{A}=\boldsymbol{\Sigma}$ and $\gamma=T/\nu$. Once the asymptotic decay class of a computable extended characteristic function has one of the general forms listed in Table~\ref{tab:identification}, the corresponding transformation density and boundary growth condition are obtained from the same table. This is the sense in which Sections~\ref{sec:gbm_dom_transf}, \ref{sec:vg_dom_transf}, and~\ref{sec:nig_dom_transf} provide representative derivations rather than model-specific restrictions.
}

\FloatBarrier
We conclude this section on the domain transformation with a few final remarks.

{
	\begin{remark}[Extension to other models]
		The transformation rules in Table~\ref{tab:identification} are not restricted to the three representative models considered in Sections~\ref{sec:gbm_dom_transf}, \ref{sec:nig_dom_transf}, and~\ref{sec:vg_dom_transf}. Given a model with computable extended characteristic function, one first fixes an admissible damping vector $\boldsymbol{R}\in\delta_V$ and determines, or bounds, the asymptotic decay of $|\Phi_{\boldsymbol{X}_T}(\boldsymbol{y}+\mathrm{i}\boldsymbol{R})|$ as $|y_j|\to\infty$, $j=1,\ldots,d$. If this decay matches one of the general forms in Table~\ref{tab:identification}, the corresponding transformation density and boundary growth condition are obtained by identifying the associated matrix $\boldsymbol{A}$ and, when applicable, the exponent $\gamma$.
		
		For the models treated in this work, GBM corresponds to the light-tailed class with $\boldsymbol{A}=\frac{T}{2}\boldsymbol{\Sigma}$, GH/NIG corresponds to the semi-heavy-tailed class with $\boldsymbol{A}=\boldsymbol{\Delta}$ and $\gamma=\delta T$, and VG corresponds to the heavy-tailed class with $\boldsymbol{A}=\boldsymbol{\Sigma}$ and $\gamma=T/\nu$. Other models can be treated in the same way once their characteristic function asymptotics have been identified. For instance, the characteristic functions in the Merton jump-diffusion model and Kou's model contain a Gaussian-type term, see \cite{kirkby2015efficient}. Hence, when the Gaussian term determines the asymptotic decay, the normal transformation used for the GBM remains applicable. Similarly, estimates for the decay of the Heston characteristic function, such as those in \cite{lee2004option}, can be used to identify the appropriate decay class.

			 More generally, when the leading asymptotic decay of 
			 $|\Phi_{\boldsymbol{X}_T}(\boldsymbol{y}+\mathrm{i}\boldsymbol{R})|$ 
			 is unknown, but upper bounds or asymptotic estimates are available, the  procedure remains applicable. The resulting transformation becomes more conservative, and its numerical efficiency will depend on the sharpness of the bound. Nevertheless, the corresponding boundary growth conditions remain crucial to avoid boundary singularities introduced by the domain transformation.

	\end{remark}
}

\begin{remark}[Damping parameters]
	The value of the damping parameters is independent of the domain transformation; thus, the rule proposed in \cite{bayer2023optimal} remains the same in this work. The independence comes from using the damping parameters that minimize the peak of the integrand at the origin, corresponding to $\Psi^{-1}(\boldsymbol{u})= \mathbf{0}$ i.e., $\boldsymbol{u} = \left( \frac{1}{2}, \ldots, \frac{1}{2}\right)$ for the transformed integrand. Hence, the original integrand is divided by $\psi(\mathbf{0})$, a constant term independent of $\boldsymbol{R}$. We find that it is numerically more stable to minimize the peak of the log-transformed integrand i.e., $\log(\mid g(\boldsymbol{0}; \boldsymbol{R}) \mid)$ instead of minimizing the peak of the integrand  i.e., $g(\boldsymbol{0}; \boldsymbol{R})$. 	For high-dimensional problems $d > 10$, some optimizers, such as L-BFGS-B \cite{zhu1997algorithm}, may not converge to the optimal solution. However, the trust-region method (see \cite{conn2000trust}) was empirically observed to be robust in high dimensions.
\end{remark}

{
	\begin{remark}[Extension to path-dependent options]
		The valuation formula in Proposition \ref{prop:Multivariate Fourier pricing valuation formula} is applied in this paper to European-type options, where the payoff is written as a function of the terminal log-price vector $\boldsymbol{X}_T$. For path-dependent options, the RQMC quadrature and the domain transformation procedure developed in Section \ref{sec:RQMC_fourier_pricing} can be used only after a Fourier valuation formula of the same form has first been obtained.
		
		More precisely, suppose that the path-dependent payoff admits the representation $P(\boldsymbol{F}_T)$, where $\boldsymbol{F}_T=(F_T^{(1)},\ldots,F_T^{(m)})\in\mathbb{R}^m$ is a vector of path functionals of $(\boldsymbol{X}_t)_{0\le t\le T}$, or equivalently of $(\boldsymbol{S}_t)_{0\le t\le T}$, such that the payoff is fully determined by $\boldsymbol{F}_T$. Examples include monitored log-prices, averages of log-prices, averages of asset prices, or running extrema. The analogue of Proposition~\ref{prop:Multivariate Fourier pricing valuation formula} then requires the extended Fourier transform $\widehat{P}$ of the payoff and the extended characteristic function
		$$
		\Phi_{\boldsymbol{F}_T}(\boldsymbol{z})
		=
		\mathbb{E}
		\left[
		\exp\left(
		\mathrm{i}\boldsymbol{z}^{\top}\boldsymbol{F}_T
		\right)
		\right]
		$$
		of this path-functional vector, both defined on a common strip of analyticity. If these objects are available, then the option value can be written as
		$$
		V
		=
		(2\pi)^{-m}e^{-rT}
		\int_{\mathbb{R}^m}
		\operatorname{Re}
		\left[
		\Phi_{\boldsymbol{F}_T}(\boldsymbol{y}+\mathrm{i}\boldsymbol{R})
		\widehat{P}(\boldsymbol{y}+\mathrm{i}\boldsymbol{R})
		\right]
		\,\mathrm{d}\boldsymbol{y},
		$$
		and the RQMC/domain transformation methodology applies to this integration problem.
		
		For instance, for a continuously monitored fixed-strike geometric Asian call in one dimension, one may take
		$$
		F_T
		=
		\frac{1}{T}\int_0^T X_t\,\mathrm{d}t,
		\qquad
		P(x)
		=
		\left(e^x-K\right)^+.
		$$
		The option payoff is then $P(F_T)$, since $e^{F_T}$ is the geometric average of the asset price. Under GBM, $F_T$ is Gaussian, and therefore $\Phi_{F_T}$ is available explicitly. For more general L\'evy models, however, the required input is the extended characteristic function of the path functional $F_T$, not merely the terminal characteristic function $\Phi_{X_T}$. This characteristic function has to be derived separately, and is in most cases not given explicitly in a form that immediately allows one to analyze its asymptotic decay and determine the appropriate domain transformation.
		
		For some path-dependent options under L\'evy models, different transform tools may be available. For example, \cite{eberlein2011analyticity} use Wiener-Hopf factorization to obtain semi-analytical formulas for the extended characteristic functions of the supremum and infimum of a L\'evy process, leading to Fourier valuation formulas for options depending on running extrema, such as one-touch and lookback options. In this setting, the Wiener-Hopf factors consist of integrals that have to be computed numerically. An alternative route was proposed by \cite{hackmann2016approximating}, who approximate L\'evy processes with completely monotone jumps, including VG and NIG models, by hyperexponential processes. The advantage of the hyperexponential class is that its Wiener-Hopf factors are given explicitly, which leads to more efficient transform-based methods for barrier, lookback, and Asian options.
		
		Thus, the extension to path-dependent options is model- and payoff-specific. Once the corresponding Fourier representation has been derived, or obtained through a tractable approximation, the RQMC and domain transformation procedure can be applied.
	\end{remark}
}
\section{Numerical Experiments and Results}\label{sec: num_exp_results}
This section presents the results of numerical experiments conducted for pricing multi-asset European basket put options with equal weights (i.e., $w_i = \frac{1}{d}$), spread call, call on minimum (call on min), and cash-or-nothing (CON) call options. 	Table~\ref{table:payoffs} presents the scaled version of these payoffs and their Fourier transforms. These payoff functions adhere to Assumption \ref{ass:Assumptions on the payoff}. Further details on their derivation are provided in \cite{eberlein2010analysis, hubalek2005variance, hurd2010fourier}.
\FloatBarrier
\begin{table}[h!]
	\centering
	\begin{tabular}{| c| c| c| }
		\hline 
		\small Payoff &  \small $P(\boldsymbol{X}_T)$&  \small  \small $\hat{P}(\boldsymbol{z})$\\
		\hline
		\small	Basket put & \small	$ \max \left(1 -   \sum_{j=1}^d e^{X_T^j}, 0\right)$ & $ \frac{\prod_{j=1}^{d} \Gamma\left(- \mathrm{i} z_{j}\right)}{\Gamma\left(-\mathrm{i} \sum_{j=1}^{d} z_{j}+2\right)}$ \\ 
		\hline
		
		\small Spread call & \small $\max \left( e^{X_T^1} -   \sum_{j=2}^d e^{X_T^j} - 1 , 0\right)$ &  $\frac{\Gamma\left(\mathrm{i}\left(z_1+\sum_{j=2}^d z_j\right)-1\right) \prod_{j=2}^d \Gamma\left(-\mathrm{i} z_j\right)}{\Gamma\left(\mathrm{i} z_1+1\right)}$ \\ 
		\hline
		\small	Call on min & \small	$ \max \left( \min\left(e^{X_T^1},\ldots,e^{X_T^d}\right) - 1  , 0\right)$&  $ \frac{  1 }{\left(\mathrm{i}\left(\sum_{j=1}^{d} z_{j}\right)-1\right) \prod_{j=1}^{d}\left(\mathrm{i} z_{j}\right)}$ \\
		\hline
		\small CON call & \small $\prod_{ j = 1}^d \mathds{1}_{ \{ e^{X^j_T} > 1  \} }(X^j_T)$ &  $\prod_{j = 1}^d  
		\left( \frac{ 1}{ \mathrm{i} z_j} \right)$ \\
		\hline
	\end{tabular}
	\captionsetup{justification=centering,margin=0.3cm}
	\caption{Payoff functions (scaled), $P(\boldsymbol{X}_T)$, their extended Fourier transform, $\hat{P}(\boldsymbol{z})$. The corresponding domain of analyticity for each of the payoff functions, $\delta_P$, is provided in Table \ref{table:payoff_strip_table}. }
	\label{table:payoffs}
\end{table}
\FloatBarrier

We note that when working with scaled payoffs, it is necessary to define the variable $ X_t^j $ appropriately. For a basket put option, this is given by: $X_t^j := \log\left(\frac{S_t^j}{dK}\right),$ while for a call on min, CON call, and spread call options, it is defined as: $ X_t^j :=\log\left(\frac{S_t^j}{K}\right)$.
We tested the performance of RQMC with the appropriate domain transformation (see Section~\ref{sec:RQMC_fourier_pricing}) for the GBM, VG, and GH models with various parameter constellations and dimensions $d = 1, \ldots ,15$. The tested model parameters of marginal distributions are taken from the literature on model calibration \cite{kirkby2015efficient, bayer2018smoothing,aguilar2020some,healy2021pricing}. We considered relative errors normalized by the reference prices to compare the methods. The statistical error of RQMC is defined as in \eqref{rqmc_error}, and the relative statistical error is given by $$\text{Relative Statistical Error} = \frac{ \text{Statistical Error} }{\text{Reference Value}},$$
where the reference values are computed using the MC method with $M = 10^8$ samples, {unless stated otherwise}. The numerical results were obtained using Google Colab with the standard configuration. {The computations reported below follow the RQMC-Fourier pricing pipeline summarized in Algorithm~\ref{alg:rqmc_fourier} in Appendix~\ref{app:rqmc_fourier_pipeline}.} The code containing the implementation of our proposed approach is available on GitHub\footnote{\small
	\href{https://github.com/Michael-Samet/Quasi-Monte-Carlo-for-Efficient-Fourier-Pricing-of-Multi-Asset-Options/tree/main/Premia\%20Implementation}{https://github.com/Michael-Samet/Quasi-Monte-Carlo-for-Efficient-Fourier-Pricing}
}.
\FloatBarrier
\subsection{Effect of Domain Transformation on RQMC Convergence}
\label{sec:dom_transf_effect}
This section illustrates the effect of the parameters of the distribution proposed for the domain transformation in Section~\ref{sec:RQMC_fourier_pricing} on the convergence of the RQMC method for put options in 1D under the GBM, VG, and GH models. Figures~\ref{fig:gbm_dom_trans_effect}, \ref{fig:nig_dom_trans_effect}, \ref{fig:vg1_dom_trans_effect}, and \ref{fig:vg2_dom_trans_effect} demonstrate that the values of the parameters that do not satisfy the boundary growth conditions presented in Table~\ref{tab:univariate_dom_tranf} lead to integrands that are unbounded near the boundary of  $[0,1]$. Moreover, Figures~\ref{fig:gbm_dom_trans_effect_conv}, \ref{fig:nig_dom_trans_effect_conv}, \ref{fig:vg1_dom_trans_effect_conv}, and \ref{fig:vg2_dom_trans_effect_conv} demonstrate that these singular integrands exhibit much slower convergence rates of the QMC method. For instance, Figure~\ref{fig:gbm_dom_trans_effect} indicates that, for the GBM model, when the boundary growth condition is violated (i.e., $\tilde{\sigma} = 1 < \frac{1}{\sqrt{T} \sigma} = 5$), the integrand $\tilde{g}(u)$ increases considerably near the boundary, and the associated convergence of RQMC deteriorates. This case is interesting because the choice of $\tilde{\sigma} =1$ is typical in the literature (e.g., in \cite{ballotta2022powering}). The parameter $\tilde{\sigma}$ does not carry a physical significance; thus, the transformation is usually performed using the standard normal distribution, which adversely affects the convergence in our setting. The corresponding error is two orders of magnitude larger than the error obtained by RQMC when the parameter $\tilde{\sigma}$ satisfies the boundary growth condition. In addition, Figures \ref{fig:vg1_dom_trans_effect_conv}, \ref{fig:vg2_dom_trans_effect_conv} visualize the considerable influence of the choice of both parameters, $\tilde{\nu}$ and $\tilde{\sigma}$, on the convergence of RQMC in the VG model. These results motivate the use of the generalized Student’s $t$-distribution instead of its standard counterpart in which the scaling parameter is fixed to $\tilde{\sigma} = 1$, as in \cite{ouyang2024achieving}. Finally, Figure~\ref{fig:nig_dom_trans_effect_conv} illustrates that, for the GH model, the error of RQMC is about three orders of magnitude lower than the case of $\tilde{\sigma} = 1$,  if the domain transformation parameters are chosen appropriately, according to the procedure proposed in Section~\ref{sec:nig_dom_transf}.
\FloatBarrier
\begin{figure}[h!]
		\vspace{-0.1cm}
 \centering	
 	\begin{subfigure}{0.4\textwidth}
\includegraphics[width=\linewidth]{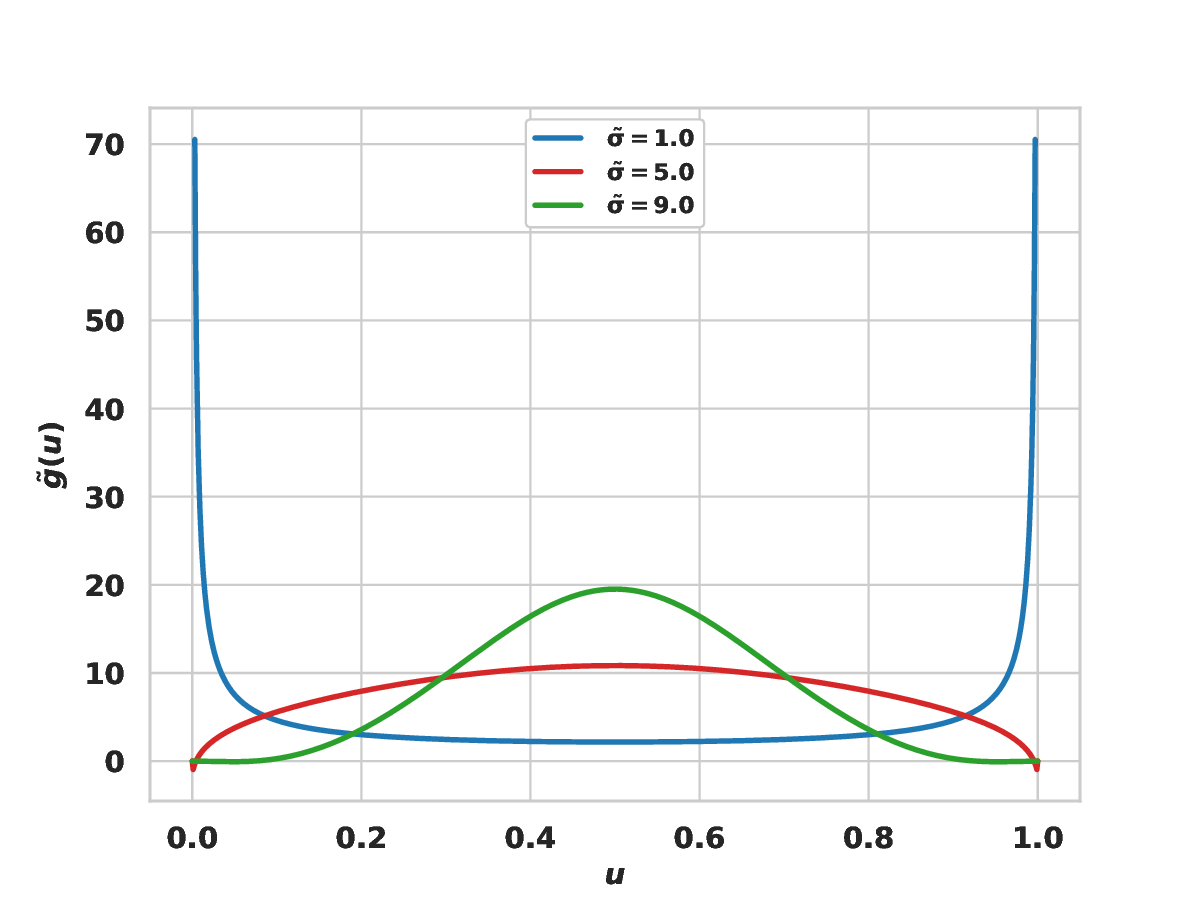}
\caption{}
\label{fig:gbm_dom_trans_effect}
 	\end{subfigure}
  	\begin{subfigure}{0.4\textwidth}
\includegraphics[width=\linewidth]{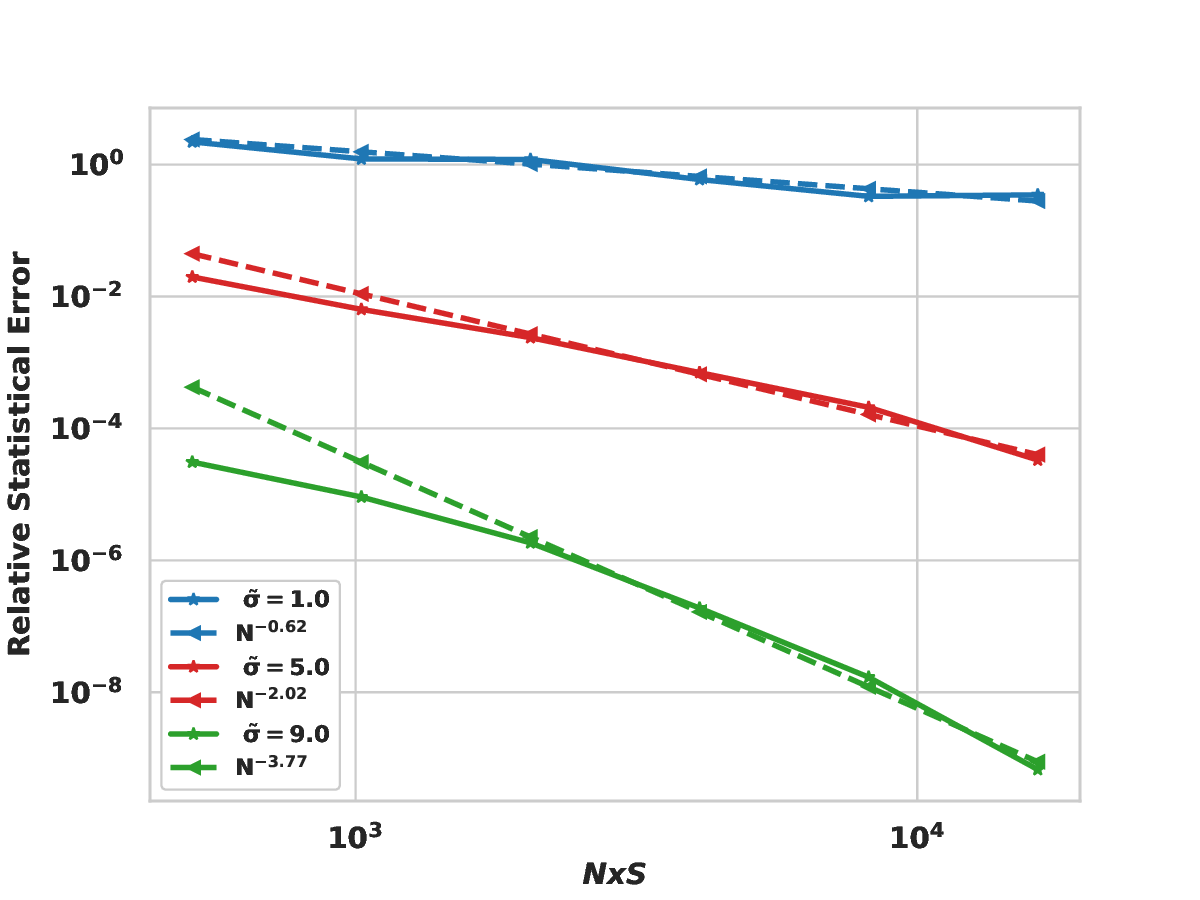}
\caption{}
\label{fig:gbm_dom_trans_effect_conv}
 	\end{subfigure}
	\caption[GBM: Effect of the parameter $\tilde{\sigma}$ on (a) the shape of the transformed integrand, $\tilde{g}(u)$,  and (b) convergence of the relative statistical error of QMC]{Effect of the parameter $\tilde{\sigma}$ on (a) the shape of the transformed integrand $\tilde{g}(u)$ and (b) convergence of the relative statistical error of RQMC with $S_0 = 100$, $K = 100$, $r = 0$, $T= 1$, and $\sigma = 0.2$ for a one-dimensional call option under the GBM model. $N$: number of QMC points; $S = 32$: number of digital shifts. Boundary growth condition limit: $\overline{\sigma} = \frac{1}{ \sqrt{T}\sigma} = 5$.} 
\end{figure}
\FloatBarrier

\FloatBarrier
\begin{figure}[h!]
	\centering	
	\begin{subfigure}{0.4\textwidth}
		\includegraphics[width=\linewidth]{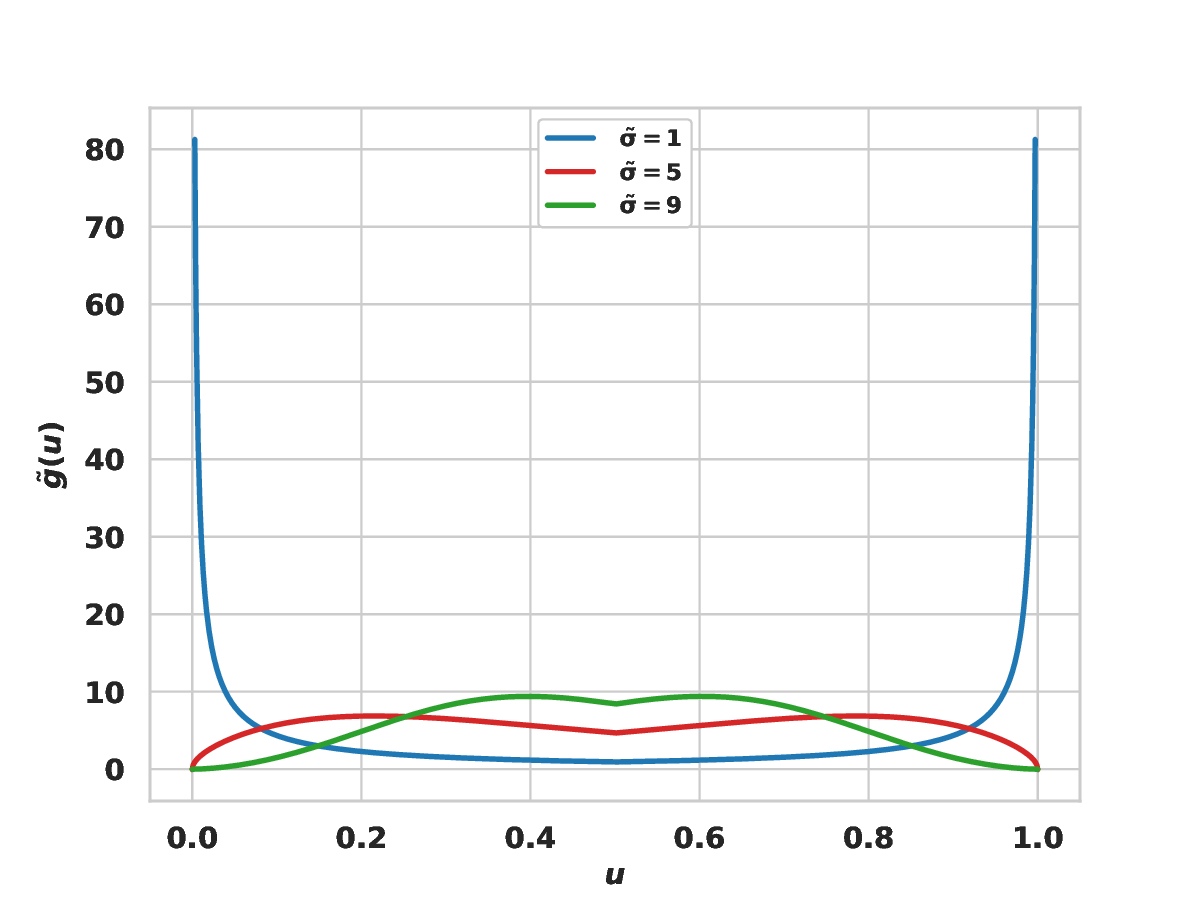}
		\caption{}
		\label{fig:nig_dom_trans_effect}
	\end{subfigure}
	\begin{subfigure}{0.4\textwidth}
		\includegraphics[width=\linewidth]{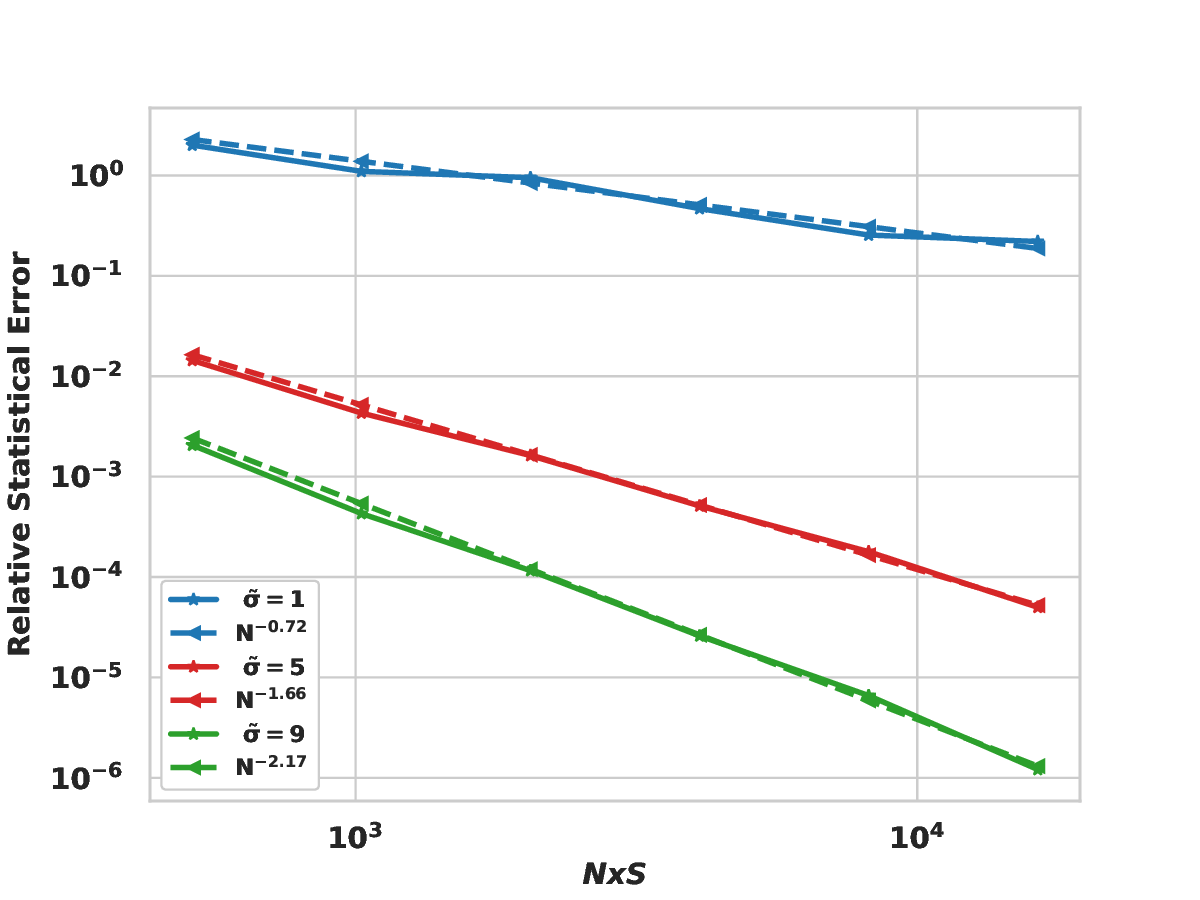}
		\caption{}
		\label{fig:nig_dom_trans_effect_conv}
	\end{subfigure}
	\caption[GH: Effect of the parameter $\tilde{\sigma}$ on (a) the shape of the transformed integrand, $\tilde{g}(u)$,  and (b) convergence of the relative statistical error of QMC]{Effect of the parameter $\tilde{\sigma}$ on (a) the shape of the transformed integrand $\tilde{g}(u)$ and (b) convergence of the relative statistical error of RQMC with $S_0 = 100$, $K = 100$, $r = 0$, $T= 1$, $\alpha = 20$, $\beta = -3, \delta = 0.2$ and $\lambda = 1$ for a one-dimensional call option under the GH  model. $N$: number of QMC points; $S = 32$: number of digital shifts. Boundary growth condition limit: $\overline{\sigma} = \frac{1}{ T \delta } = 5$.} 

\end{figure}
\FloatBarrier

\begin{figure}[H] 
	\centering
		\begin{subfigure}{0.4\textwidth}
		\includegraphics[width=\linewidth]{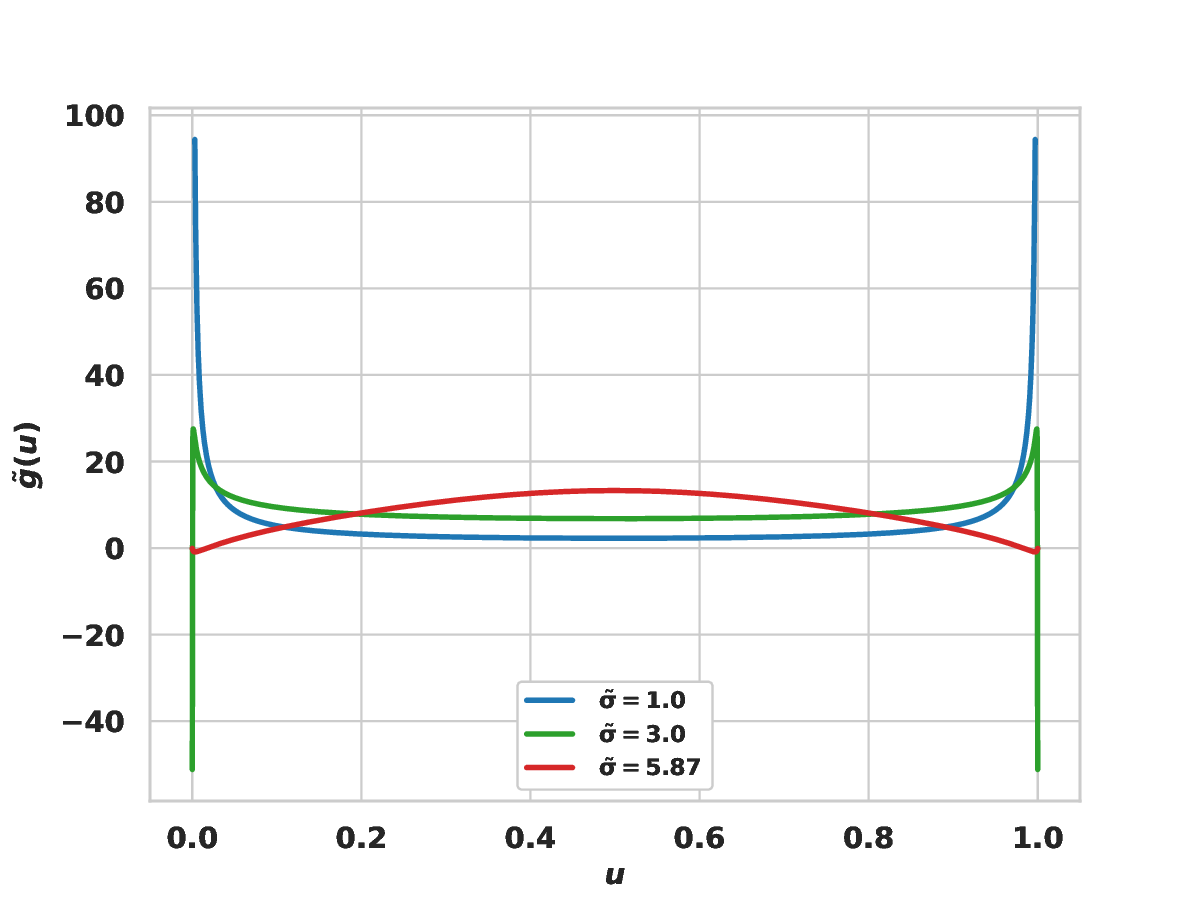} 
		\caption{}
		\label{fig:vg1_dom_trans_effect}
	\end{subfigure}
	\begin{subfigure}{0.4\textwidth}
		\includegraphics[width=\linewidth]{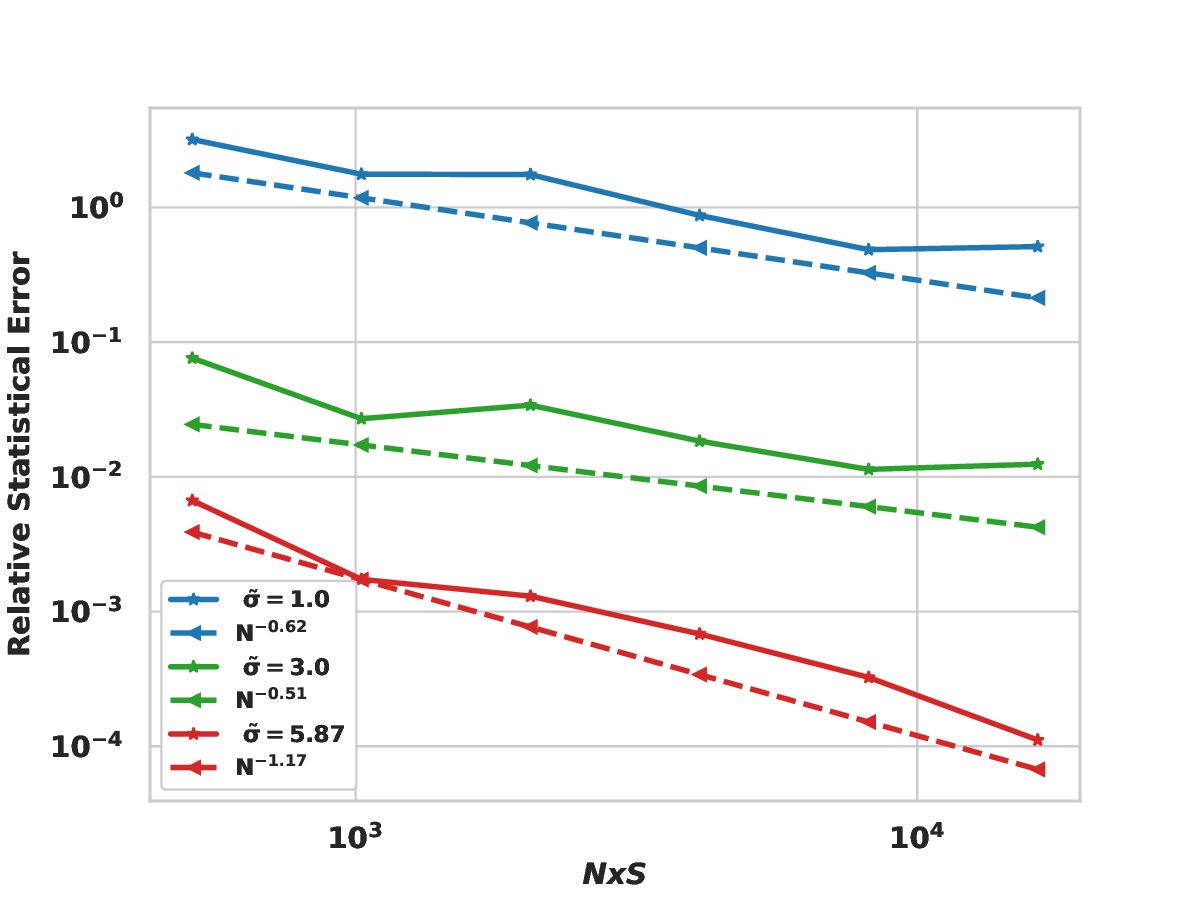}
		\caption{}
		\label{fig:vg1_dom_trans_effect_conv}
	\end{subfigure}
	\caption[VG: Effect of the parameter $\tilde{\sigma}$ on (left) the shape of the transformed integrand, $\tilde{g}(u)$,  and (right) convergence of statistical error of QMC]{Effect of the parameter $\tilde{\sigma}$ on (a) the shape of the transformed integrand $\tilde{g}(u)$ and (b) convergence of the RQMC error for a one-dimensional call option under the VG model with $S_0 = 100$, $K = 100$, $r = 0$, $T= 1$, $\sigma = 0.2$, $\theta = -0.3$, and $\nu = 0.1$. $N$: number of QMC points; $S = 32$: number of digital shifts. For the domain transformation, $\tilde{\nu} = \frac{2T}{ \nu } -1 = 19$. The critical value for the domain transformation is $\tilde{\sigma} = \left[ \frac{\nu \sigma^2 \tilde{\nu}}{2} \right]^{ \frac{T}{\nu - 2T}} (C_{\tilde{\nu}})^{-\frac{\nu}{ \nu - 2T}} = 5.87$.  } 
\end{figure}

\begin{figure}[H] 
	\centering
	\begin{subfigure}{0.4\textwidth}
		\includegraphics[width=\linewidth]{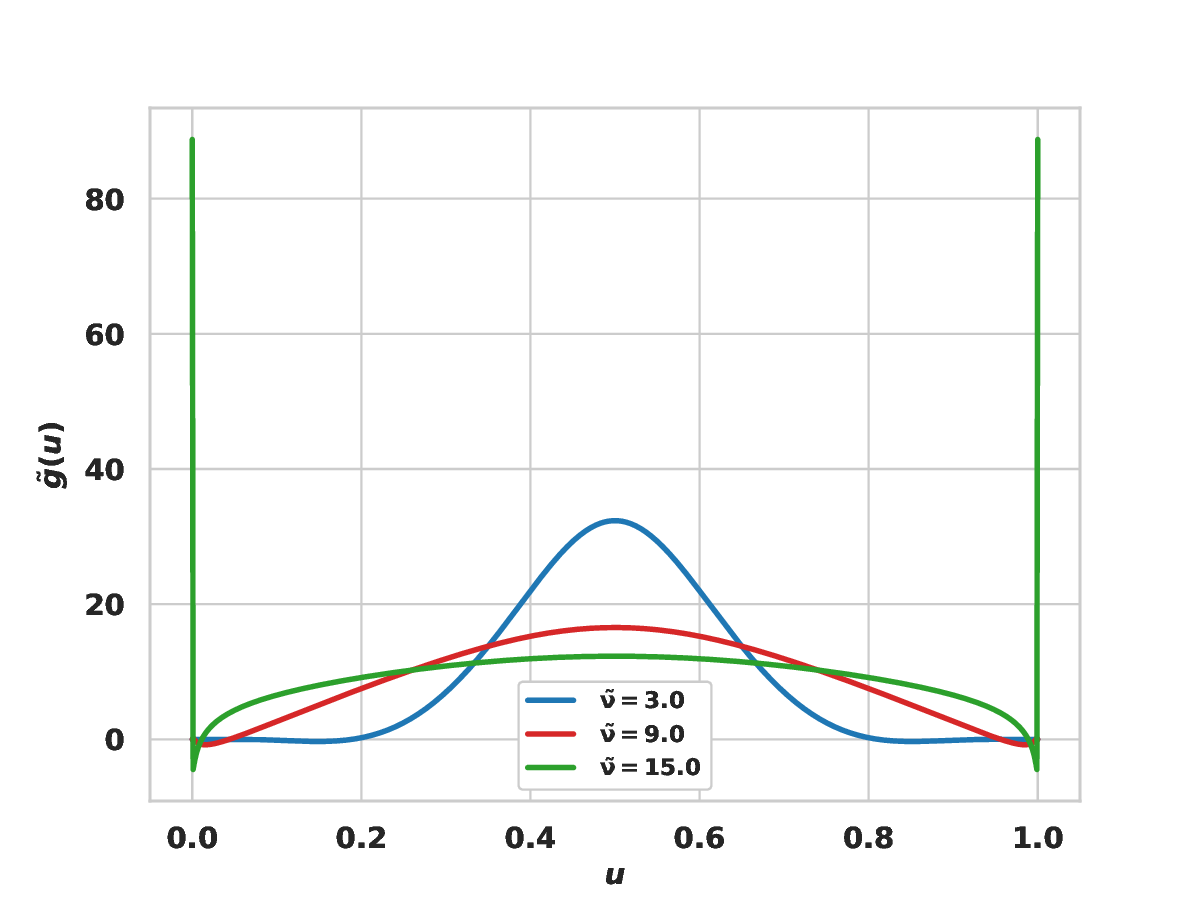}
		\caption{}
		\label{fig:vg2_dom_trans_effect}
	\end{subfigure}
	\begin{subfigure}{0.4\textwidth}
		\includegraphics[width=\linewidth]{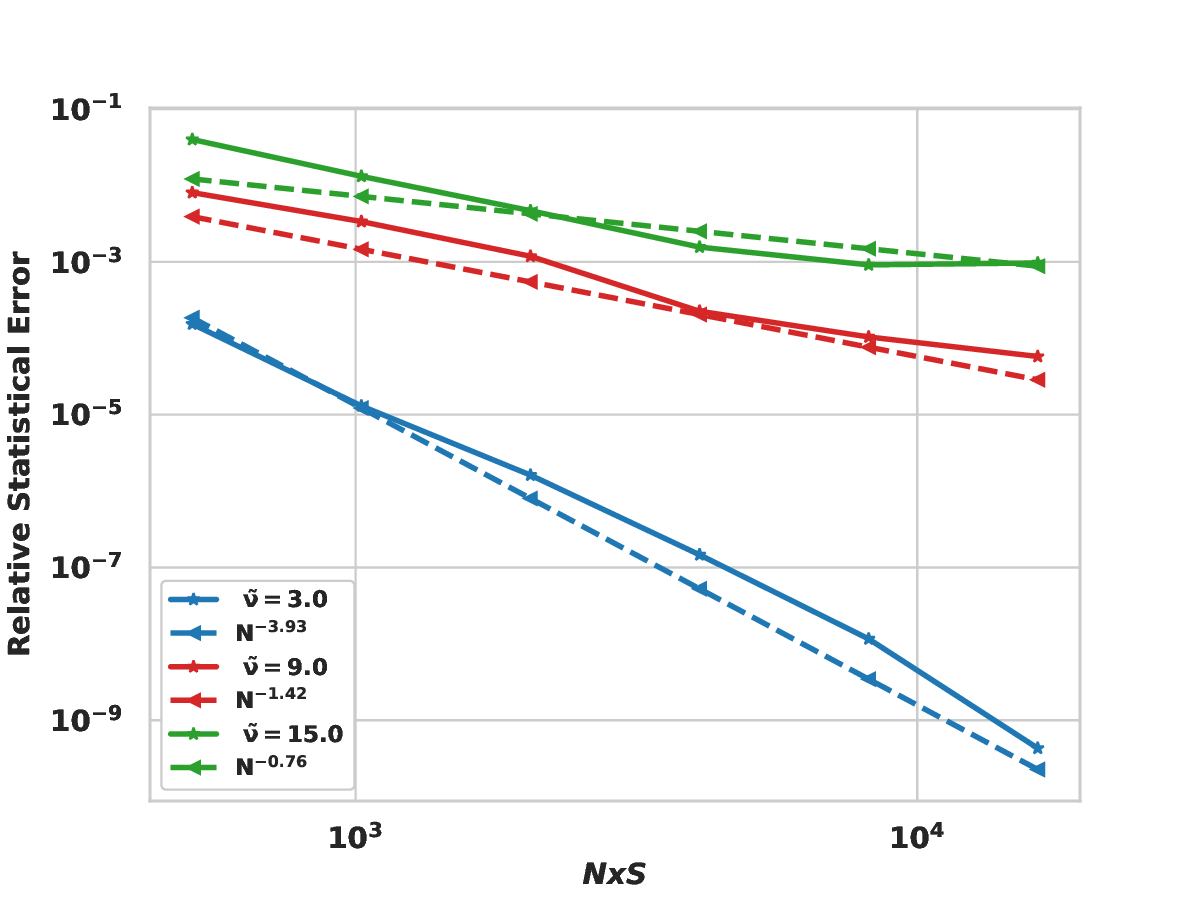} 
		\caption{}
		\label{fig:vg2_dom_trans_effect_conv}
	\end{subfigure}
	\caption[VG: Effect of the parameter $\tilde{\nu}$ on (left) the shape of the transformed integrand, $\tilde{g}(u)$,  and (right) convergence of statistical error of QMC]{Effect of the parameter $\tilde{\nu}$ on (a) the shape of the transformed integrand $\tilde{g}(u)$ and (b) convergence of the RQMC error for a one-dimensional call option under the VG model with $S_0 = 100$, $K = 100$, $r = 0$, $T= 1$, $\sigma = 0.2$, $\theta = -0.3$, and $\nu = 0.2$. $N$: number of QMC points; $S = 32$: number of digital shifts. For the domain transformation, $\tilde{\sigma} = \left[ \frac{\nu \sigma^2 \tilde{\nu}}{2} \right]^{ \frac{T}{\nu - 2T}} (C_{\tilde{\nu}})^{-\frac{\nu}{ \nu - 2T}}$ for each value of $\tilde{\nu}$. The critical value for the domain transformation is $\bar{\nu} = \frac{2T}{\nu}-1= 9$.   } 
\end{figure}

\subsection[Computational Comparison with the MC method in the Physical Space]{Computational Comparison of the Proposed Approach with the MC method in the Physical Space}

\label{sec:qmc_vs_mc_phys_sec}
This section demonstrates the advantage of the RQMC method in the Fourier space compared to the MC method in the physical space when the domain transformation from $\mathbb{R}^d$ to $[0,1]^d$ is appropriately performed as proposed in Section~\ref{sec:model_spec_dom_transf}. For illustration, we cover some examples of call on min, CON call, basket put and spread call options under the VG and the GH models.  Figures \ref{fig:phymc_vs_fourqmc_vg_CON_call_on_min} and \ref{fig:phymc_vs_fourqmc_vg_CON_call} reveal that the proposed approach significantly outperforms the MC method for options with up to six assets, particularly for small relative tolerances because the convergence rate of RQMC can be up to twice as fast as the convergence rate of the MC method. On the other hand, the advantage of using the RQMC over the MC method is less pronounced for the basket put and spread call options, which we illustrate for the example of the GH model. Figures \ref{fig:phymc_vs_fourqmc_nig_basket_put} and \ref{fig:phymc_vs_fourqmc_nig_spread_call} show that for three-dimensional basket put and spread call options, the advantage of using RQMC in the Fourier space over the MC method depends on the target relative tolerance level, with clear advantage observed for tolerances lower than $\text{TOL} = 10^{-2}$. A similar conclusion was drawn in the work of \cite{junike2023multidimensional}. We note that in comparison to call on min and CON call options in Figures  \ref{fig:cpu_vg_call_on_min} and \ref{fig:cpu_vg_CON_call}, the runtime of the RQMC method to achieve the $\text{TOL} = 10^{-1}$ is slower, although the convergence rates are  superior to those of the MC method. Nevertheless, Figures \ref{fig:phymc_vs_fourqmc_nig_basket_put_otm} and \ref{fig:phymc_vs_fourqmc_nig_spread_call_otm}  demonstrate that the RQMC method  handles basket put and spread call options which are deep out-of-the-money significantly better than the MC method.



\FloatBarrier
\begin{figure}[h!]
	\centering	
	\begin{subfigure}{0.4\textwidth}
		\includegraphics[width=\linewidth]{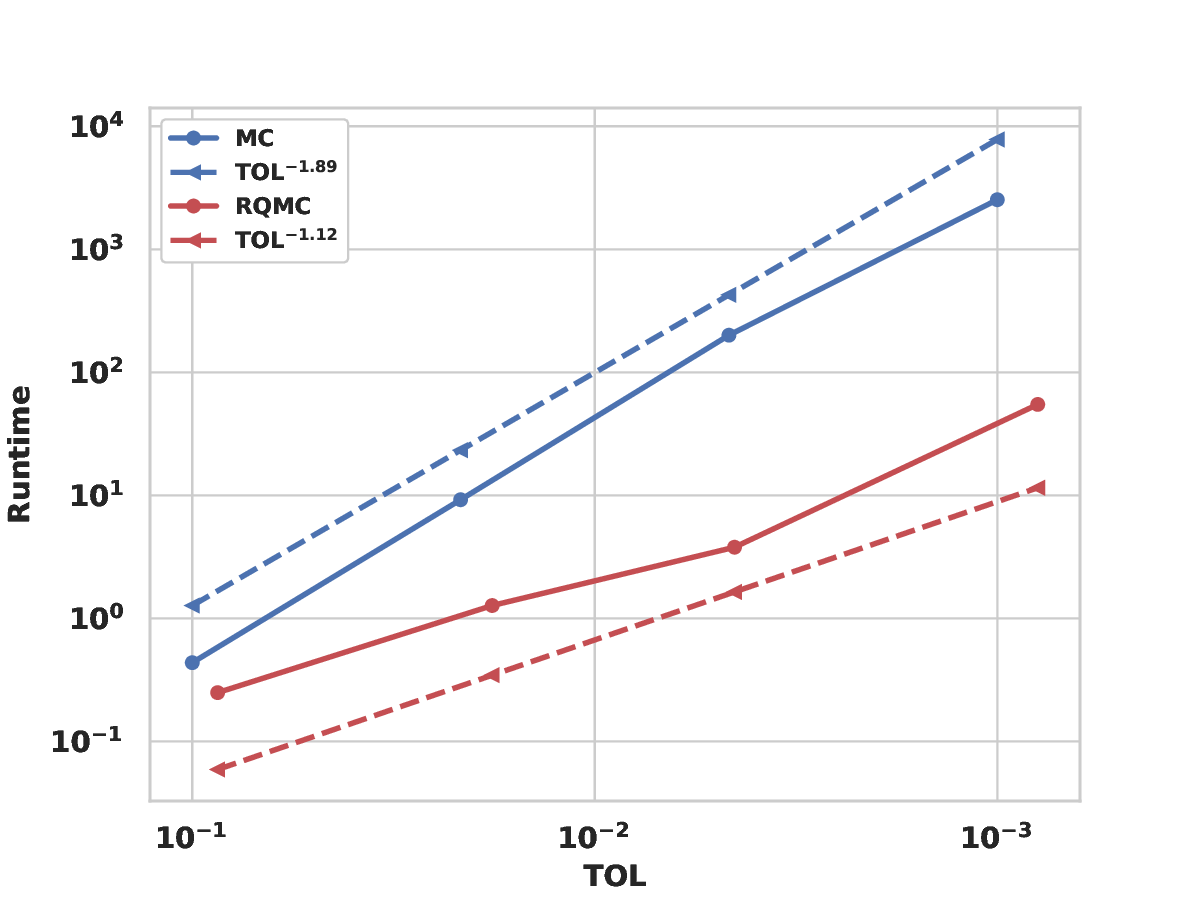}
		\caption{6D-CON call}
		\label{fig:phymc_vs_fourqmc_vg_CON_call}
	\end{subfigure}
		\hfill
	\begin{subfigure}{0.4\textwidth}
		\includegraphics[width=\linewidth]{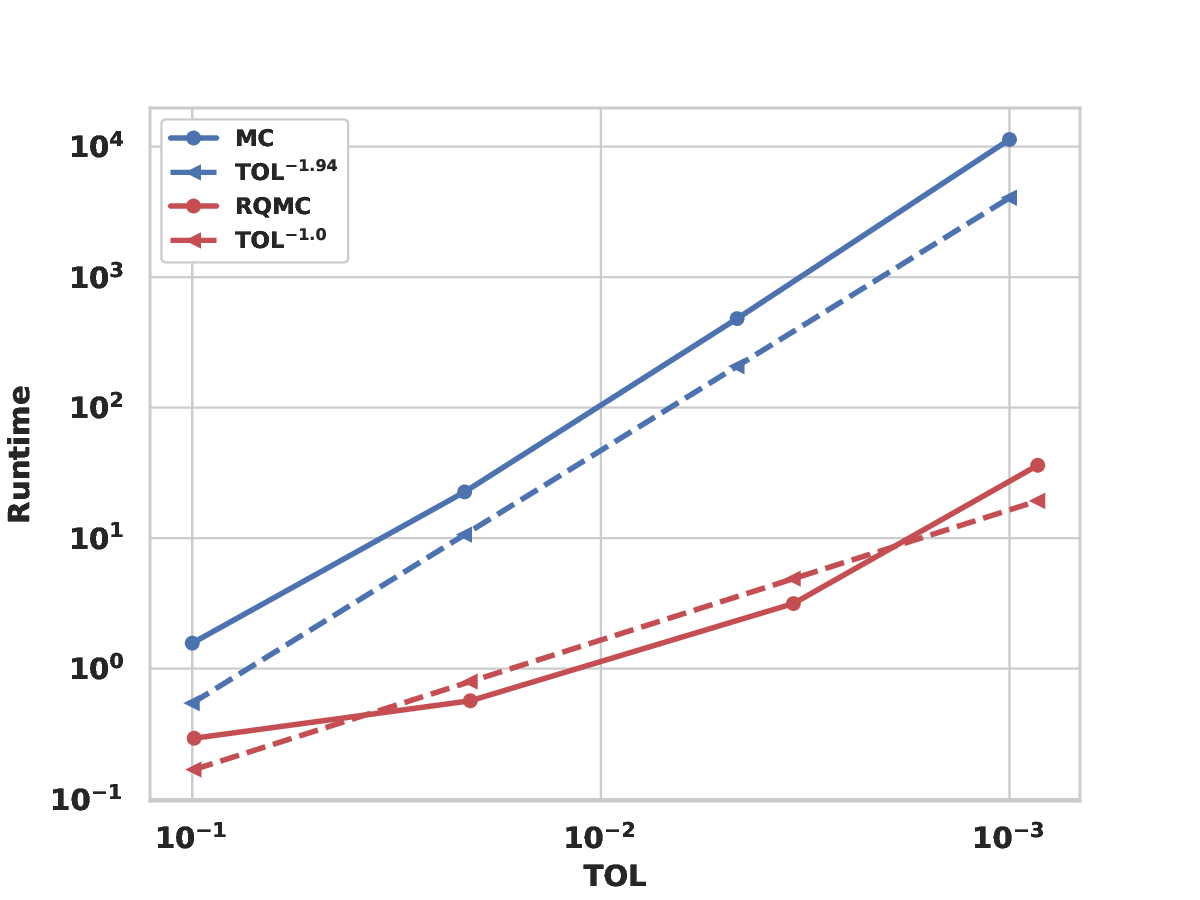}
		\caption{6D-call on min}
		\label{fig:phymc_vs_fourqmc_vg_CON_call_on_min}
	\end{subfigure}
\caption[]{VG: Average runtime in seconds with respect to relative tolerance levels TOL for (a) the six-dimensional CON call and (b) six-dimensional call on min option with for $S_0^j = 100$, $K = 100, r = 0, T = 1$, $\sigma_j =   0.4$, $\theta_j = -0.3$, $\nu = 0.1$ for all $j = 1,\ldots, 6$, and $\boldsymbol{\Sigma}_{ij} =   \rho_{ij}\sigma_i \sigma_j $ with  $\rho_{ij} = \frac{0.2}{1 + 0.1 | i - j|}$. The used domain transformation parameters are as in Table \ref{tab:multivariate_dom_tranf}. }  	\hspace{-3cm}
\label{fig:vg_mcphys_vs_qmcfour}
\end{figure}
\FloatBarrier

\FloatBarrier
\begin{figure}[htbp!] 
	\centering
	\begin{subfigure}{0.4\textwidth}
		\includegraphics[width=\linewidth]{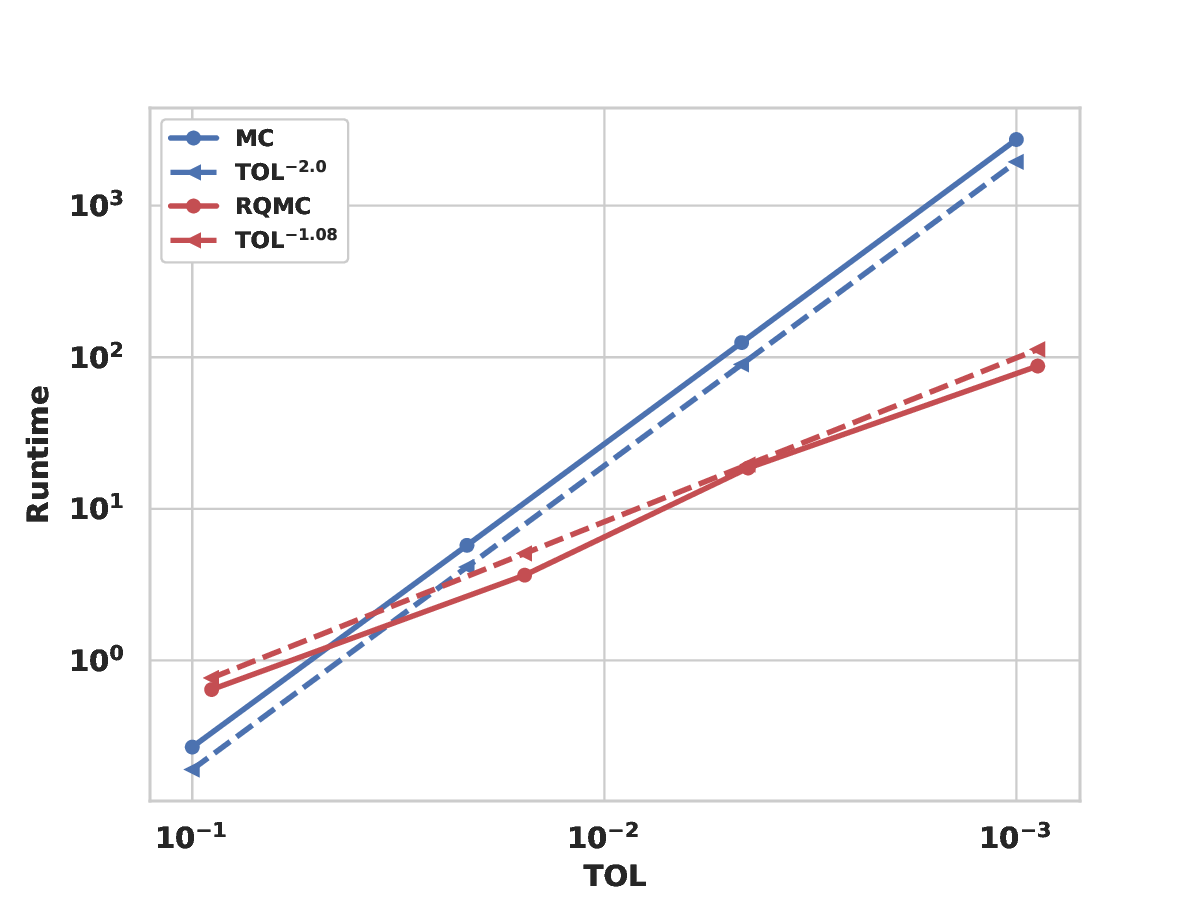}
		\caption{3D-basket put}
		\label{fig:phymc_vs_fourqmc_nig_basket_put}
	\end{subfigure}
	\hfill
	\begin{subfigure}{0.4\textwidth}
		\includegraphics[width=\linewidth]{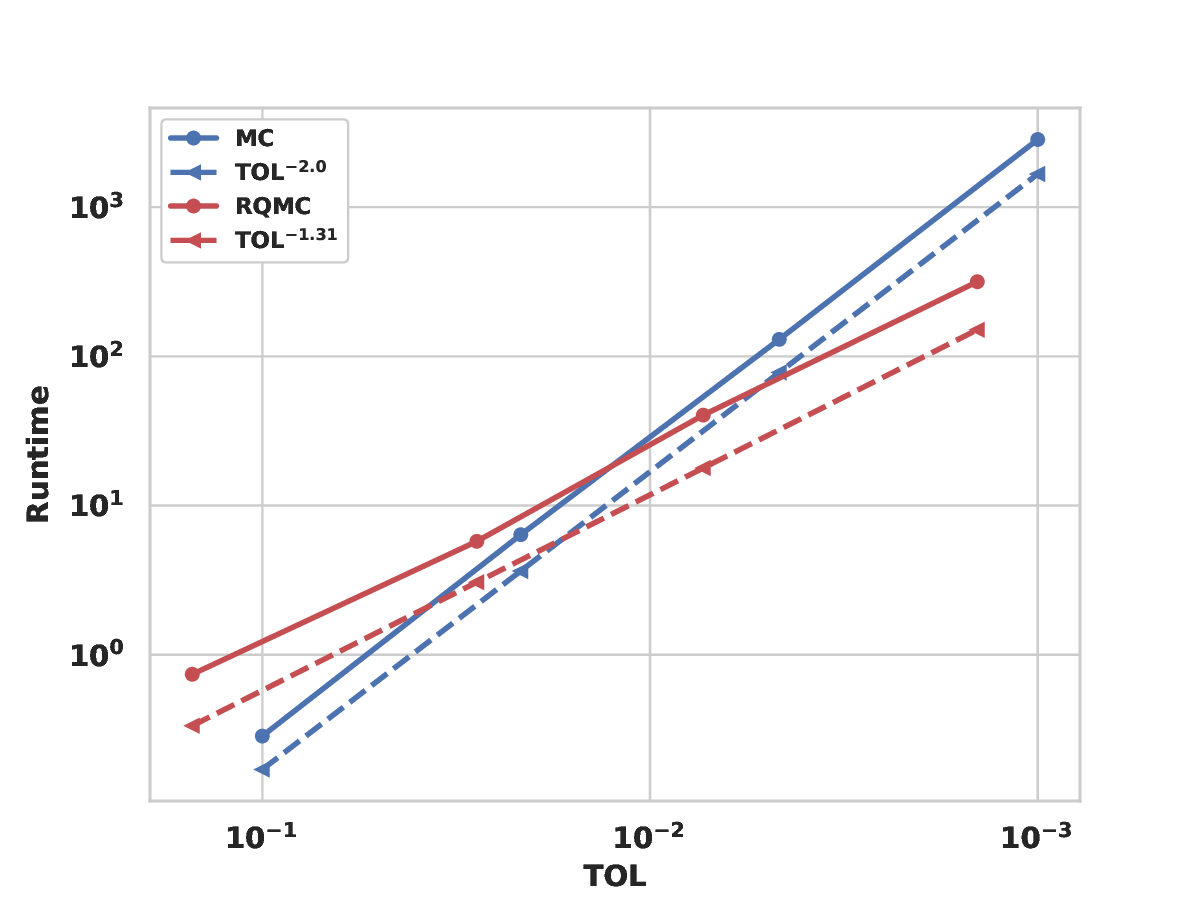}
		\caption{3D-spread call}
		\label{fig:phymc_vs_fourqmc_nig_spread_call}
	\end{subfigure}
	\caption[]{GH: Average runtime in seconds with respect to relative tolerance levels TOL for (a) the three-dimensional basket put with $S_0^j = 100$ and $K = 100$ for $j = 1,\ldots, 3$ and (b) three-dimensional spread call with $S_0^1= 100$, $S_0^j= \frac{100}{3}$ for $j = 1,2$, and $K = \frac{100}{3}$. Both experiments are done with parameters $r = 0, T = 1$, $\alpha = 10$, $\beta_j = -3$, $\delta = 0.2, \lambda = -\frac{1}{2}$ for all $j = 1,\ldots, 3$, and $\boldsymbol{\Delta} = \boldsymbol{I_3}$. The used domain transformation parameters are as in Table~\ref{tab:multivariate_dom_tranf}.} 
	\label{fig:gbm_mcphys_vs_qmcfour}
\end{figure}

\FloatBarrier

\begin{figure}[htbp!] 
	\centering
	\begin{subfigure}{0.4\textwidth}
		\includegraphics[width=\linewidth]{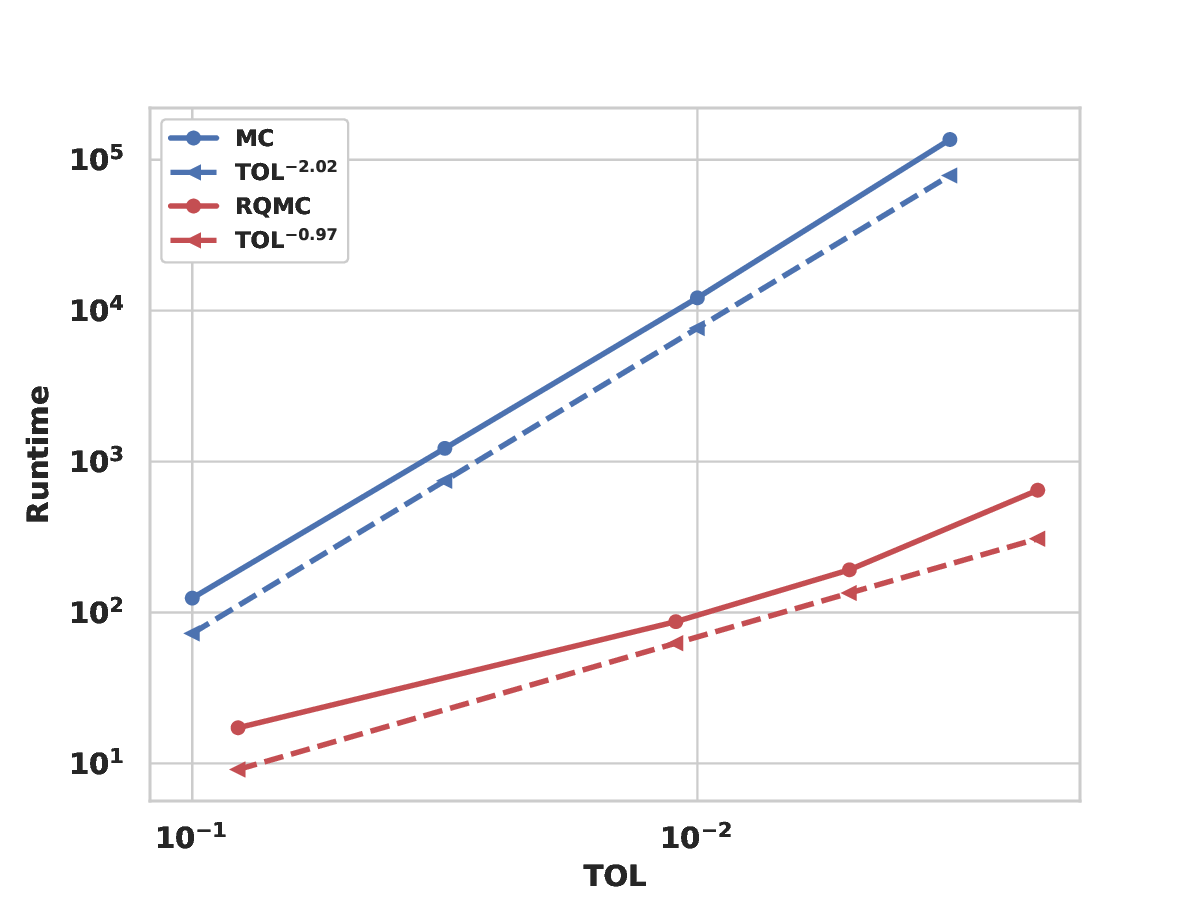}
		\caption{3D-basket put}
		\label{fig:phymc_vs_fourqmc_nig_basket_put_otm}
	\end{subfigure}
	\hfill
	\begin{subfigure}{0.4\textwidth}
		\includegraphics[width=\linewidth]{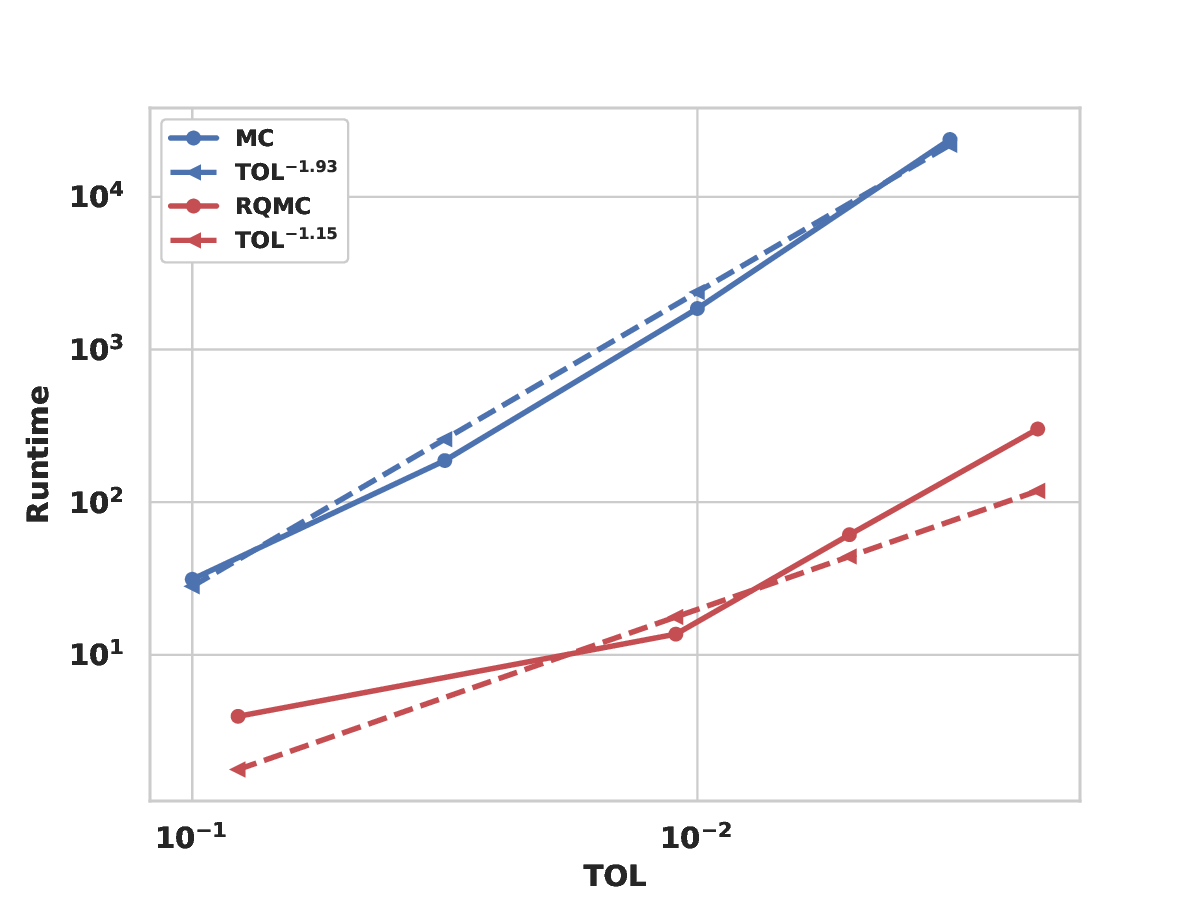}
		\caption{3D-spread call}
		\label{fig:phymc_vs_fourqmc_nig_spread_call_otm}
	\end{subfigure}
	\caption[]{GH: Average runtime in seconds with respect to relative tolerance levels TOL for (a) the three-dimensional basket put with $S_0^j = 100$ and $K = 60$ for $j = 1,\ldots, 3$ and (b) three-dimensional spread call with $S_0^1= 100$, $S_0^j= 50$ for $j = 1,2$, and $K = 50$. Both experiments are done with parameters $r = 0, T = 1$, $\alpha = 10$, $\beta_j = -3$, $\delta = 0.2, \lambda = -\frac{1}{2}$ for all $j = 1,\ldots, 3$, and $\boldsymbol{\Delta} = \boldsymbol{I_3}$. The used domain transformation parameters are as in Table~\ref{tab:multivariate_dom_tranf}.} 
	\label{fig:gbm_mcphys_vs_qmcfour_otm}
\end{figure}

\FloatBarrier
\vspace{-1cm}
\subsection{Runtime Comparison of RQMC with the MC and TP Quadratures }

\label{sec:qmc_vs_quad_run}
This section aims to compare the computational efficiency of the RQMC method with the commonly employed MC method in the physical space and the TP Gauss--Laguerre quadrature in the Fourier domain \cite{wiktorsson2015notes}. The runtimes in Figure \ref{fig:cpu_qmc_mc_quad}  are the average times in seconds of seven runs for each of the methods to achieve a relative tolerance $\text{TOL} = 10^{-2}$. 
 For the TP approach, only CPU times of up to five dimensions are measured in Python, and the values for the higher dimensions are numerically extrapolated due to the very slow convergence. For the MC and RQMC methods, the criterion for error convergence is the relative statistical error being less than the relative tolerance of $\text{TOL} = 10^{-2}$ . In contrast, for the TP quadrature, the stopping criterion is based on the exact relative error. The exact relative error is defined as the normalized absolute difference between the TP quadrature estimate and reference value computed using the MC method with $M = 10^9$ samples. Consequently, the statistical error of the MC and RQMC methods is an upper bound; thus, the CPU times for the MC and RQMC methods are conservative because, in practice, they converge faster with respect to the exact relative error. Figures~~\ref{fig:cpu_nig_call_on_min}, \ref{fig:cpu_nig_CON_call}, \ref{fig:cpu_vg_call_on_min}, \ref{fig:cpu_vg_CON_call} illustrate that the RQMC method applied in the Fourier space alleviates the curse of dimensionality, in contrast to the TP quadrature rule for which the cost grows exponentially with the dimensions. If the contour of integration is appropriately chosen and the domain transformation is handled carefully based on the proposed approach, the RQMC method significantly outperforms the MC method and TP quadrature for options with up to 15 underlying assets for the call on min and CON call options under GH and VG models.  In addition, although the convergence rate of the MC method is dimension-independent, the implied error constant increases with the dimensions. As a result, the RQMC approach reaches the target relative tolerance about 100 times faster than the MC method in the case of call on min options, and 1000 times faster in the case of CON call options. 
 
 {
 The largest tested dimension $d=15$ does not represent a  limitation of the proposed Fourier-RQMC methodology. It was chosen as an example that is already beyond the typical practical range of Fourier methods, while still allowing us to compute reference values and measure runtimes across several models and payoffs. At this dimension, numerical experiments already require substantial compute, and obtaining sufficiently accurate reference values becomes challenging, in particular when the reference is computed by physical-space Monte Carlo. Moreover, the results across increasing dimensions and across the three model classes show a consistent qualitative behaviour that the MC error decreases at the standard statistical rate, whereas the proposed Fourier-RQMC method benefits from the integration of the smoother Fourier integrand  with the appropriate domain transformation. These observed trends provide a basis for extrapolating the relative behaviour of the two methods beyond the tested dimensions, although not as a rigorous substitute for a separate high-dimensional benchmarking study. Extending the work to substantially larger dimensions, e.g., $d=50$, would require an optimized implementation, relying for instance on vectorization techniques or parallelization, and a more optimized methodology to obtain accurate reference values. }

\FloatBarrier
\begin{figure}[h!]
	\centering
	\begin{subfigure}[t]{0.45\textwidth}
		\centering
		\includegraphics[width=\linewidth]{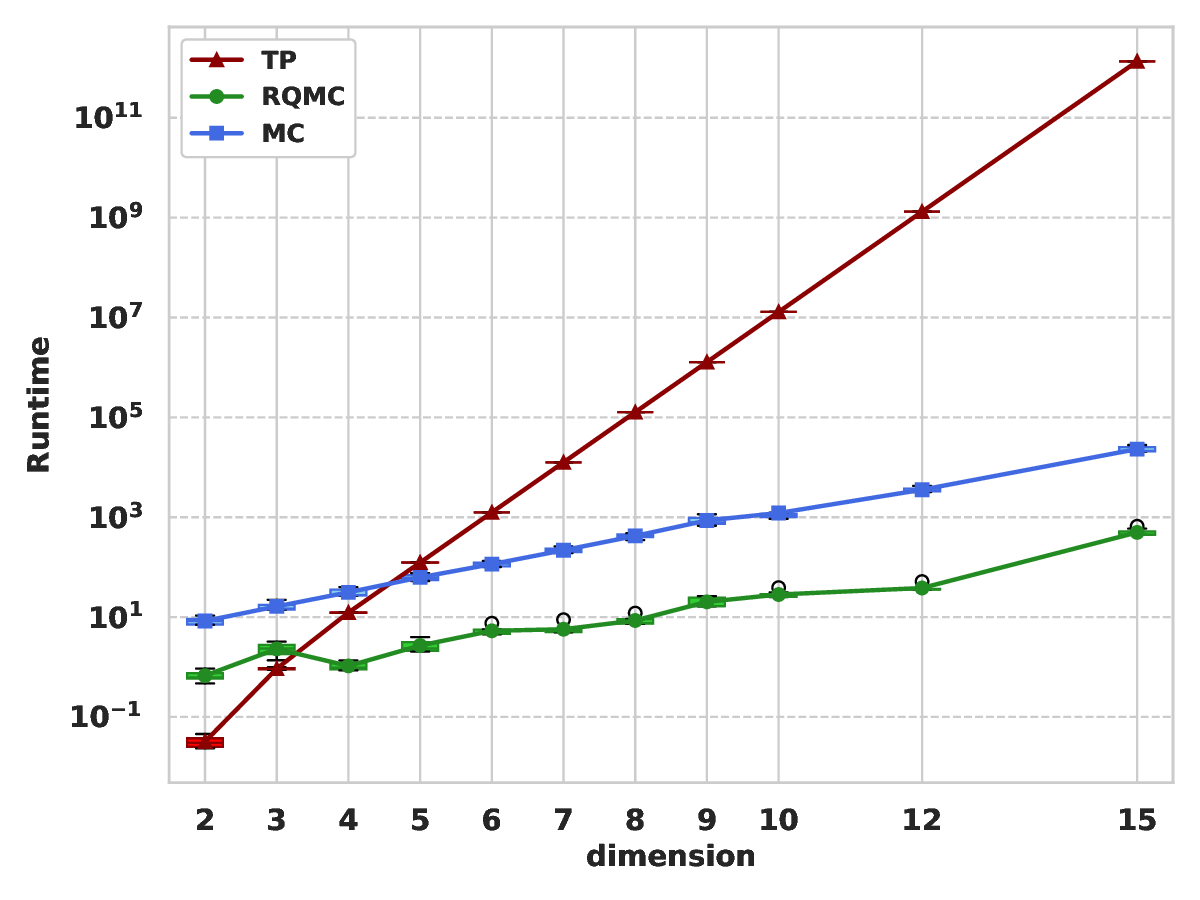}
		\caption{GH: call on min}
		\label{fig:cpu_nig_call_on_min} 
	\end{subfigure}
	\hfill
	\begin{subfigure}[t]{0.45\textwidth}
		\centering
		\includegraphics[width=\linewidth]{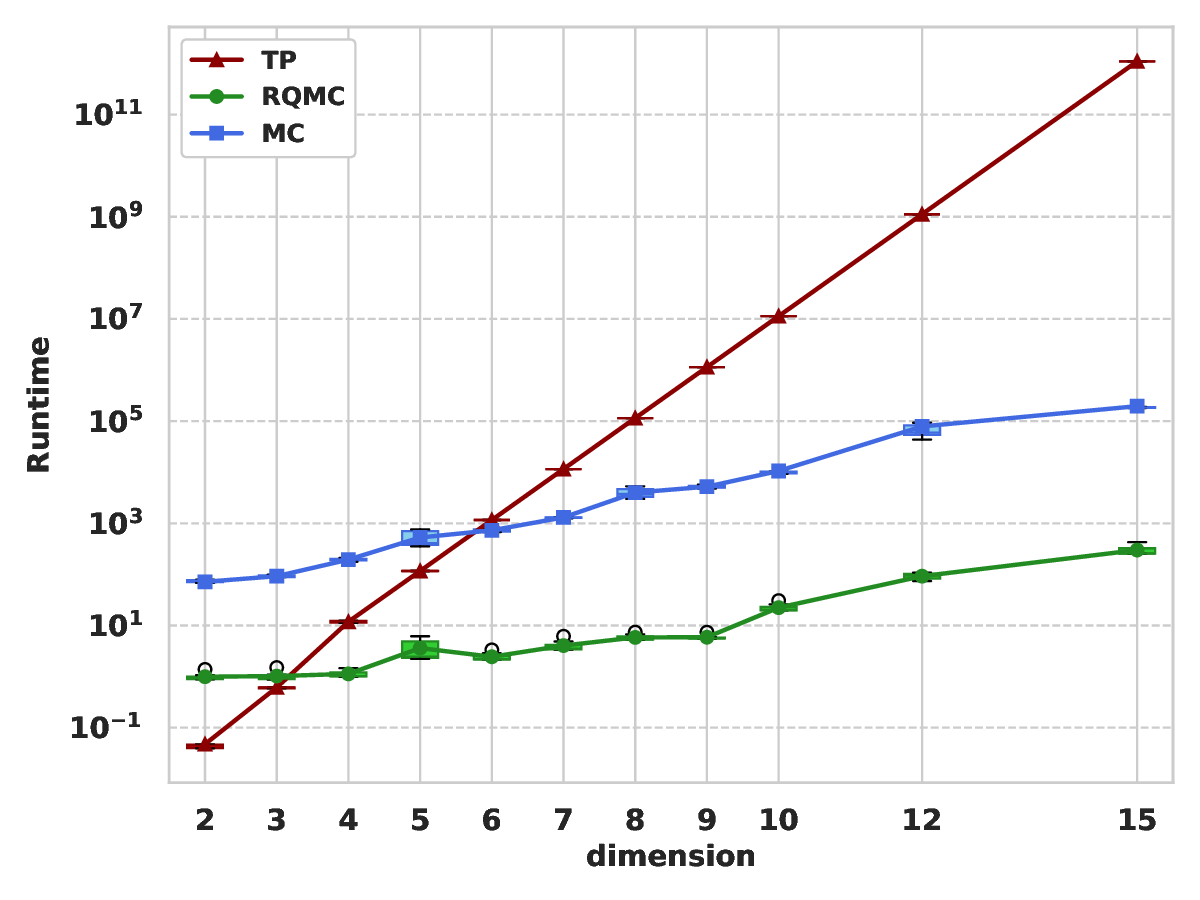}
		\caption{GH: CON call}
		\label{fig:cpu_nig_CON_call}
	\end{subfigure}
	
	\vspace{0.5cm} 
	\begin{subfigure}[t]{0.45\textwidth}
		\centering
		\includegraphics[width=\linewidth]{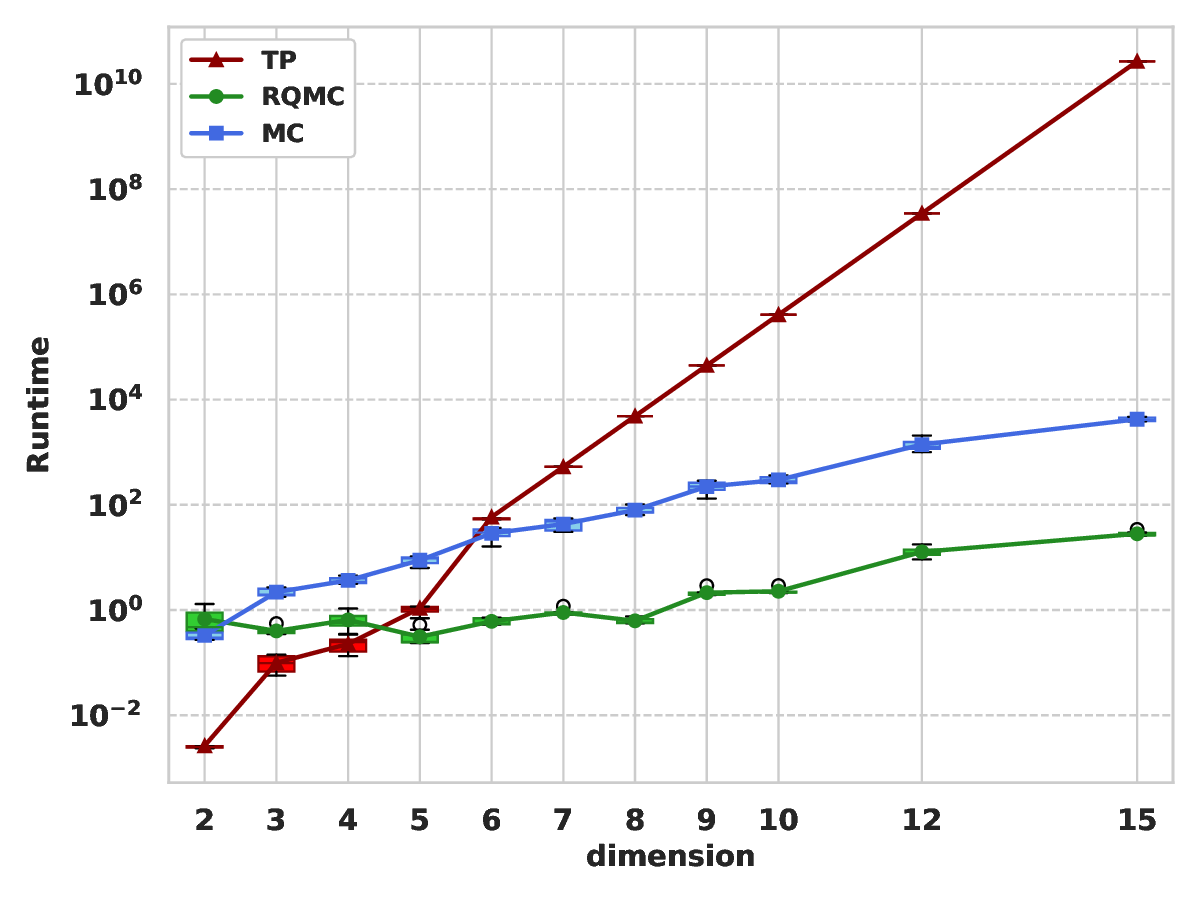}
		\caption{VG: call on min}
		\label{fig:cpu_vg_call_on_min}
	\end{subfigure}
	\hfill
	\begin{subfigure}[t]{0.45\textwidth}
		\centering
		\includegraphics[width=\linewidth]{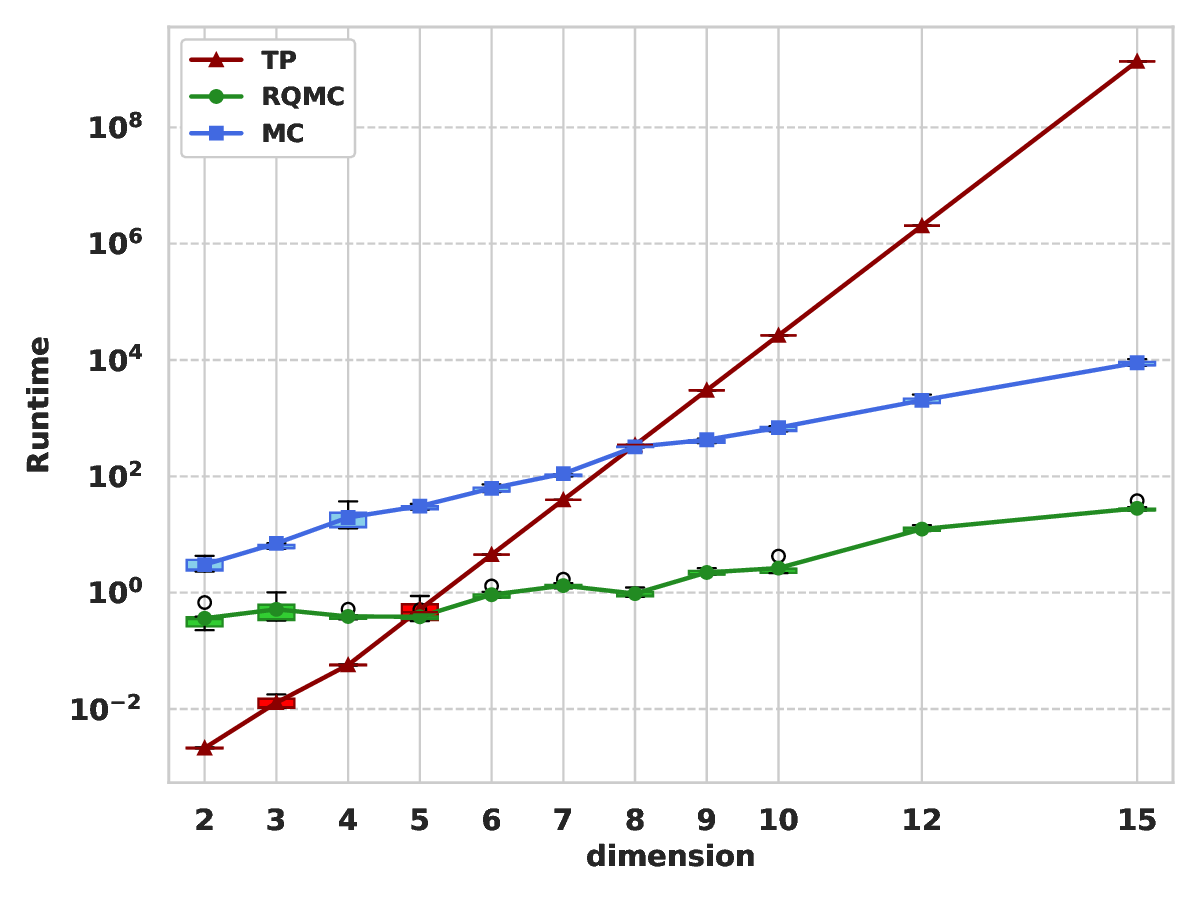}
		\caption{VG: CON call}
		\label{fig:cpu_vg_CON_call}
	\end{subfigure}
	
\caption[]{Comparison of runtime (in seconds) for the RQMC and TP quadrature methods in Fourier space, and the MC method in physical space, to achieve a relative error of $\text{TOL} = 10^{-2}$ across varying dimensions. Results are presented for two models: (a,b) the GH model with parameters $\alpha = 12$, $\beta_j = -3$, $\delta = 0.2, \lambda  = -\frac{1}{2}$, and $\boldsymbol{\Delta} = \boldsymbol{I}_d$ for (a) call on min and (b) CON call payoffs; (c,d) the VG model with parameters $\sigma_j = 0.4$, $\theta_j = -0.3$, $\nu = 0.1$, and $\boldsymbol{\Sigma} = \boldsymbol{I}_d$ for (c) call on min and (d) CON call payoffs. All experiments used $S_0^j = 100$, $K = 100$, $r = 0$, and $T = 1$ for all $j = 1, \ldots, d$. Domain transformations were applied as specified in Table~\ref{tab:multivariate_dom_tranf}, and optimal damping parameters were selected following the guidelines in \cite{bayer2023optimal}. The RQMC method employed $S = 30$ digital shifts.}

	\label{fig:cpu_qmc_mc_quad}
\end{figure}
\FloatBarrier


\textbf{Acknowledgments} C. Bayer gratefully acknowledges support from the German Research Foundation (DFG) via the Cluster of Excellence MATH+ (Project AA4-2). This publication is based on work supported by the King Abdullah University of Science and Technology (KAUST) Office of Sponsored Research (OSR) under Award No. OSR-2019-CRG8-4033 and the Alexander von Humboldt Foundation. Antonis Papapantoleon gratefully acknowledges the financial support from the Hellenic Foundation for Research and Innovation (Grant No. HFRI-FM17-2152). M. Samet acknolwedges the support from the Helmholtz School for Data Science in Life, Earth and Energy (HDS-LEE).

 \textbf{Declarations of Interest} The authors report no conflicts of interest. The authors alone are responsible for the content and writing of the paper.


\bibliographystyle{plain}
\bibliography{bibliography} 
\appendix

\section{Proof of Proposition \ref{prop:Multivariate Fourier pricing valuation formula}}
\label{sec:proof_of_fourier_valuation_formula}
Suppose Assumption \ref{ass:Assumptions on  the distribution} holds, then there exists $\boldsymbol{R} \in \delta_X \subseteq \mathbb{R}^d$, such that the Fourier transform of the exponentially dampened conditional transition probability density function of $\boldsymbol{X}_T$  is given by:

\begin{equation}
	\left(	\widehat{ e^{ - \boldsymbol{R}^\top \boldsymbol{x}} \rho_{\boldsymbol{X}_T} } \right)(\boldsymbol{y}) = \int_{\mathbb{R}^d} e^{- \mathrm{i}  \boldsymbol{y}^\top \boldsymbol{x}} e^{ - \boldsymbol{R}^\top \boldsymbol{x}} \rho_{\boldsymbol{X}_T} (\boldsymbol{x}) \mathrm{d} \boldsymbol{x} = \Phi_{\boldsymbol{X}_T}(\mathrm{i} \boldsymbol{R}-\boldsymbol{y}), \quad \boldsymbol{y} \in \mathbb{R}^d.
\end{equation}

In addition, under Assumption \ref{ass:Assumptions on  the distribution}, the characteristic function $ \boldsymbol{y} \mapsto \Phi_{\boldsymbol{X}_T}(\mathrm{i} \boldsymbol{R} - \boldsymbol{y}) \in L^1(\mathbb{R}^d)$. Consequently, the inverse Fourier transform theorem enables us to express the density function as:

\begin{equation}
	\label{eq:density_inverse_fourier}
	\rho_{\boldsymbol{X}_T} (\boldsymbol{x})  = e^{ \boldsymbol{R}^\top \boldsymbol{x}} \left( 2 \pi \right)^{-d} \Re \left[ \int_{\mathbb{R}^d} e^{\mathrm{i}  \boldsymbol{y}^\top \boldsymbol{x}} \Phi_{\boldsymbol{X}_T}(\mathrm{i} \boldsymbol{R}-\boldsymbol{y}) \mathrm{d} \boldsymbol{y} \right].
\end{equation}

The European option price is given as follows:
\begin{equation}
	\begin{aligned}
		V(\boldsymbol{\Theta}_X, \boldsymbol{\Theta}_P) &:= e^{-rT}\mathbb{E}[P(\boldsymbol{X}_T)] = e^{-rT}\int _{\mathbb{R}^d} P(\boldsymbol{x})  	\rho_{\boldsymbol{X}_T} (\boldsymbol{x}) \mathrm{d}  \boldsymbol{x} \\
		&= e^{-rT}\int_{\mathbb{R}^d}  P(\boldsymbol{x}) \left( e^{ \boldsymbol{R}^\top \boldsymbol{x}} \left( 2 \pi \right)^{-d}  \Re \left[\int_{\mathbb{R}^d} e^{\mathrm{i}  \boldsymbol{y}^\top \boldsymbol{x}} \Phi_{\boldsymbol{X}_T}(\mathrm{i} \boldsymbol{R}-\boldsymbol{y}) \mathrm{d} \boldsymbol{y} \right] \right)  \mathrm{d}  \boldsymbol{x} \\
		& =  e^{-rT}\left( 2 \pi \right)^{-d}   \Re \left[\int_{\mathbb{R}^d}  \Phi_{\boldsymbol{X}_T}(\mathrm{i} \boldsymbol{R}-\boldsymbol{y})  \left( \int_{\mathbb{R}^d}  e^{\mathrm{i}  \boldsymbol{y}^\top \boldsymbol{x}} e^{ \boldsymbol{R}^\top \boldsymbol{x}}  P(\boldsymbol{x}) \mathrm{d}\boldsymbol{x}   \right)  \mathrm{d} \boldsymbol{y} \right]  \\
		&=   e^{-rT}\left( 2 \pi \right)^{-d}  \Re \left[\int_{\mathbb{R}^d}  \Phi_{\boldsymbol{X}_T}(\mathrm{i} \boldsymbol{R}-\boldsymbol{y})  \hat{P}(\mathrm{i} \boldsymbol{R}-\boldsymbol{y}) \mathrm{d} \boldsymbol{y} \right] \\
		& =  e^{-rT}\left( 2 \pi \right)^{-d}  \Re \left[\int_{\mathbb{R}^d}  \Phi_{\boldsymbol{X}_T}(\boldsymbol{y} + \mathrm{i} \boldsymbol{R})  \hat{P}(\boldsymbol{y} + \mathrm{i} \boldsymbol{R}) \mathrm{d} \boldsymbol{y} \right]
	\end{aligned}
\end{equation}

The derivation proceeds by substituting Equation \eqref{eq:density_inverse_fourier} in the second line, and by further restricting the contour of integration to $\boldsymbol{R} \in \delta_V := \delta_P \cap \delta_X \neq \emptyset$. Application of Fubini's theorem in the third line permits the interchange of integration order. The final line follows from the change of variables $\boldsymbol{y}\mapsto-\boldsymbol{y}$ in the integral over $\mathbb{R}^d$.

The application of Fubini's theorem is validated by demonstrating the absolute integrability of:

\begin{equation}
	\begin{aligned}
		\int_{\mathbb{R}^d} \int_{\mathbb{R}^d } | \Phi_{\boldsymbol{X}_T}(\mathrm{i} \boldsymbol{R}-\boldsymbol{y}) | | e^{\mathrm{i}  \boldsymbol{y}^\top \boldsymbol{x}} | e^{ \boldsymbol{R}^\top \boldsymbol{x}}  P(\boldsymbol{x}) \mathrm{d}\boldsymbol{y}  \mathrm{d}\boldsymbol{x} & \leq \int_{\mathbb{R}^d} e^{ \boldsymbol{R}^\top \boldsymbol{x}}  P(\boldsymbol{x})  \left( \int_{\mathbb{R}^d}  | \Phi_{\boldsymbol{X}_T}(\mathrm{i} \boldsymbol{R}-\boldsymbol{y}) |  \mathrm{d}\boldsymbol{y}\right)\mathrm{d}\boldsymbol{x} \\
		& \leq C  \int_{\mathbb{R}^d} e^{ \boldsymbol{R}^\top \boldsymbol{x}}  P(\boldsymbol{x}) \mathrm{d}\boldsymbol{x} < + \infty
	\end{aligned}
\end{equation}

Here, \(C > 0\) exists by virtue of Assumption \ref{ass:Assumptions on  the distribution}, while the finiteness of the final bound is guaranteed by Assumption \ref{ass:Assumptions on  the payoff}.

\section{Pricing Models}
\label{sec:pricing_models}
For the asset dynamics in this work, we studied three models given by Examples~\ref{ex:GBM_model}, \ref{ex:NIG_model}, \ref{ex:GH_model}, and \ref{ex:VG_model}.

\begin{example}[Geometric Brownian Motion (GBM)]
	\label{ex:GBM_model}
	The discounted characteristic function of the GBM model under the risk-neutral pricing measure is given in the form of $ \Phi^{\text{GBM}}_{\boldsymbol{X}_T}(\boldsymbol{z}) = e^{\mathrm{i} \boldsymbol{z}^{\top} ( \boldsymbol{X}_0  + (r + \boldsymbol{\mu}_{\text{GBM}}) T ) } \phi^{\text{GBM}}_{\boldsymbol{X}_T}(\boldsymbol{z})$ for $\boldsymbol{z} = \boldsymbol{y} + \mathrm{i} \boldsymbol{R} \in \mathbb{C}^d $ with $\boldsymbol{R}  \in \delta_X^{\text{GBM}} = \mathbb{R}^d$, by
	\begin{equation}
		\phi^{\text{GBM}}_{\boldsymbol{X}_T}(\boldsymbol{z}) = \exp \left(-\frac{T}{2} \boldsymbol{z}^{\top} \boldsymbol{\Sigma} \boldsymbol{z}\right), \quad \Im[\boldsymbol{z}] \in \delta_X^{\text{GBM}}.
	\end{equation}
where $\boldsymbol{\sigma} = (\sigma_1,\ldots, \sigma_d) \in \mathbb{R}^d_+$ is the vector of volatilities, and we denote by $\boldsymbol{\Sigma} \in \mathbb{R}^{d \times d}$ the covariance matrix of the log returns i.e. $\boldsymbol{\Sigma}_{ij}=\rho_{i,j} \sigma_i \sigma_j$, with $\rho_{i,j}$ denoting the correlation between the Brownian motions of the $i^{th}$ and $j^{th}$ asset price processes. Moreover, $\boldsymbol{\mu}_{\text{GBM}}$ is a vector of drift correction terms that ensure that $\left\{e^{-r t} S^j_t\right\}_{ t \geq 0}$ is a martingale for all $j = 1,\ldots d$, and is given by
	$$
	\mu_{\text{GBM}}^j  = - \frac{\sigma_j^2}{2}, \quad  j =1, \ldots, d.
	$$

\end{example}

\begin{example}[Normal Inverse Gaussian (NIG)] 
	\label{ex:NIG_model}
\sloppy	The discounted characteristic function of the NIG model under the risk-neutral pricing measure is given in the form of
	$ \Phi^{\text{NIG}}_{\boldsymbol{X}_T}(\boldsymbol{z}) = e^{\mathrm{i} \boldsymbol{z}^{\top} ( \boldsymbol{X}_0  + (r + \boldsymbol{\mu}_{\text{NIG}}) T ) } \phi^{\text{NIG}}_{\boldsymbol{X}_T}(\boldsymbol{z})$ for $\boldsymbol{z} = \boldsymbol{y} + \mathrm{i} \boldsymbol{R} \in \mathbb{C}^d $ with $\boldsymbol{R}  \in \delta_X^{\text{NIG}} = \left\{\boldsymbol{R} \in \mathbb{R}^d \mid \alpha^2-(\boldsymbol{\beta}-\boldsymbol{R})^{\top} \boldsymbol{\Delta}(\boldsymbol{\beta}-\boldsymbol{R})>0\right\}$, for $\quad \Im[\boldsymbol{z}] \in \delta_X^{\text{NIG}}$, by  \cite{prause1999generalized}
	\begin{equation}
		\phi^{\text{NIG}}_{\boldsymbol{X}_T}(\boldsymbol{z}) = \exp \left(\delta T\left(\sqrt{\alpha^2-\boldsymbol{\beta}^{\top} \boldsymbol{\Delta} \boldsymbol{\beta}}-\sqrt{\alpha^2-(\boldsymbol{\beta}+\mathrm{i} \boldsymbol{z})^{\top} \boldsymbol{\Delta}(\boldsymbol{\beta}+\mathrm{i} \boldsymbol{z})}\right)\right)
	\end{equation}
where $\alpha \in \mathbb{R}_{+}, \delta \in \mathbb{R}_{+}, \boldsymbol{\beta} \in \mathbb{R}^d$ with $\alpha^2>\boldsymbol{\beta}^{\mathrm{T}} \boldsymbol{\Delta} \boldsymbol{\beta},$ and $\boldsymbol{\Delta} \in \mathbb{R}^{d \times d}$ is a symmetric positive definite matrix with a unit determinant i.e. $| \det(\boldsymbol{\Delta})  | = 1$, related to the covariance matrix of the log returns as follows \cite{eberlein2010analysis}
	$$
	\boldsymbol{\Sigma} = \delta \left(\alpha^2 - \boldsymbol{\beta}^{\top} \boldsymbol{\Delta} \boldsymbol{\beta} \right)^{-\frac{1}{2}} \left( \boldsymbol{\Delta} + \left(\alpha^2 - \boldsymbol{\beta}^{\top} \boldsymbol{\Delta} \boldsymbol{\beta}\right)^{-1} \boldsymbol{\Delta} \boldsymbol{\beta} \boldsymbol{\beta}^{\top} \boldsymbol{\Delta} \right)
	$$
	Moreover, $\boldsymbol{\mu}_{\text{NIG}}$ is a vector of drift correction terms that ensures that $\left\{e^{-r t} S^j_t \mid t \geq 0\right\}$ is a martingale for all $j = 1,\ldots d$, and is given by
	$$
	\mu_{\text{NIG}}^j=-\delta\left(\sqrt{\alpha^2-\beta_j^2}-\sqrt{\alpha^2-\left(\beta_j+1\right)^2}\right), \quad j = 1, \ldots, d .
	$$

\end{example}

\begin{example}[Generalized Hyperbolic (GH)]
	\label{ex:GH_model}
	\sloppy
	The discounted characteristic function of the GH model under the risk-neutral  measure is given in the form 
	$\Phi^{\text{GH}}_{\boldsymbol{X}_T}(\boldsymbol{z}) = e^{\mathrm{i} \boldsymbol{z}^{\top} ( \boldsymbol{X}_0  + ( r + \boldsymbol{\mu}_{\text{GH}}) T ) }\phi^{\text{GH}}_{\boldsymbol{X}_T}(\boldsymbol{z})$	for $\boldsymbol{z} = \boldsymbol{y} + \mathrm{i} \boldsymbol{R}$, with $\boldsymbol{R} \in \delta_X^{\text{GH}}= \left\{\boldsymbol{R} \in \mathbb{R}^d \mid \alpha^2-(\boldsymbol{\beta}-\boldsymbol{R})^{\top} \boldsymbol{\Delta}(\boldsymbol{\beta}-\boldsymbol{R})>0\right\}$   for 	\quad $\Im[\boldsymbol{z}] \in \delta_X^{\text{GH}}$ by \cite{prause1999generalized}
	
	\begin{equation}
		\phi^{\text{GH}}_{\boldsymbol{X}_T}(\boldsymbol{z}) =\left(\frac{\alpha^2-\boldsymbol{\beta}^{\top} \boldsymbol{\Delta} \boldsymbol{\beta}}{\alpha^2  -(\boldsymbol{\beta}+ \mathrm{i} \boldsymbol{z})^{\top}\boldsymbol{\Delta}(\boldsymbol{\beta}+ \mathrm{i} \boldsymbol{z})}\right)^{\lambda / 2} \frac{\mathrm{~K}_\lambda\left(\delta T \sqrt{\alpha^2-(\boldsymbol{\beta}+ \mathrm{i} \boldsymbol{z})^{\top}\boldsymbol{\Delta}(\boldsymbol{\beta}+ \mathrm{i} \boldsymbol{z})}\right)}{\mathrm{K}_\lambda\left(\delta T \sqrt{\alpha^2-\boldsymbol{\beta}^{\top} \boldsymbol{\Delta} \boldsymbol{\beta}}\right)}
	\end{equation}
	
	where $\alpha \in \mathbb{R}_{+}, \delta \in \mathbb{R}_{+}, \boldsymbol{\beta} \in \mathbb{R}^d, \lambda \in \mathbb{R}$ with $\alpha^2>\boldsymbol{\beta}^{\mathrm{T}} \boldsymbol{\Delta} \boldsymbol{\beta},$ where $\boldsymbol{\Delta} \in \mathbb{R}^{d \times d}$ is a symmetric positive definite matrix with a unit determinant i.e. $| \det(\boldsymbol{\Delta})  | = 1$, related to the covariance matrix of the log returns as follows \cite{eberlein2010analysis}
	
	$$
	\boldsymbol{\Sigma} = \delta \left(\alpha^2 - \boldsymbol{\beta}^{\top} \boldsymbol{\Delta} \boldsymbol{\beta} \right)^{-\frac{1}{2}} \left( \boldsymbol{\Delta} + \left(\alpha^2 - \boldsymbol{\beta}^{\top} \boldsymbol{\Delta} \boldsymbol{\beta}\right)^{-1} \boldsymbol{\Delta} \boldsymbol{\beta} \boldsymbol{\beta}^{\top} \boldsymbol{\Delta} \right)
	$$
	Moreover, is a vectorof drift correction terms that ensures that $\left\{e^{-r t} S^j_t \mid t \geq 0\right\}$ is a martingale for all $j = 1,\ldots d$. In the case of the GH model, we do not have an explicit expression for $\boldsymbol{\mu}_{\text{GH}}$, hence we compute it by evaluating the 1D characteristic function as follows
	$$
	\mu_{\text{GH}}^j =  - \frac{1}{T} \log\left( 	\phi^{\text{GH}}_{X_T^j}( - \mathrm{i}) \right), \quad j = 1,\ldots,d
	$$
	where
	\begin{equation}
		\phi^{\text{GH}}_{X_T^j}(z) = \left(\frac{\alpha^2-\beta^2}{\alpha^2-(\beta+i z)^2}\right)^{\lambda / 2} \frac{\mathrm{~K}_\lambda\left(\delta \sqrt{\alpha^2-(\beta+i z)^2}\right)}{\mathrm{K}_\lambda\left(\delta \sqrt{\alpha^2-\beta^2}\right)}, \quad \Im[z] \in \delta_X^{\text{GH}}.
	\end{equation}
	
	\begin{remark}
		The GH model coincides with the NIG model for  $\lambda = -\frac{1}{2}$ and coincides with the hyperbolic model for $\lambda = 1$  \cite{prause1999generalized}.
	\end{remark}

\end{example}

\begin{example}[Variance Gamma (VG)]
		\label{ex:VG_model}
	The discounted characteristic function of the VG model under the risk-neutral  measure is given in the form  $ \Phi^{\text{VG}}_{\boldsymbol{X}_T}(\boldsymbol{z}) = e^{\mathrm{i} \boldsymbol{z}^{\top} ( \boldsymbol{X}_0  + (r  + \boldsymbol{\mu}_{\text{VG}}) T ) } \phi^{\text{VG}}_{\boldsymbol{X}_T}(\boldsymbol{z})$  for $\boldsymbol{z} = \boldsymbol{y} + \mathrm{i} \boldsymbol{R} \in \mathbb{C}^d $ with $\boldsymbol{R}  \in \delta_X^{\text{VG}} = \left\{\boldsymbol{R} \in \mathbb{R}^d \left\lvert\, 1+\nu \boldsymbol{R}^{\top} \boldsymbol{\theta}-\frac{1}{2} \nu \boldsymbol{R}^{\top} \boldsymbol{\Sigma} \boldsymbol{R}>0\right.\right\} $, by \cite{luciano2006multivariate}
	
	\begin{equation}
		\phi^{\text{VG}}_{\boldsymbol{X}_T}(\boldsymbol{z}) =\left(1-\mathrm{i} \nu \boldsymbol{z}^{\top} \boldsymbol{\theta}+\frac{1}{2} \nu \boldsymbol{z}^{\top} \boldsymbol{\Sigma} \boldsymbol{z}\right)^{-T / \nu},  \quad \Im[\boldsymbol{z}] \in \delta_X^{\text{VG}}.
	\end{equation}
	
	where    $\boldsymbol{\sigma} = (\sigma_1,\ldots, \sigma_d) \in \mathbb{R}^d_+$, $\boldsymbol{\theta} = (\theta_1, \ldots, \theta_d) \in \mathbb{R}^d$, $\nu >0$, and $\boldsymbol{\Sigma} \in \mathbb{R}^{d \times d}$ denotes the covariance matrix of the log returns i.e. $\boldsymbol{\Sigma}_{i j}=\rho_{i, j} \sigma_i \sigma_j$  with $\rho_{i,j}$ denoting the correlation between the Brownian motions of the $i^{th}$ and $j^{th}$ asset price processes. Moreover, $\boldsymbol{\mu}_{V G}$ is a vector of drift correction terms that ensures that $\left\{e^{-r t} S^j_t \mid t \geq 0\right\}$ is a martingale for all $j = 1,\ldots d$, and is given by
	
	$$
	\mu_{V G}^j=\frac{1}{\nu} \log \left(1-\frac{1}{2} \sigma_j^2 \nu-\theta_j \nu\right), \quad j=1, \ldots, d .
	$$

\end{example}

\section{Strip of Analyticity of the Characteristic Functions}
This section presents the strip of analyticity of some examples of the characteristic functions considered in this work (see Table \ref{table:chf_table})
\FloatBarrier
\begin{table}[h]
	\centering
	\begin{tabular}{| p{2cm} | p{6cm}  |}
		\hline   \textbf{Model} &  $\delta_X$  \\
		\hline
		\textbf{GBM} &  \small $ \mathbb{R}^d $  \\ 
		\hline 
		\textbf{GH, NIG}   & \small $\{\boldsymbol{R} \in \mathbb{R}^d \; | \; \alpha^2 -  (\boldsymbol{\beta} -\boldsymbol{R})^{\prime} \boldsymbol{\Delta}( \boldsymbol{\beta} - \boldsymbol{R} )>0\}$  \\
		\hline
			\textbf{VG} & \small $\{\boldsymbol{R} \in \mathbb{R}^d \; | \; 1+ \nu \boldsymbol{R}^{\prime} \boldsymbol{\theta} -\frac{1}{2} \nu\boldsymbol{R}^{\prime} \boldsymbol{\Sigma} \boldsymbol{R}  > 0\}$  \\
		\hline
	\end{tabular}
	\caption{Strip of analyticity,  $\delta_X$, of the characteristic functions  for the different pricing models.}
	\label{table:strip_table} 
\end{table}

\section{Strip of Analyticity of the Fourier Transforms of the Payoff Functions}
This section presents the strip of analyticity of the Fourier transforms of some examples  of payoff functions considered in this work  (see Table \ref{table:payoffs}) 
\FloatBarrier
\begin{table}[h]
	\centering
	\begin{tabular}{| p{3cm} | p{8cm}  |}
		\hline   \textbf{Payoff} &   $\delta_P$  \\
		\hline
		\textbf{Basket put} &  \small $\{ \boldsymbol{R}\in \mathbb{R}^d  \; |  R_j > 0  \; \}$  \\ 
		\hline 
		\textbf{Spread call} & \small$\{ \boldsymbol{R}\in \mathbb{R}^d  \; | R_j >0, j = 2,\ldots d, R_1 <-1-\sum_{j=2}^d R_j   \; \}$ \\
		\hline
		\textbf{Call on min}   & \small $\{ \boldsymbol{R}\in \mathbb{R}^d  \; |  R_j < 0, \sum_{j = 1}^d R_j < -1   \; \}$  \\
		
			\hline
		\textbf{CON call}   & \small $\{ \boldsymbol{R}\in \mathbb{R}^d  \; |  R_j < 0  \; \}$  \\
		\hline
	\end{tabular}
	\caption{Strip of analyticity,  $\delta_P$, of the Fourier transforms of the payoff functions.}
	\label{table:payoff_strip_table}
\end{table}

\section[Semi-heavy-tailed characteristic functions: example of independent assets]{Domain transformation for semi-heavy-tailed characteristic functions: example of the GH model for independent assets}
\label{sec:domain_transf_gh_independent}
\paragraph{Product-form domain transformation}
Following the same line of reasoning as in  Section \newline ~\ref{sec:gbm_dom_transf}, we consider the setting of independent assets; hence, we have
\begin{equation}
	\phi^{GH}_{\boldsymbol{X}_T}(\boldsymbol{z}) = \prod_{j = 1}^d \phi^{GH}_{X^j_T}(z_j), \boldsymbol{z} \in \mathbb{C}^d, \Im[\boldsymbol{z} ] \in \delta_X^{GH},
\end{equation}
where \cite{prause1999generalized}
	\begin{equation}
		\label{eq:gh_1D_CF}
		\phi^{GH}_{X_T^j}(z_j) = 	\left(\frac{\alpha^2-\beta_j^2}{\alpha^2-(\beta_j+\mathrm{i} z_j)^2}\right)^{\lambda / 2}  \frac{K_\lambda\left(\delta T \sqrt{\alpha^2-(\beta_j+\mathrm{i} z_j)^2}\right)}{K_\lambda\left(\delta T \sqrt{\alpha^2-\beta_j^2}\right)},  \quad z_j \in \mathbb{C}.
\end{equation}
where $K_{\lambda}(\cdot)$ is the modified Bessel function of the second kind. We recall that the Bessel function satisfies the following relation
\begin{equation}
	K_{\lambda}(x) \stackrel{x\to +\infty}{\sim} \sqrt{\frac{\pi}{2x}} e^{-x}.
\end{equation}
Consequently, $\phi^{GH}_{X_T^j}(z_j)$ behaves asymptotically as a double-exponential function, i.e., \newline $|\phi^{GH}_{X_T^j}(z_j)|\leq C \exp(- \gamma |\Re[z_j]| )$ as $\Re[z_j] \to \pm \infty$.  In fact, we have that
\begin{equation}
\phi^{GH}_{X_T^j}(y_j + \mathrm{i} R_j) \stackrel{|y_j| \to \infty}{\sim} \left(\frac{\alpha^2 - \beta^2}{y^2_j} \right)^{\lambda / 2} \sqrt{\frac{\pi}{2 \delta T |y_j|}} \frac{ \exp(- \delta T |y_j|)}{K_\lambda\left(\delta T \sqrt{\alpha^2-\beta_j^2}\right)}.
\end{equation}
Hence, we choose the density for the domain transformation to be that of  the Laplace distribution, also known as the double exponential distribution, given by $\psi^{lap}(\boldsymbol{y}) = \prod_{j = 1}^d \psi^{lap}_j(y_j)$, where
\begin{equation}
	\label{eq:product_form_nig_density}
	\psi^{lap}_j(y_j) =  \frac{\exp(- \frac{\mid y_j \mid}{\tilde{\sigma}_j } )}{2 \tilde{\sigma}_j}, \quad  y_j \in \mathbb{R}, \tilde{\sigma_j} > 0.
\end{equation}
Focusing on the leading asymptotic terms, we encapsulate the polynomial prefactor in the following function
\begin{equation}
	Q^{GH}_j(y_j) :=	\left(\frac{\alpha^2 - \beta^2}{y^2_j} \right)^{\lambda / 2} \sqrt{\frac{\pi}{2 \delta T |y_j|}} \frac{1}{K_\lambda\left(\delta T \sqrt{\alpha^2-\beta_j^2}\right)}.
\end{equation}
Upon defining the functional form of $\psi^{lap}_j(u_j)$, the objective is to identify an appropriate selection for the parameters $\{\tilde{\sigma}_j\}_{j=1}^d$. To begin, we introduce the function $r^{GH}_{lap,j}(\cdot)$, representing the ratio of the characteristic function of the GH-distributed RV $X^j_T$ to the density $\psi^{lap}_j(\cdot)$: 
\begin{equation}
	r^{GH}_{lap,j}(\Psi_{lap}^{-1}(u_j)) :=  \frac{ \phi^{GH}_{X^j_T}(\Psi_{lap}^{-1}(u_j)+  \mathrm{i} R_j)}{	\psi^{lap}_j(\Psi_{lap}^{-1}(u_j))},  \quad u_j \in [0,1], \boldsymbol{R}  \in \delta_V^{GH}.
\end{equation}
Then, $r^{GH}_{lap,j}(\Psi_{lap}^{-1}(u_j))$ can be approximated as follows
\begin{equation}
	\small
	\label{eq:ratio_nig}
	\begin{aligned}
		r^{GH}_{lap,j}(\Psi_{lap}^{-1}(u_j)) &\stackrel{u_j \to  0}{\sim} \ 2 Q^{GH}_j(\Psi^{-1}_{lap}(u_j))   \underbrace{ \tilde{\sigma}_j \exp\left( - |\Psi^{-1}_{lap}(u_j)| \left(\delta T - \frac{1}{\tilde{\sigma}_j} \right) \right) }_{h_{lap,j}^{GH}(u_j)}.
	\end{aligned}
\end{equation}

From \eqref{eq:ratio_nig}, we outline three possible cases by focusing on the  limiting behavior of the term $h_{lap,j}^{GH}(\Psi_{lap}^{-1}(u_j))$ as $u_j \to 0$:
\begin{equation}
	\label{eq:mnig_conditions}
	\lim_{u_j \to 0} h_{lap,j}^{GH}(\Psi_{lap}^{-1}(u_j)) =
	\left\{
	\begin{array}{lll}
		+\infty & \text{if } \tilde{\sigma}_j < \frac{1}{\delta T} & (i), \\[0.5pt]
		\tilde{\sigma}_j & \text{if } \tilde{\sigma}_j = \frac{1}{\delta T} & (ii), \\[0.5pt]
		0 & \text{if } \tilde{\sigma}_j > \frac{1}{\delta T} & (iii).
	\end{array}
	\right.
\end{equation}

From \eqref{eq:mnig_conditions}, an appropriate choice of the parameters $ \{ \tilde{\sigma}_j\}_{j = 1}^d$ satisfies either the condition in (ii) or (iii) (i.e., $\tilde{\sigma}_j = \overline{\sigma}_j + \epsilon_j$, where $\overline{\sigma}_j = \frac{1}{\delta T}$ and $\epsilon \geq 0$). Despite larger values of $\epsilon_j$ resulting in a faster decay of the integrand to zero, they concurrently amplify the magnitude of the mixed first partial derivatives. Consequently, relatively large values of $\tilde{\sigma}_j$ may degrade the performance of RQMC in high dimensions.

\section{Proof of Proposition  \ref{prop:nig_proposition}}
\label{sec:nig_dom_transf_proof}
Let $\boldsymbol{Y} \sim \boldsymbol{ML}_d(\boldsymbol{0}, \boldsymbol{\tilde{\Sigma}})$ denote a $d$-dimensional multivariate Laplace distribution with zero mean and covariance matrix $\boldsymbol{\tilde{\Sigma}}$. The random vector $\boldsymbol{Y}$ can be expressed in its normal variance-mean form as $\boldsymbol{Y} \stackrel{d}{=} \sqrt{W}\boldsymbol{N}$, where $\boldsymbol{N} \sim \mathcal{N}_d(\boldsymbol{0}, \boldsymbol{\tilde{\Sigma}})$ is a $d$-dimensional multivariate normal distribution with zero mean and covariance matrix $\boldsymbol{\tilde{\Sigma}}$, and $W \sim \text{Exp}(1)$ follows a one-dimensional exponential distribution with rate 1. Furthermore, $\boldsymbol{N}$ can be written as $\boldsymbol{N} \stackrel{d}{=} \boldsymbol{\tilde{L}}\boldsymbol{Z}$, where $\boldsymbol{Z} \sim \mathcal{N}_d(\boldsymbol{0}, \boldsymbol{I}_d)$ follows $d$-dimensional standard normal distribution, and $\boldsymbol{\tilde{L}}$ is the square root matrix of $\boldsymbol{\tilde{\Sigma}}$.

Consequently, the PDF of $\boldsymbol{Y}$ can be expressed in terms of the PDFs of $\boldsymbol{N}$ and $W$ as follows:
\begin{equation}
	\label{eq:laplace_product_density}
	\psi_{\boldsymbol{Y}}(\boldsymbol{y}) = \int_{0}^{+\infty} w^{-d/2} \psi_W(w) \psi_{\boldsymbol{N}}\left( \frac{\boldsymbol{y}}{\sqrt{w}} \right) \mathrm{d}w, \quad \boldsymbol{y} \in \mathbb{R}^d.
\end{equation}

Using this PDF expression in \eqref{eq:laplace_product_density}, we can rewrite the integrand in \eqref{QOI} as follows.

\begin{equation}
	\begin{aligned}
		\int_{\mathbb{R}^d} g(\boldsymbol{y}) \mathrm{d} \boldsymbol{y} &=  \int_{\mathbb{R}^d} \frac{g(\boldsymbol{y})}{\psi_{\boldsymbol{Y}}(\boldsymbol{y})} \psi_{\boldsymbol{Y}}(\boldsymbol{y})\mathrm{d} \boldsymbol{y} \\
		&= \int_{\mathbb{R}^d} \frac{g(\boldsymbol{y})}{\psi_{\boldsymbol{Y}}(\boldsymbol{y})} 
		\left( \int_{0}^{+\infty} w^{-d/2} \psi_W(w) \psi_{\boldsymbol{N}}\left( \frac{\boldsymbol{y}}{\sqrt{w}} \right)  \mathrm{d}w \right) \mathrm{d} \boldsymbol{y} \\
		&= \int_{\mathbb{R}^d}  \int_{0}^{+\infty}   \psi_W(w)\frac{g( \sqrt{w}\boldsymbol{y})}{\psi_{\boldsymbol{Y}}  (\sqrt{w}\boldsymbol{y})} \psi_{\boldsymbol{N}}( \boldsymbol{y})  \mathrm{d}w \mathrm{d} \boldsymbol{y} \\
		&= \int_{\mathbb{R}^d}  \int_{0}^{+\infty}   \psi_W(w)\frac{g( \sqrt{w}	\boldsymbol{\tilde{L}}\boldsymbol{y})}{\psi_{\boldsymbol{Y}}  (\sqrt{w}	\boldsymbol{\tilde{L}} \boldsymbol{y})} \psi_{\boldsymbol{Z}}( \boldsymbol{y})  \mathrm{d}w \mathrm{d} \boldsymbol{y} \\
		&= \int_{[0,1]^d}  \int_{0}^{+\infty} \psi_W(w)\frac{g\left( \sqrt{w}	\boldsymbol{\tilde{L}} \Psi_{\boldsymbol{Z}}^{-1}(\boldsymbol{u})\right)}{\psi_{\boldsymbol{Y}}  \left( \sqrt{w} 	\boldsymbol{\tilde{L}}\Psi_{\boldsymbol{Z}}^{-1}(\boldsymbol{u})\right)}  \mathrm{d}w \mathrm{d} \boldsymbol{u} \\
		&= \int_{[0,1]^d}  \int_{[0,1]}\frac{g\left( \sqrt{\Psi_W^{-1}(u^{\prime})} \boldsymbol{\tilde{L}} \Psi_{\boldsymbol{Z}}^{-1}(\boldsymbol{u})\right)}{\psi_{\boldsymbol{Y}}  \left( \sqrt{\Psi_W^{-1}(u^{\prime})} \boldsymbol{\tilde{L}}\Psi_{\boldsymbol{Z}}^{-1}(\boldsymbol{u})\right)}  \mathrm{d} u^{\prime}\mathrm{d} \boldsymbol{u} \\
		&= \int_{[0,1]^{d+1}} 
		\frac{
			g\left( 
			\sqrt{\Psi_W^{-1}(u_{d+1})}
			\boldsymbol{\tilde{L}}	\Psi_{\boldsymbol{Z}}^{-1}
			(\boldsymbol{u}_{1:d})
			\right)
		}{
			\psi_{\boldsymbol{Y}}\left( 
			\sqrt{\Psi_W^{-1}(u_{d+1})}
			\boldsymbol{\tilde{L}}		\Psi_{\boldsymbol{Z}}^{-1}
			(\boldsymbol{u}_{1:d})
			\right)
		} 
		\mathrm{d}\boldsymbol{u},
	\end{aligned}
\end{equation}
where in the second line we plugged in the expression of the PDF in \eqref{eq:laplace_product_density}, in the third line we used the change of variable $\boldsymbol{y}^{\prime} = \frac{\boldsymbol{y}}{\sqrt{w}}$, in the fourth line we applied the change of variable $\boldsymbol{y}^{\prime}  = \boldsymbol{\tilde{L}}^{-1} \boldsymbol{y}$, in the fifth line we  applied the domain transformation mapping $\boldsymbol{y} = \Psi_{\boldsymbol{Z}}^{-1}(\boldsymbol{u})$ followed by the mapping $w = \Psi_W^{-1}(u^{\prime})$. In the last line, due to the non-nested nature of the integrals, we merged them into a single $(d+1)$-dimensional integral. 

\section[Heavy-tailed characteristic functions: example of independent assets]{Domain transformation for heavy-tailed characteristic functions: example of the VG model for independent assets}
\label{sec:domain_transf_vg_independent}

\paragraph{Product-form domain transformation}
This section follows the same steps as in Section~\ref{sec:gbm_dom_transf} to obtain an appropriate domain transformation for the heavy-tailed characteristic functions, using the example of the VG model, for which we have
$$
\phi^{VG}_{\boldsymbol{X}_T}(\boldsymbol{z}) = \prod_{j = 1}^d \phi^{VG}_{X_T}(z_j), \boldsymbol{z} \in \mathbb{C}^d, \quad \Im[ \boldsymbol{z} ]\in \delta^{VG}_X,
$$
where
$$
\phi^{VG}_{X_T}(z_j) = \left(1-\mathrm{i} \nu \theta_j z_j + \frac{\nu \sigma_j^2 }{2}z_j^2 \right)^{-T / \nu}, \quad z_j \in \mathbb{C}, \Im[ z_j ]\in \delta^{VG}_X.
$$
As $\phi^{VG}_{X_T}(z_j)$ is a rational function, a natural choice of the density $\psi(\cdot)$ is also to be a rational function. A suitable candidate density with this form is the generalized Student's $t$-distribution, $\psi^{stu}(\cdot)$, which we can write in the product-form as  $\psi^{stu}(\boldsymbol{y}) = \prod_{j = 1}^d \psi^{stu}_j(y_j)$, with
$$
\label{eq:product_form_vg_density}
\psi^{stu}_j(y_j) :=  \frac{\Gamma\left(\frac{\tilde{\nu}+1}{2}\right)}{\sqrt{\tilde{\nu}\pi} \tilde{\sigma}_j\Gamma\left(\frac{\tilde{\nu}}{2}\right)}\left(1+\frac{y_j^2}{\tilde{\nu}\tilde{\sigma}_j^2}\right)^{-(\tilde{\nu}+1 ) / 2}, y_j \in \mathbb{R}, \tilde{\sigma}_j > 0, \tilde{\nu} > 0.
$$
Other related studies \cite{kuo2006randomly,ouyang2024achieving} have typically considered the standard Student’s $t$-distribution ($\tilde{\sigma} =1$) and not the generalized distribution. Section~\ref{sec: num_exp_results} reveals the advantage of including the scaling parameters $\{\tilde{\sigma}_j\}_{j =1}^d$ and their effects on the convergence of RQMC. 
After specifying the functional form of $\psi^{stu}_j(y_j)$, the aim is to determine an appropriate range for the parameters $\tilde{\nu}$ and $\{\tilde{\sigma}_j\}_ {j = 1}^d$. We first define the function $r^{VG}_{stu,j}(\cdot)$ as the ratio of the characteristic function of the variable $X^j_T$ and the proposed density $\psi^{stu}_j(\cdot)$:
$$
r^{VG}_{stu,j}(\Psi_{stu}^{-1}(u_j)) :=  \frac{ \phi^{VG}_{X_T}(\Psi_{stu}^{-1}(u_j)+  \mathrm{i} R_j)}{	\psi^{stu}_j(\Psi_{stu}^{-1}(u_j))}, \quad u_j \in [0,1],  \boldsymbol{R}  \in \delta^{VG}_V.
$$
To determine the appropriate parameters, we replace the characteristic function and the proposed density with their explicit expressions, and obtain the following: 
\begin{equation}
	\label{eq:T2_uniariateVG}
	\hspace{-1cm} 
	\begin{aligned}
		r_{stu,j}^{VG}(\Psi_{stu}^{-1}(u_j)) &:=   	C_{\tilde{\nu}}    \tilde{\sigma}_j \times \frac{\left[ (\Psi_{stu}^{-1}(u_j))^2\left( \frac{\nu \sigma_j^2}{2} + \mathrm{i} \frac{\nu \sigma_j^2 R_j - \nu \theta_j}{\Psi_{stu}^{-1}(u_j)} +  \frac{1 + \nu \theta_j R_j - \frac{\nu \sigma_j^2 R_j^2}{2}}{(\Psi_{stu}^{-1}(u_j))^2} \right) \right]^{-\frac{T}{\nu}}}{ \left[ 1+ \frac{(\Psi_{stu}^{-1}(u_j))^2}{\tilde{\nu} \tilde{\sigma}_j^2} \right]^{- \frac{ \tilde{\nu} + 1}{2} } },
	\end{aligned}
\end{equation}
where $C_{\tilde{\nu}} = \frac{  \sqrt{\tilde{\nu}} \sqrt{ \pi} \Gamma\left(\frac{\tilde{\nu}}{2}\right)}{\Gamma\left(\frac{\tilde{\nu}+1}{2}\right)}$. We are interested in the asymptotic behavior of $	r_{stu,j}^{VG}(u_j) $ as $u_j \to 0$ i.e., as  $\Psi^{-1}(u_j) \to -\infty$; thus, we approximate $	r_{stu,j}^{VG}(u_j) $ near $u_j \to 0$ as follows:

\begin{equation}
	\begin{aligned}
		r_{stu,j}^{VG}(\Psi_{stu}^{-1}(u_j))  &=  C_{\tilde{\nu}}  \tilde{\sigma}_j \times \frac{\left[ (\Psi_{stu}^{-1}(u_j))^2\left( \frac{\nu \sigma_j^2}{2} + \mathrm{i} \frac{\nu \sigma_j^2 R_j - \nu \theta_j}{\Psi_{stu}^{-1}(u_j)} +  \frac{1 + \nu \theta_j R_j - \frac{\nu \sigma_j^2 R_j^2}{2}}{(\Psi_{stu}^{-1}(u_j))^2} \right) \right]^{-\frac{T}{\nu}}}{ \left[ 1+ \frac{(\Psi_{stu}^{-1}(u_j))^2}{\tilde{\nu} \tilde{\sigma}_j^2} \right]^{- \frac{ \tilde{\nu} + 1}{2} } } \\
		& \stackrel{u_j \to 0}{\sim}   C_{\tilde{\nu}} \tilde{\sigma}_j \frac{\left[\frac{\nu \sigma_j^2}{2}\left(\Psi_{stu}^{-1}(u_j)\right)^2\right]^{-\frac{T}{\nu}}}{\left[\frac{\left(\Psi_{stu}^{-1}(u_j)\right)^2}{\tilde{\nu}  \tilde{\sigma}_j ^2}\right]^{-\frac{\tilde{\nu} +1}{2}}} := h_{stu}^{VG}(u_j)
	\end{aligned}
\end{equation}
Given the expression of $h_{stu}^{VG}(\Psi_{stu}^{-1}(u_j))$, we  enumerate three possible limits 
\begin{equation}
	\label{eq:vg_conditions}
	\lim_{u_j \to 0} h_{stu}^{VG}(\Psi_{stu}^{-1}(u_j)) =
	\left\{
	\begin{array}{lll}
		+\infty & \text{if } \tilde{\nu} > \frac{2T}{\nu} - 1 & (i), \\[0.5pt]
		C_{\tilde{\nu}} \tilde{\sigma}_j \left[\frac{\nu \sigma_j^2 \tilde{\nu} \tilde{\sigma}_j^2}{2}\right]^{- \frac{T}{\nu}}
		& \text{if } \tilde{\nu} = \frac{2T}{\nu} - 1 & (ii), \\[0.5pt]
		0 & \text{if } \tilde{\nu} < \frac{2T}{\nu} - 1 & (iii).
	\end{array}
	\right.
\end{equation}

From \eqref{eq:vg_conditions}, an appropriate choice of the parameter $\tilde{\nu}$ satisfies either the condition in Case (ii) or (iii), where (ii) implicitly relies on the constraint $ \frac{2T}{ \nu } -1 > 0$, which, if not satisfied, implies that the characteristic function is not integrable \cite{kirkby2015efficient}, and which violates our Assumption \ref{ass:Assumptions on  the distribution}. After specifying the value of the parameter $\tilde{\nu}$, a candidate choice for $\{\tilde{\sigma}_j\}_ {j = 1}^d$ is 
\begin{equation}
	\label{eq:vg_sigma_choice}
	\tilde{\sigma}_j = \left[  \frac{\nu \sigma^2_j \tilde{\nu}}{2}  \right]^{ \frac{T}{\nu - 2T}}   (C_{\tilde{\nu}})^{-\frac{\nu}{\nu - 2T}}, \; \forall \; j \in \mathbb{I}_d.
\end{equation}
When $\tilde{\nu} = \frac{2T}{\nu}-1$, this choice in \eqref{eq:vg_sigma_choice} implies that $\lim_{u_j \to 0} h_{stu}^{VG}(\Psi_{stu}^{-1}(u_j)) = 1$. Consequently, this choice reduces the adverse effect that large values of $\tilde{\sigma}_j$ could have on the magnitude of mixed partial derivatives of the integrand, deteriorating the efficiency of RQMC. In summary, a suitable choice for $\tilde{\nu},\{\tilde{\sigma}_j\}_ {j = 1}^d$ is $\tilde{\nu} = \overline{\nu} - \epsilon$, where $\overline{\nu} = \frac{2T}{ \nu} - 1 $ defines the critical value and $0 \leq \epsilon < \overline{\nu}$, setting $\tilde{\sigma}_j = \left[  \frac{\nu \sigma^2_j \tilde{\nu}}{2}  \right]^{ \frac{T}{\nu - 2T}}   (C_{\tilde{\nu}})^{-\frac{\nu}{\nu - 2T}}$.


\section{Proof of Proposition \ref{prop:vg_proposition} }
\label{sec:vg_dom_transf_proof}
Let $\boldsymbol{Y} \sim t_d(\boldsymbol{0}, \boldsymbol{\tilde{\Sigma}}, \tilde{\nu})$ denote a $d$-dimensional Student-$t$ distribution with zero mean, covariance matrix $\boldsymbol{\tilde{\Sigma}}$, and degrees of freedom $\tilde{\nu}$. The random vector $\boldsymbol{Y}$ can be expressed in its normal variance-mean form as $\boldsymbol{Y} \stackrel{d}{=} \frac{\sqrt{\tilde{\nu}}}{\sqrt{W}}\boldsymbol{N}$, where $\boldsymbol{N} \sim \mathcal{N}_d(\boldsymbol{0}, \boldsymbol{\tilde{\Sigma}})$ is a $d$-dimensional multivariate normal distribution with zero mean and covariance matrix $\boldsymbol{\tilde{\Sigma}}$, and $W \sim \chi^2(\tilde{\nu})$ follows a one-dimensional chi-squared distribution with $\tilde{\nu}$ degrees of freedom. To simplify notation, we define $\boldsymbol{\tilde{N}} \sim \mathcal{N}_d(\boldsymbol{0}, \tilde{\nu}\boldsymbol{\tilde{\Sigma}})$. Furthermore, $\boldsymbol{\tilde{N}}$ can be written as $\boldsymbol{\tilde{N}}\stackrel{d}{=} \boldsymbol{\tilde{L}}\boldsymbol{Z}$, where $\boldsymbol{Z} \sim \mathcal{N}_d(\boldsymbol{0}, \boldsymbol{I}_d)$ follows a  $d$-dimensional standard normal distribution, and $\boldsymbol{\tilde{L}}$ is the square root matrix of $ \tilde{\nu} \boldsymbol{\tilde{\Sigma}}$. 

Consequently, the PDF of $\boldsymbol{Y}$ can be expressed in terms of the PDFs of $W$ and $\boldsymbol{\tilde{N}}$ as follows:
\begin{equation}
	\label{eq:student_product_density}
	\psi_{\boldsymbol{Y}}(\boldsymbol{y}) = \int_{0}^{+\infty} w^{d/2} \psi_W(w) \psi_{\boldsymbol{\tilde{N}}}( \sqrt{w}\boldsymbol{y} )  \mathrm{d}w, \quad \boldsymbol{y} \in \mathbb{R}^d.
\end{equation}

Using this PDF expression in \eqref{eq:student_product_density}, we can rewrite the integrand in \eqref{QOI} as follows.
\begin{equation}
	\begin{aligned}
		\int_{\mathbb{R}^d} g(\boldsymbol{y}) \mathrm{d} \boldsymbol{y} &=  
		\int_{\mathbb{R}^d} \frac{g(\boldsymbol{y})}{\psi_{\boldsymbol{Y}}(\boldsymbol{y})} 
		\psi_{\boldsymbol{Y}}(\boldsymbol{y})\mathrm{d} \boldsymbol{y} \\
		&= \int_{\mathbb{R}^d} \frac{g(\boldsymbol{y})}{\psi_{\boldsymbol{Y}}(\boldsymbol{y})} 
		\left( \int_{0}^{+\infty} w^{d/2} \psi_W(w) 
		\psi_{\boldsymbol{\tilde{N}}}\left( \sqrt{w}\boldsymbol{y} \right)  
		\mathrm{d}w \right) \mathrm{d} \boldsymbol{y} \\
		&= \int_{\mathbb{R}^d}  \int_{0}^{+\infty}   
		\psi_W(w)\frac{g( \frac{\boldsymbol{y}}{\sqrt{w}})}
		{\psi_{\boldsymbol{Y}}  (\frac{\boldsymbol{y}}{\sqrt{w}})} 
		\psi_{\boldsymbol{\tilde{N}}}( \boldsymbol{y})  
		\mathrm{d}w \mathrm{d} \boldsymbol{y} \\
		&= \int_{\mathbb{R}^d}  \int_{0}^{+\infty}   
		\psi_W(w)\frac{g( \frac{\boldsymbol{\tilde{L}} \boldsymbol{y}}{\sqrt{w}})}
		{\psi_{\boldsymbol{Y}}  (\frac{\boldsymbol{\tilde{L}} \boldsymbol{y}}{\sqrt{w}})} 
		\psi_{\boldsymbol{Z}}( \boldsymbol{y})  
		\mathrm{d}w \mathrm{d} \boldsymbol{y} \\
		&= \int_{[0,1]^d}  \int_{0}^{+\infty} 
		\psi_W(w)\frac{g\left( \frac{\boldsymbol{\tilde{L}} \Psi_{\boldsymbol{Z}}^{-1}(\boldsymbol{u})}
			{\sqrt{w}}\right)}{\psi_{\boldsymbol{Y}}  
			\left( \frac{\boldsymbol{\tilde{L}} \Psi_{\boldsymbol{Z}}^{-1}(\boldsymbol{u})}{\sqrt{w}}\right)}  
		\mathrm{d}w \mathrm{d} \boldsymbol{u} \\
		&= \int_{[0,1]^d}  \int_{[0,1]}\frac{g\left( 
			\frac{\boldsymbol{\tilde{L}} \Psi_{\boldsymbol{Z}}^{-1}(\boldsymbol{u})}
			{\sqrt{\Psi_W^{-1}(u^{\prime})}}\right)}
		{\psi_{\boldsymbol{Y}}  \left( 
			\frac{\boldsymbol{\tilde{L}} \Psi_{\boldsymbol{Z}}^{-1}(\boldsymbol{u})}
			{\sqrt{\Psi_W^{-1}(u^{\prime})}} \right)}  
		\mathrm{d} u^{\prime}\mathrm{d} \boldsymbol{u} \\
		&= \int_{[0,1]^{d+1}} 
		\frac{
			g\left( 
			\frac{\boldsymbol{\tilde{L}}
				\Psi_{\boldsymbol{Z}}^{-1}
				(\boldsymbol{u}_{1:d})}{\sqrt{\Psi_W^{-1}(u_{d+1})}}
			\right)
		}{
			\psi_{\boldsymbol{Y}}\left( 
			\frac{\boldsymbol{\tilde{L}}
				\Psi_{\boldsymbol{Z}}^{-1}
				(\boldsymbol{u}_{1:d})}{\sqrt{\Psi_W^{-1}(u_{d+1})}}
			\right)
		} 
		\mathrm{d}\boldsymbol{u}
	\end{aligned}
\end{equation}
where in the second line we plugged in the expression of the PDF in  \eqref{eq:student_product_density}, in the third line we used the change of variable $\boldsymbol{y}^{\prime} = \sqrt{w}\boldsymbol{y}$, in the fourth line we applied the change of variable $\boldsymbol{y}^{\prime}  = \boldsymbol{\tilde{L}}^{-1} \boldsymbol{y}$, in the fifth line we  applied the domain transformation mapping $\boldsymbol{y} = \Psi_{\boldsymbol{Z}}^{-1}(\boldsymbol{u})$ followed by the mapping $w = \Psi_W^{-1}(u^{\prime})$. In the last line, due to the non-nested nature of the integrals, we merged them into a single $(d+1)$-dimensional integral.

\section{Fourier-RQMC Pricing Pipeline}
	\label{app:rqmc_fourier_pipeline}

{
Before applying Algorithm~\ref{alg:rqmc_fourier}, the transformation density $\psi_{\boldsymbol{Y}}$, its parameters, the transformation map $\mathcal{T}_{\psi}$, and the RQMC dimension $m$ are selected according to Tables \ref{tab:univariate_dom_tranf}, \ref{tab:multivariate_dom_tranf}, and~\ref{tab:identification}, together with the transformation formulas \eqref{eq:dom_transf_general}, \eqref{eq:transformed_fourier_pricing_integral}, \eqref{eq:nig_transformed_integral_rep}, and~\eqref{eq:vg_nested_integ}.
\begin{algorithm}[H]
	\caption{Fourier-RQMC pricing with domain transformation}
	\label{alg:rqmc_fourier}
	\begin{algorithmic}
		
		\Require $\boldsymbol{\Theta}_X$, $\boldsymbol{\Theta}_P$, $\Phi_{\boldsymbol{X}_T}$, $\widehat{P}$, $\delta_V$,  $\psi_{\boldsymbol{Y}}$, $\mathcal{T}_{\psi}$, $m$, $N_0$, $N_{\max}$, $S$, $\alpha$, $\varepsilon_{\mathrm{rel}}$.
		
		\State $C_\alpha \gets \Psi_{\mathcal{N}(0,1)}^{-1}(1-\alpha/2)$.
		
		\State $\boldsymbol{R}^*
		\gets
		\displaystyle
		\underset{\boldsymbol{R}\in\delta_V}{\arg\min}\;
		g\left(
		\boldsymbol{0};
		\boldsymbol{R},
		\boldsymbol{\Theta}_X,
		\boldsymbol{\Theta}_P
		\right)$, using \eqref{eq:optimal_damping_rule}.

	\State $\widetilde{g}(\cdot;\boldsymbol{R}^*)
	\gets
	\displaystyle
	\frac{
		g\left(
		\mathcal{T}_{\psi}(\cdot);
		\boldsymbol{R}^*,
		\boldsymbol{\Theta}_X,
		\boldsymbol{\Theta}_P
		\right)
	}{
		\psi_{\boldsymbol{Y}}\left(
		\mathcal{T}_{\psi}(\cdot)
		\right)
	}$, using \eqref{eq:dom_transf_general}, \eqref{eq:transformed_fourier_pricing_integral}, \eqref{eq:nig_transformed_integral_rep}, or \eqref{eq:vg_nested_integ}.
		
		\State $N\gets N_0$.
		
		\Repeat
		\For{$s=1,\ldots,S$}
		\State Generate a Sobol point set randomized by digital shifting
		$
		\{\boldsymbol{u}_n^{(s)}\}_{n=1}^{N}
		\subset [0,1]^m
		$.
		
		\State $V_N^{(s)}
		\gets
		\displaystyle
		\frac{1}{N}
		\sum_{n=1}^{N}
		\widetilde{g}
		\left(
		\boldsymbol{u}_n^{(s)};
		\boldsymbol{R}^*
		\right)$.
		\EndFor
		
		\State $\widehat{V}
		\gets
		\displaystyle
		\frac{1}{S}
		\sum_{s=1}^{S}
		V_N^{(s)}$.
		
		\State $\widehat{e}
		\gets
		\displaystyle
		C_\alpha
		\left(
		\frac{1}{S(S-1)}
		\sum_{s=1}^{S}
		\left(
		V_N^{(s)}-\widehat{V}
		\right)^2
		\right)^{1/2}$, using \eqref{rqmc_error}.
		
		\If{$\widehat{e}/|\widehat{V}|>\varepsilon_{\mathrm{rel}}$}
		\State $N\gets 2N$.
		\EndIf
		
		\Until{$\widehat{e}/|\widehat{V}|\leq\varepsilon_{\mathrm{rel}}$ or $N>N_{\max}$}
		
		\State \Return $\widehat{V}$, $\widehat{e}$.
		
	\end{algorithmic}
\end{algorithm}
}

{
	\noindent The implementation of Algorithm~\ref{alg:rqmc_fourier} is publicly available at
	\href{https://github.com/Michael-Samet/Quasi-Monte-Carlo-for-Efficient-Fourier-Pricing-of-Multi-Asset-Options/tree/main/Premia\%20Implementation}{https://github.com/Michael-Samet/Quasi-Monte-Carlo-for-Efficient-Fourier-Pricing}.
}
\section[Comparison with the Monte Carlo Method in the Fourier Space]{Comparison of the Proposed Approach with the Monte Carlo Method in the Fourier Space}
\label{sec:qmc_vs_mc_fourier}
This section demonstrates the advantage of employing QMC in the Fourier space over applying the MC method in the Fourier space, as in \cite{ballotta2022powering}. Although the MC method does not require a domain transformation, we must still introduce a density from which to sample. Hence, for the sake of comparison between the two methods, the MC estimator in the Fourier space will be defined as 
	\begin{equation}
		Q^{MC}_{d,n}[g] : = \frac{1}{N}\sum_{n = 1}^N \tilde{g}\left( \boldsymbol{u}_n\right), \quad \boldsymbol{u}_n \in [0,1]^d
\end{equation}
where  $\{\boldsymbol{u}_n \}_{n = 1}^N$ are independent and identically distributed (i.i.d.) samples drawn from the uniform distribution $\mathcal{U}([0,1]^d)$, and the transformed integrand $\tilde{g}(\cdot)$ is defined in the same was as in Section \ref{sec:RQMC_fourier_pricing}. The following numerical experiments demonstrate that applying RQMC with the appropriate domain transformation allows retaining nearly optimal convergence rates, in contrast to the MC method, in which the convergence rate is insensitive to the regularity of the integrand. Figures~\ref{fig_fourmc_vs_fourqmc_gbm_basket}, \ref{fig_fourmc_vs_fourqmc_vg_basket}, and \ref{fig_fourmc_vs_fourqmc_nig_basket} illustrate that employing RQMC in the Fourier space achieves a relative statistical error of about one order of magnitude lower than that of MC in the Fourier space. Moreover, the convergence rates of RQMC for the pricing of 4D basket put options range between $\mathcal{O}(N^{-1})$ and $\mathcal{O}(N^{-1.3})$, which is double the rate $\mathcal{O}(N^{-0.5})$ of the method. These convergence rates indicate that the RQMC method takes advantage of the analyticity of the integrand in the Fourier domain.
\FloatBarrier
\begin{figure}[h!]
	\centering	
	\begin{subfigure}{0.45\textwidth}	\includegraphics[width=\linewidth]{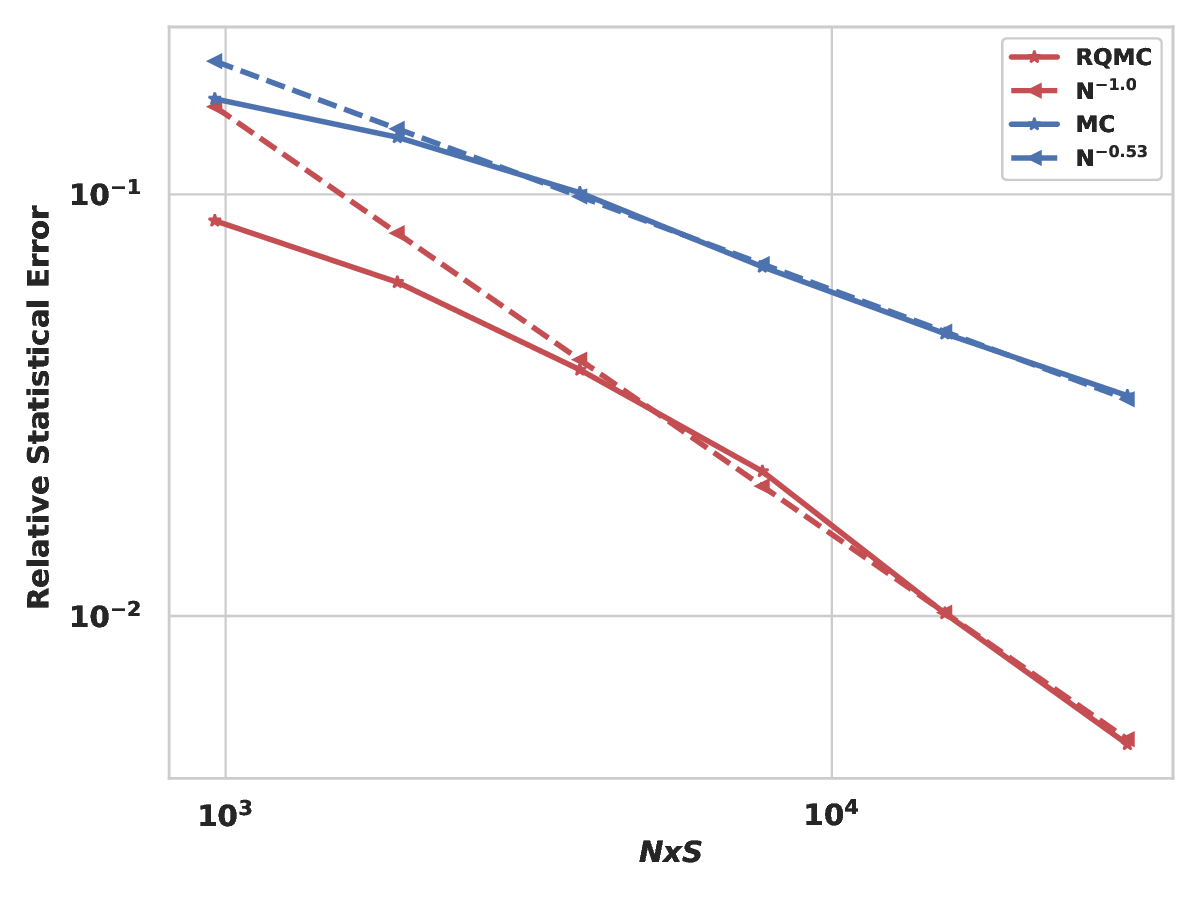}	\caption{}	\label{fig_fourmc_vs_fourqmc_gbm_basket}	\end{subfigure} 
	\begin{subfigure}{0.45\textwidth}
		\includegraphics[width=\linewidth]{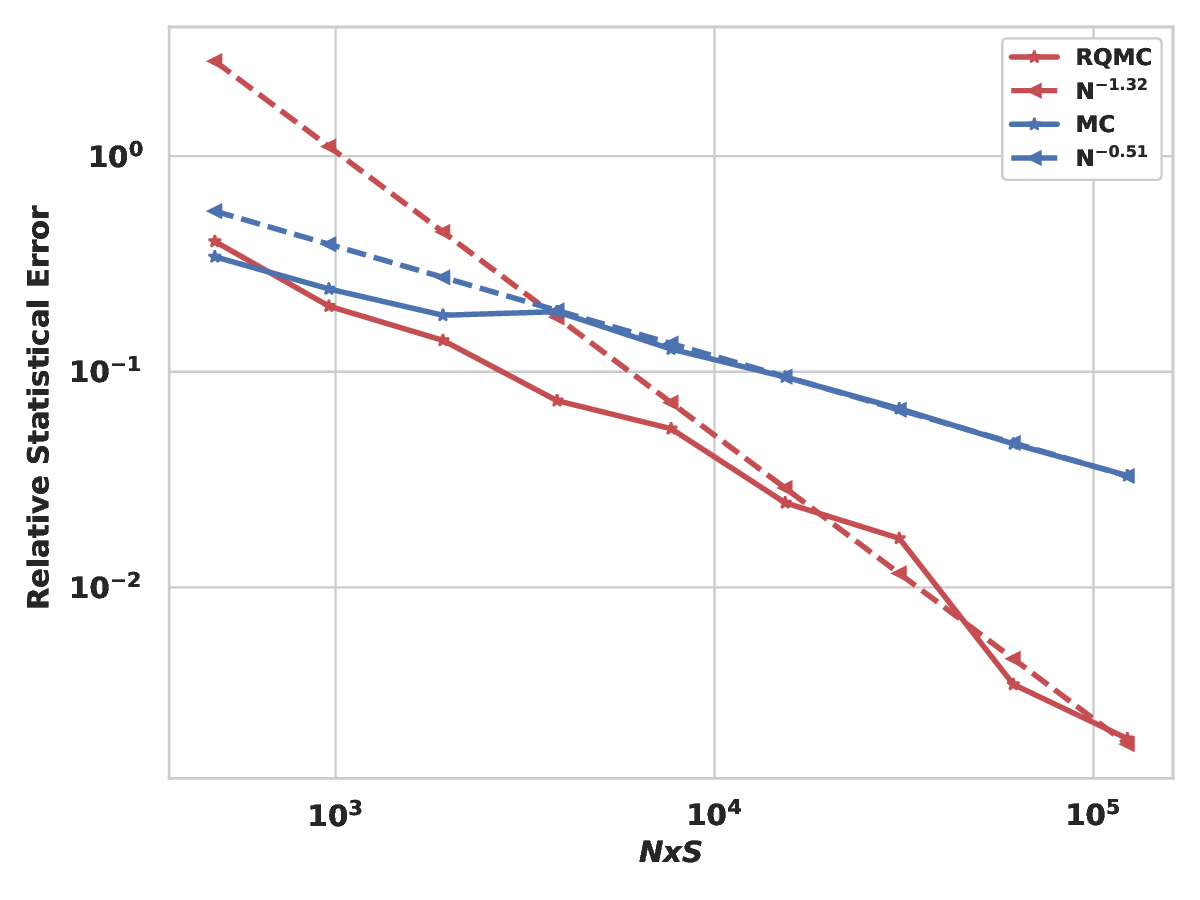}
		\caption{}
		\label{fig_fourmc_vs_fourqmc_vg_basket}
	\end{subfigure} \\
	\begin{subfigure}{0.45\textwidth}
		\includegraphics[width=\linewidth]{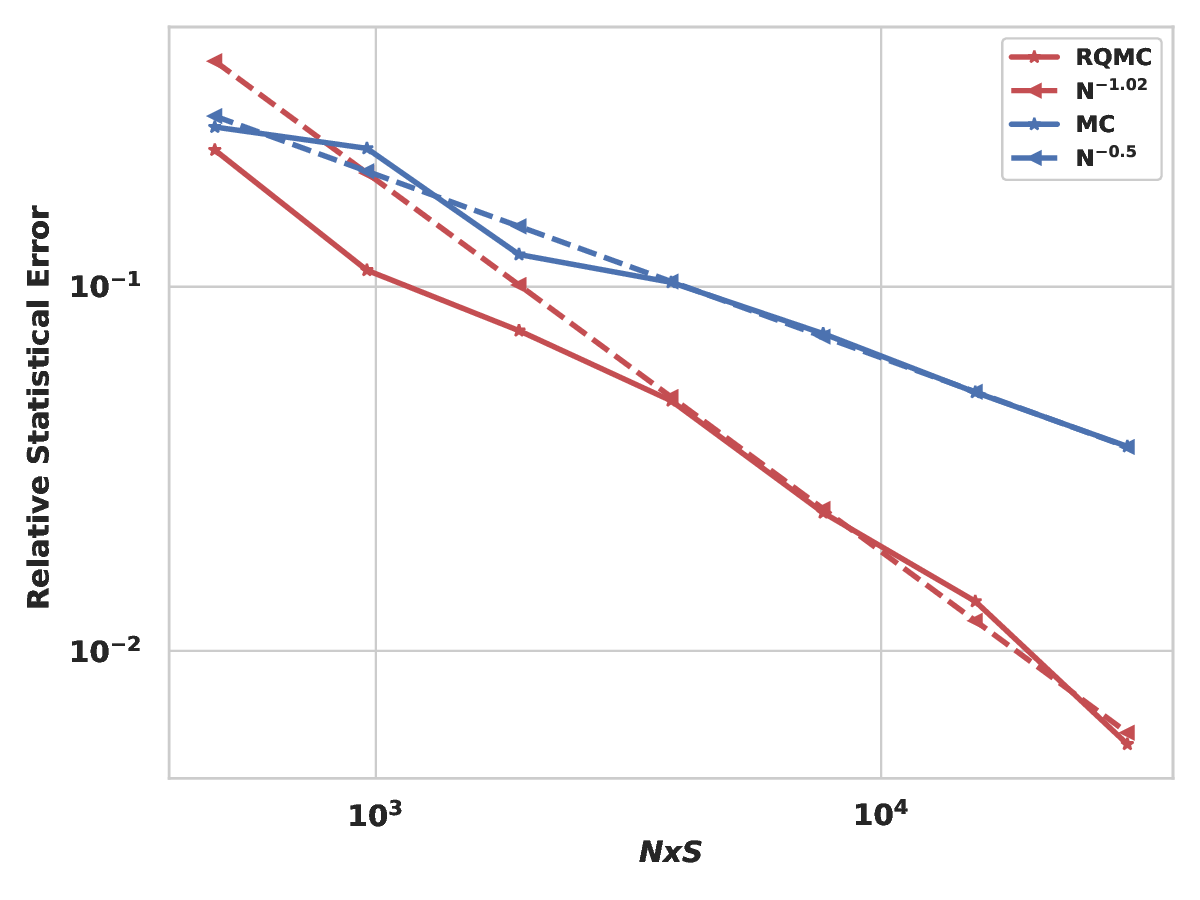}
		\caption{}
		\label{fig_fourmc_vs_fourqmc_nig_basket}
	\end{subfigure}
	\caption[]{4D basket put: convergence of the relative statistical error with respect to $N \times S$ function evaluations. For RQMC, $N$: number of QMC points; $S = 30$: number of independent and identically distributed uniform shifts. For MC, $N\times S$: number of random samples in the approximation. $S_0^j = 100$, $K = 100$, $r = 0$, and $T = 1$. (a) GBM model: $\sigma_j = 0.2$, $\boldsymbol{\Sigma} = \boldsymbol{I_d}$, $\overline{\mathbf{R}}_j = 4.44$, $\tilde{\sigma}_j = \frac{1}{\sqrt{T} \sigma} = 5$. (b) VG model: $\sigma_j = 0.4$, $\theta_j = -0.3$, $\nu = 0.2$, $\boldsymbol{\Sigma} = \boldsymbol{I_d}$, $\overline{\mathbf{R}}_j = 1.31$, $\tilde{\nu} = \frac{2T}{\nu}-1 = 9$, $\tilde{\sigma}_j = \left[ \frac{\nu \sigma^2_j \tilde{\nu}}{2} \right]^{ \frac{T}{\nu - 2T}} (\tilde{\nu})^{\frac{\nu}{4T - 2 \nu}}= 3.31$. (c) NIG model: $ \alpha = 20$, $\beta_j = -3$, $\delta = 0.2$, $\boldsymbol{\Delta} = \boldsymbol{I_d}$, $\overline{\mathbf{R}}_j = 5.73$, $\forall j = 1,\ldots,4$, and the domain transformation parameters are $ \tilde{\sigma}_j = \frac{1}{T \delta} = 10 $.  } 
	\label{fig:cpu_qmc_quad}
\end{figure}

\FloatBarrier

\end{document}